\documentclass[aps,english,floatfix,superscriptaddress,tightenlines,onecolumn,nofootinbib,11pt]{revtex4-2}

\usepackage{comment}
\usepackage{amsmath} 
\usepackage{mathtools}
\usepackage{amssymb}

\usepackage{amsthm}
\usepackage{thmtools}
\usepackage[normalem]{ulem}

\theoremstyle{plain}
\newtheorem{theorem}{Theorem}
\newtheorem{assumption}{Assumption}

\theoremstyle{definition}

\theoremstyle{plain}

\newtheorem{hypothesis}[theorem]{Hypothesis}

\usepackage{amssymb}
\usepackage{amsmath,bm}
\usepackage{amssymb}
\usepackage{graphicx}
\usepackage{amsfonts}         
\usepackage{fancybox}

\usepackage[usenames,dvipsnames]{xcolor}

\definecolor{Mathematica1}{rgb}{0.368417, 0.506779, 0.709798}
\definecolor{Mathematica2}{rgb}{0.880722, 0.611041, 0.142051}
\definecolor{Mathematica3}{rgb}{0.560181, 0.691569, 0.194885}

\usepackage[pagebackref,  
bookmarks={false}, pdfauthor={Entanglement Bootstrap Collaboration}, pdftitle={Yay, physics!}]{hyperref}
\hypersetup{colorlinks=true, 
linkcolor=Mathematica1,
citecolor=Mathematica3,
filecolor=OliveGreen, 
urlcolor=Mathematica2,
filebordercolor={.8 .8 1}, 
urlbordercolor={.8 .8 0}}
\usepackage{soul}
\setstcolor{Red}

\definecolor{darkgreen}{rgb}{0,0.4,0}
\definecolor{darkred}{rgb}{0.4,0,0}
\definecolor{darkblue}{rgb}{0,0,0.4}
\definecolor{lightblue}{rgb}{.6,.6,0.9}

\definecolor{uglybrown}{rgb}{0.8,  0.7,  0.5}

\definecolor{palatinatepurple}{rgb}{0.41, 0.16, 0.38}
\definecolor{celebrationcolor}{rgb}{0.75,  0.0,  0.9}

\definecolor{shadecolor}{rgb}{0.90,0.90,0.90}
\definecolor{DVcolor}{rgb}{0.95,  0.5,  0.2}
\definecolor{lightbluemuons}{rgb}{0.0,.65,1.0}

\definecolor{chartreuse}{rgb}{0.70, 1.00, 0.00}

\usepackage{tikz}
\usetikzlibrary{shapes}
\usetikzlibrary{trees}
\usetikzlibrary{matrix,arrows} 				
\usetikzlibrary{positioning}				
\usetikzlibrary{calc,through}				
\usetikzlibrary{decorations.pathreplacing}  
\usepackage[tikz]{bclogo} 					
\usepackage{pgffor}							

\usetikzlibrary{decorations.markings}
\usetikzlibrary{intersections}				

\usetikzlibrary{arrows,decorations.pathmorphing,backgrounds,positioning,fit,petri,automata,shadows,calendar,mindmap, graphs}
\usetikzlibrary{arrows.meta,bending}

\tikzset{
    vector/.style={decorate, decoration={snake}, draw},
    fermion/.style={postaction={decorate},
        decoration={markings,mark=at position .55 with {\arrow{>}}}},
    fermionbar/.style={draw, postaction={decorate},
        decoration={markings,mark=at position .55 with {\arrow{<}}}},
    fermionnoarrow/.style={},
    gluon/.style={decorate,
        decoration={coil,amplitude=4pt, segment length=5pt}},
    scalar/.style={dashed, postaction={decorate},
        decoration={markings,mark=at position .55 with {\arrow{>}}}},
    scalarbar/.style={dashed, postaction={decorate},
        decoration={markings,mark=at position .55 with {\arrow{<}}}},
    scalarnoarrow/.style={dashed,draw},
	vectorscalar/.style={loosely dotted,draw=black, postaction={decorate}},
}

\def\centerarc[#1](#2)(#3:#4:#5)
    { \draw[#1] ($(#2)+({#5*cos(#3)},{#5*sin(#3)})$) arc (#3:#4:#5); }

\usepackage{enumitem}

\usepackage{slashed}

\usepackage[font=small,labelfont=bf]{caption}
\DeclareCaptionFont{tiny}{\tiny}
\usepackage{epstopdf}
\usepackage{wasysym}

\usepackage[yyyymmdd,hhmmss]{datetime}   
\usepackage{framed}  
 \usepackage{mdframed}
 \usepackage{wrapfig}
 \usepackage{yfonts}  

\def\ketbra#1#2{ | #1 \rangle\hskip-2pt\langle #2|}

\newmdenv[%
        backgroundcolor=lightgray,
    linecolor=black,
    outerlinewidth=2pt,
]{boxedandshaded}

\numberwithin{equation}{section}

\renewcommand{\theequation}{\arabic{section}.\arabic{equation}}

\newlength{\extraspace}
\newlength{\extraspaces}
\def\be{\begin{equation}}
\def\ee{\end{equation}}

\newcommand{\bea}{\begin{eqnarray}}
\newcommand{\eea}{\end{eqnarray}}

\def\p{\partial}

\def\eps{\epsilon}

\def\Tr{{{\rm Tr}}}
\def\tr{{\rm tr}}

\def\bra#1{\left\langle#1\right|}
\def\ket#1{\left|#1\right\rangle}

\def\CD{{\cal D}}

\def\CH{{\cal H}}

\def\CM{{\cal M}}

\def\CO{{\cal O}}

\def\CV{{\cal V}}
\def\CW{{\cal W}}

\def\II{\relax{I\kern-.10em I}}

\def\IB{\relax{\rm I\kern-.18em B}}

\def\ID{\relax{\rm I\kern-.18em D}}
\def\IE{\relax{\rm I\kern-.18em E}}
\def\IF{\relax{\rm I\kern-.18em F}}
\def\IG{\relax\hbox{$\inbar\kern-.3em{\rm G}$}}
\def\IGa{\relax\hbox{${\rm I}\kern-.18em\Gamma$}}
\def\IH{\relax{\rm I\kern-.18em H}}
\def\II{\relax{\rm I\kern-.18em I}}
\def\IK{\relax{\rm I\kern-.18em K}}

\def\inbar{\,\vrule height1.5ex width.4pt depth0pt}

\def\p{\partial}

\def\simgt{\hskip0.05in\relax{ 
\raise3.0pt\hbox{ $>$
{\lower5.0pt\hbox{\kern-1.05em $\sim$}} }} \hskip0.05in}

\def\im{{\rm im\ }}

\def\lp10{\ell_p^{10}}
\def\lp11{\ell_p^{11}}
\def\R11{R_{11}}

\def\frac#1#2{{#1 \over #2}}

\newdimen\tableauside\tableauside=1.0ex
\newdimen\tableaurule\tableaurule=0.4pt
\newdimen\tableaustep
\def\phantomhrule#1{\hbox{\vbox to0pt{\hrule height\tableaurule width#1\vss}}}
\def\phantomvrule#1{\vbox{\hbox to0pt{\vrule width\tableaurule height#1\hss}}}
\def\sqr{\vbox{%
  \phantomhrule\tableaustep
  \hbox{\phantomvrule\tableaustep\kern\tableaustep\phantomvrule\tableaustep}%
  \hbox{\vbox{\phantomhrule\tableauside}\kern-\tableaurule}}}
\def\squares#1{\hbox{\count0=#1\noindent\loop\sqr
  \advance\count0 by-1 \ifnum\count0>0\repeat}}
\def\tableau#1{\vcenter{\offinterlineskip
  \tableaustep=\tableauside\advance\tableaustep by-\tableaurule
  \kern\normallineskip\hbox
    {\kern\normallineskip\vbox
      {\gettableau#1 0 }%
     \kern\normallineskip\kern\tableaurule}%
  \kern\normallineskip\kern\tableaurule}}
\def\gettableau#1 {\ifnum#1=0\let\next=\null\else
  \squares{#1}\let\next=\gettableau\fi\next}

\def\({\left(}
\def\){\right)}

\def\ii{{\bf i}}

\def\aa{{\bf a}}

\def\AA{{\bf A}}

\def\lsim{\mathrel{\mathstrut\smash{\ooalign{\raise2.5pt\hbox{$<$}\cr\lower2.5pt\hbox{$\sim$}}}}}
\def\gsim{\mathrel{\mathstrut\smash{\ooalign{\raise2.5pt\hbox{$>$}\cr\lower2.5pt\hbox{$\sim$}}}}}

\def\overleftrightarrow#1{\vbox{\ialign{##\crcr
     $\leftrightarrow$\crcr\noalign{\kern-0pt\nointerlineskip}
     $\hfil\displaystyle{#1}\hfil$\crcr}}}
     
     \def\overleftarrow#1{\vbox{\ialign{##\crcr
     $\leftarrow$\crcr\noalign{\kern-0pt\nointerlineskip}
     $\hfil\displaystyle{#1}\hfil$\crcr}}}

\newif{\ifeq}           
\eqtrue                 
\newcounter{lecturecounter}

\usetikzlibrary{calc, intersections} 

\definecolor{XLgreen}{RGB}{34,139,34}
\definecolor{JMblue}{RGB}{25,25,125}
\definecolor{BSorange}{RGB}{140,50,0}

\usepackage[export]{adjustbox}
\usepackage{graphicx}
\usepackage{xspace}
\usepackage{bbm}
\usepackage{comment}

\def\({\left(} 
\def\){\right)}
\def\[{\left[} 
\def\]{\right]}

\def\fc{\mathfrak{c}_{\mathrm{tot}}}
\def\dune{\mathbb{D}}
\def\Aone{{\bf A1}\xspace}
\def\Azero{{\bf A0}\xspace}
\allowdisplaybreaks
\newcommand{\eqover}[1]{\buildrel{#1}\over{=}}

\def\smooth{\text{smooth}}
\def\corner{\text{corner}}

\def\rr{{\bf r}}
\def\Hrec{H_{\text{rec}}}
\def\rec{\text{rec}}
\def\tot{\text{tot}}
\def\euler{\text{Euler}}
\def\inte{\text{int}}
\def\CFT{\text{CFT}}

\allowbreak

\makeatletter
\def\p@subsection{}
\makeatother

\begin{document}
\title{\bfseries\Large 
Chiral gapped states are universally non-topological
}
\author{Xiang Li, Ting-Chun Lin, Yahya Alavirad,  
John McGreevy\\[1mm]
\it\small  Department of Physics, University of California San Diego, La Jolla, CA 92093, USA}

\begin{abstract}
We propose an operator generalization of the Li-Haldane conjecture regarding the entanglement Hamiltonian of a disk in a 2+1D chiral gapped groundstate.
The logic applies to regions with sharp corners, from which we derive several universal properties regarding corner entanglement. 
These universal properties follow from a set of locally-checkable conditions on the wavefunction. 
We also define a quantity $(\mathfrak{c}_{\text{tot}})_{\text{min}}$ that reflects the robustness of corner entanglement contributions, and show that it 
provides an obstruction to a gapped boundary.
One reward from our analysis is that we can construct a local gapped Hamiltonian within the same chiral gapped phase from a given wavefunction; we conjecture that it is closer to the low-energy renormalization group fixed point than the original parent Hamiltonian. 
Our analysis of corner entanglement reveals the emergence of a universal conformal geometry encoded in the entanglement structure of bulk regions of chiral gapped states that is not visible in topological field theory. Our formalism also gives an explanation of the modular commutator formula for the chiral central charge. 
\end{abstract}
\maketitle
\tableofcontents
\date{\today}

\section{Introduction}
{\bf Background and motivations.}
In the realm of 2+1D gapped states, there are states that, if put on a disk, have gapless degrees of freedom along the boundary that cannot be gapped by any local perturbation. In the scope of this paper, we regard this feature of ungappable edge for a gapped state as the defining property of \emph{chiral gapped states}, as such states generically break time-reversal symmetry\footnote{This convenient nomenclature is not perfect.  In fact there are states invariant under some action of time-reversal symmetry that do not admit gapped boundary, such as the T-Pfaffian described at the end of \cite{Barkeshli:2016mew}.  Thanks to Maissam Barkeshli for bringing this example to our attention.  We note that our nomenclature also fails in the other direction, i.e.~there are gapped states with gapped boundaries that are not time-reversal invariant either for universal \cite{Ellison:2026hvk} or non-universal reasons.}. 
Chiral gapped states possess many other interesting features, such as quantized electric or thermal Hall conductance, no commuting-projector parent Hamiltonian, and no zero-correlation-length renormalization group (RG) representative in a tensor product Hilbert space with finite local dimension. 

Chiral gapped states also have a close relation to conformal field theories (CFTs), known as the bulk/edge correspondence. Here we enumerate some of the connections: First, the gapless boundary is described by a CFT that is anomalous. The anomaly is manifest as modular non-invariance and can be detected by chiral central charge $c_{-}$ or higher central charge \cite{Kaidi:2021gbs, Kobayashi:2023eev}. Such an anomaly prevents the CFT from existing on its own in a tensor product Hilbert space. Second, consider a chiral gapped groundstate $\ket{\Psi}$ on a disk; the low-lying spectrum of the entanglement Hamiltonian $K_A$ (which will be explicitly defined later) of a disk $A$ matches that of the edge CFT Hamiltonian. This was first conjectured by Li and Haldane \cite{Li-Haldane}. Third, the algebraic theory of the infrared (IR) renormalization group (RG) fixed point \cite{Kitaev:2005hzj}, matches the algebraic description of the CFT \cite{Moore:1988uz, Moore:1988ss, Moore:1988qv, Moore:1989yh, Moore:1989vd, Kong:2019byq,Kong:2019cuu}. They both admit a mathematical structure that can be described by unitary modular tensor category (UMTC); this common structure has played an important role in the development of both subjects.  

The motivating question for this paper is: \emph{Can we develop a logical framework to attribute all these properties, which are universal in the phase, to a few locally-checkable entanglement conditions of a given representative state of the phase?} There has been lots of progress in the Entanglement Bootstrap program \cite{Shi_2020, Shi:2019ngw,Shi:2020jxd, Shi:2020rne, knots-paper, Shi:2018bfb, Shi:2018krj, Kim:2021tse, Kim2021, Shi:2023kwr, Lin:2023pvl, Kim:2023ydi,  Kim:2024amo, Vir,kim2024conformalgeometryentanglement,no-go, figure-eight, Kim:2024gtp, Yang:2025pke,Li:2025czz} that shares the same motivation in various physical contexts. 
The idea is to identify and exploit locally-checkable entanglement conditions (which can be called axioms) on a wavefunction that tell us that it represents a given category of state of matter. This amounts to deriving the defining properties of the phase of matter from the axioms. Generally speaking, the benefits from such an investigation are not only a clear logical relation among various universal properties, but also a guarantee of success of many schemes that extract the universal properties from a representative wavefunction (we will give examples later); one can also use the axioms to systematically search for new states \cite{Li:2025czz}.  Philosophically, the axioms should be renormalization group (RG) monotones, in the sense that the violations of the axioms decay under the RG flow. As a result, ultimately, one can explain why these universal properties emerge.

For example \cite{Shi_2020}, in the context of a 2+1D gapped state, based on two locally-checkable conditions in terms of entanglement entropies (known as \Azero and \Aone), one can derive much of the UMTC algebraic description of anyons. Later in \cite{Kim:2024amo, Kim:2024gtp}, it was shown that the string-net classification of non-chiral states can be derived from exact \Azero and \Aone condition in systems with finite local dimension. If one can further prove that \Azero and \Aone are RG monotones (which is a subject in progress called ``robust Entanglement bootstrap'') then we have a full explanation why UMTC or string-net model is the right description for the universal properties of 2+1D non-chiral gapped phases of matter. 

For chiral gapped phases, which lie outside the string-net classification, the origin of the chiral central charge $c_{-}$ from entanglement remains somewhat mysterious. In \cite{Kim:2021tse, Kim2021}, it is proposed and argued that one can extract $c_{-}$ via the modular commutator: 
\begin{equation}\label{eq:modular-commutator-intro}
    J(A,B,C)_{\ket{\psi}} \equiv \ii \langle \Psi| [K_{AB},K_{BC}] |\Psi\rangle = \frac{\pi}{3} c_{-},
\end{equation}
where $A,B,C$ are shown in Fig.~\ref{fig:modular-commutator-intro} (a). 
How to show that such $c_{-}$ satisfies the constraints from both UMTC and the chiral CFT is unknown. For example, for a chiral bosonic topologically-ordered system, it is known that $c_{-}$ is constrained by the quantum dimensions $d_{\mathfrak{a}}$ and topological spins $\theta_{\mathfrak{a}}$ of the anyons $\mathfrak{a}$ by the Gauss-Milgram equation $e^{2\pi \ii c_{-}/8} = \CD^{-1}\sum_{\mathfrak{a}}d_{\mathfrak{a}}^2 \theta_{\mathfrak{a}}$, where $\CD = \sqrt{\sum_{\mathfrak{a}} d_{\mathfrak{a}}^2}$ is the total quantum dimension \cite{Kitaev:2005hzj}. Can we show that $c_{-}$ in Eq.~\eqref{eq:modular-commutator-intro} satisfies this relation based on a few locally checkable entanglement conditions?

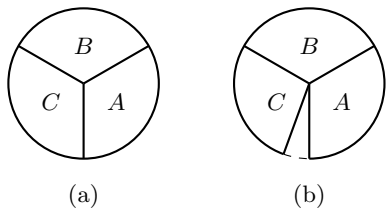
\begin{figure}[!h]
    \begin{tikzpicture}
        \draw[thick] (0,0) circle (1); 
        \draw[thick] (0,0) -- (30:1);
        \draw[thick] (0,0) -- (150:1);
        \draw[thick] (0,0) -- (270:1);
        \node at (-30:0.5) {$A$}; 
        \node at (90:0.5) {$B$};
        \node at (210:0.5) {$C$}; 
        \node at (0,-1.5) {(a)};

        \draw[thick] (3,0) ++ (-90:1) arc (-90:250:1);
        \draw[dashed] (3,0) ++ (250:1) arc (250:270:1);
        \draw[thick] (3,0) -- ++(30:1);
        \draw[thick] (3,0) -- ++(150:1);
        \draw[thick] (3,0) -- ++(270:1);
        \draw[thick] (3,0) -- ++(250:1);
        \node at ($(3,0) + (-30:0.5)$) {$A$}; 
        \node at ($(3,0) + (90:0.5)$) {$B$};
        \node at ($(3,0) + (210:0.5)$) {$C$}; 
        \node at ($(3,0) + (0,-1.5)$) {(b)};
    \end{tikzpicture} 
    \caption{Modular commutator on a ``complete'' and ``incomplete'' disk in (a) and (b) respectively.}
    \label{fig:modular-commutator-intro}
\end{figure} 

Motivated by these questions, it is desirable to formulate a bulk/edge correspondence, which connects the algebraic relations of the bulk entanglement Hamiltonians to the algebra of CFT. In this paper, we propose a hypothesis, namely the \emph{operator} bulk/edge correspondence, that serves this purpose. The role of the bulk/edge correspondence hypothesis can be briefly summarized in the following two aspects. (1) It is an explicit physical picture in which one can have a clear understanding of the universal properties in the entanglement entropies and entanglement Hamiltonians, including the various schemes such as Eq.~\eqref{eq:modular-commutator-intro} that extract universal data of the phase. (2) Furthermore, the operator bulk/edge correspondence is a navigation map towards the goal of building a logical framework that proves the universal properties discussed before from a set of local entanglement conditions. In such a framework, the operator bulk/edge correspondence is not assumed explicitly, but is the underlying rationale. In the rest of the introduction, we will elaborate more on these two aspects. 

{\bf Operator bulk/edge correspondence.}
We first set the stage. In this context of gapped chiral states in 2+1 dimensions on a lattice, there are three regimes of length scales associated with the entanglement properties. At a length scale $l \sim a$, the lattice spacing, one can change the entanglement by operations such as tensoring our state with Bell pairs or acting with a finite-depth unitary; this is clearly non-universal information and we refer to such a scale as the ``non-universal regime''. 
At length scales $l \gg \xi$, the correlation length, we have the familiar universal physics of topological quantum field theory (TQFT). Since the physics at this scale is independent of the metric, we refer to this scale as the ``topological regime''. This is also the length scale over which \Azero and \Aone hold. This in fact can also be regarded as an intrinsic definition of $\xi$ from the state. The surprising point we want to emphasize here is that there is also universal physics to be found in the regime $ a \ll l < \xi $. We name this as the ``corner regime'' because it is inevitably involved in the study of sharp corner. Notice that, if we demand a finite-dimensional local Hilbert space, the existence of such a regime is guaranteed by the no-go theorem of \cite{no-go}. 
A limit with infinite-dimensional local Hilbert space and zero correlation-length hides this important regime inside a single site. Entanglement properties in such a regime are short-ranged, but should not be treated the same as the regime of $l \sim a$. In fact, as we will show, many universal properties in the topological regime, such as Eq.~\eqref{eq:modular-commutator-intro}, have their ``roots'' in the entanglement properties in the corner regime. 

We apply the operator bulk/edge correspondence to study the universal properties of entanglement entropies and entanglement Hamiltonians of subregions. The rewards are the following:
\begin{itemize}
\item In the topological regime (meaning for regions with boundaries that are smooth on the scale of the correlation length) [Section~\ref{sec:bulk-edge}], we first obtain two conjectured parent Hamiltonians of the chiral gapped phase, reconstructed from a representative state. Consider a chiral gapped groundstate $\ket{\Psi}$ of a local Hamiltonian $H$ on a disk $D$. We obtain the reconstructed Hamiltonian density by decomposing the entanglement Hamiltonian $K_A = \sum_{v \in A}h^{\rec}_v$ of a disk $A$ inside the bulk. $K_A$ and $H$ have the same low-lying spectrum (up to a rescaling and a shift), which matches the low-lying spectrum of the edge CFT. Therefore, it is plausible to conjecture that $H_{\rec} = \sum_{v \in D}h^{\rec}_v$ (i.e.  extending the sum into the whole $D$) can be regarded as a Hamiltonian in the same phase. With $h_v^{\rec}$ near the entanglement boundary $\partial A$, we can also obtain a 1d Hamiltonian that describes the edge CFT, and matches the form of the reconstructed Hamiltonian for a CFT groundstate in \cite{Lin:2023pvl}. 
We note that \cite{Kim:2024amo} found a local parent Hamiltonian for exact {\bf A1} states, which is similarly made from local entanglement Hamiltonians, but which is a sum of commuting projectors; we do not expect this construction to be directly applicable for chiral gapped states.

\item The groundstate of the reconstructed Hamiltonian is observed to be closer than the input groundstate to the zero-correlation-length fixed point in the phase, as numerically tested in Section~\ref{sec:num}. Hence, one can iterate the procedure ``find groundstate -- reconstruct Hamiltonian -- find groundstate -- reconstruct Hamiltonian...'' to drive the state closer to the fixed point. The benefit of doing so is that the new state will have a smaller finite size error without changing the system size.

\item We apply the operator bulk/edge correspondence to study entanglement of a region $A$ with corners. \emph{A corner in $A$ is a region near $\partial A$ where the radius of curvature is smaller than the correlation length;} note that this definition applies even if the system is defined on a lattice. Based on the definition, it is inevitable that the entanglement in the intermediate regime $a \ll l < \xi$ is involved. 
In a chiral gapped state, the entanglement entropy $S_A$ of such a region must contain a corner contribution, by which we mean the contribution to the entanglement entropy from the degrees of freedom of the corner\footnote{We will give an explicit definition in the main text.}. With such a corner contribution $f(\theta)$, which is a function of the opening angle $\theta$, one can write the entanglement entropy as
\begin{align}\label{eq:SA-corner-intro}
    S_A = \alpha |\partial A| - \gamma + f(\theta) + \cdots,
\end{align}
where $|\partial A|$ is the circumference of $A$ with a non-universal coefficient $\alpha$, and $\gamma$ is a universal quantity of the phase known as the topological entanglement entropy (TEE). 
We note that there exist fine-tuned states in a gapped phase for which the $\gamma$ in Eq.~\eqref{eq:SA-corner-intro} contains spurious contributions; this usually results from the presence of a subsystem symmetry (e.g. \cite{Williamson:2018zig}), but not always \cite{Kato:2019cdi}. Such examples are ruled out by the Entanglement Bootstrap axioms \Azero and \Aone, which constrain the TEE to be $\gamma = \log\CD $.
One virtue of the Entanglement Bootstrap axioms is that they rule out such spurious contributions to the TEE.
Similar considerations apply to the extraction of $c_-$ by modular commutator.
The $\cdots$ in Eq.~\eqref{eq:SA-corner-intro} denotes subleading terms which decrease as one increases $|\partial A|$ (or decreases the correlation length). Such a decaying subleading term must exist as shown in \cite{no-go}. 

\end{itemize}

We pause to comment on our use of the word ``universal'' in the previous paragraph. By acting with a finite-depth local unitary circuit (and thus remaining in the same phase of matter), it is possible to change the value of $\gamma$. 
The spurious TEE contributions mentioned earlier are examples of such scenarios.
However, \cite{Kim:2023ydi} showed that by this process, the value of $\gamma$ can only be {\it increased} relative to its value at the zero correlation fixed-point state satisfying the {\bf A1} condition. Here we say $\gamma$ in Eq.~\eqref{eq:SA-corner-intro} contains a universal piece $\log \CD$ in the sense that $\gamma = \log \CD + a'$ where $a'\ge 0$ vanishes as the state approaches the zero-correlation-length RG fixed point.\footnote{Note such $a'$ does not decay away by increasing the size of $A$, and hence does not belong to $\cdots$ in Eq.~\eqref{eq:corner-universal-intro}.} 
Later we are going to use the word ``universal'' in a similar way in the discussion of $f(\theta)$. 

One way to see that such an $f(\theta)$ has to be present is as follows. Consider a chiral gapped state $\ket{\Psi}$ on a disk with the partition in Fig.~\ref{fig:circles}.
\begin{figure}[!h]
    \begin{tikzpicture}
        \draw[thick] (-2.2, -1.2) -- (2,-1.2) -- (2,1.2) -- (-2.2,1.2) -- cycle; 
        \draw[thick] (0,0) circle (1); 
        \draw[thick] (-1,0) circle (1); 
        \node at (-1.5,0) {$1$};
        \node at (-.5,0) {$2$};
        \node at (.5,0) {$3$};
        \node at (1.5,0) {$4$};
        \node at (1.5,0.8) {$\ket{\Psi}$};
    \end{tikzpicture}
\caption{Partition of a disk into four regions.}
\label{fig:circles}
\end{figure}
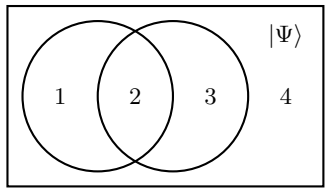
Suppose, by way of contradiction, that for an arbitrary region $A$, the state satisfies 
\be S_A = \alpha |\partial A| - \gamma_\text{topo},\ee
with a region-independent constant $\gamma_{\text{topo}}$; this is sometimes called a ``strict area law''. This implies that $I(i-1:i+1|i) = 0$, $i=1..4$\footnote{Regarding the notation, $I(A:C|B) \equiv S_{AB}+S_{BC}-S_{ABC}-S_B$ and $i+1,i-1$ should be understood modulo 4.}. A combination of arguments from \cite{Zou:2020bly} and \cite{Siva:2021cgo} (which we summarize in Appendix \ref{app:siva}) then implies that the state admits a gapped boundary and hence we find a contradiction. Notice that the regions in Fig.~\ref{fig:circles} necessarily have sharp corners.  
One way in which the conclusion can be evaded is the presence of a corner-dependent term in Eq.~\eqref{eq:SA-corner-intro}, instead of the strict area law\footnote{The $\cdots$ term will not be enough to evade this conclusion because it decays away as one increases the subsystem size.}.

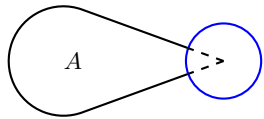
\begin{figure}[!h]
  \begin{tikzpicture}
    \draw[thick] (0,0) -- (180-20:2) arc [start angle=180-20-90, end angle=180+20+90, radius={2*tan(20)}] -- cycle; 
    \filldraw[fill=white, draw=blue, thick] (0,0) circle (0.5); 
    \draw[thick, dashed] (0,0) -- (180-20:0.5);
    \draw[thick, dashed] (0,0) -- (180+20:0.5);
    \node at (-2,0) {$A$};  
  \end{tikzpicture}
  \caption{A region $A$ with a sharp corner (dashed line), regulated by a hole. We will show that it is useful to envision a boundary CFT groundstate living on the boundary of the hole (blue line).}
  \label{fig:corner-intro}
\end{figure}

In [Section~\ref{sec:hypothesis}], we discuss the application of the operator bulk/edge correspondence to regions with sharp corners. In particular, we propose a regularization of a sharp corner by a hole [Fig.~\ref{fig:corner-intro}]. The edge CFT lives on the boundary of the hole. As we will discuss further, one can imagine that the hole 
is made by a local unitary as in \cite{Siva:2021cgo}, or that it was there already, as in \cite{Sopenko:2023utk}. In general, one can also argue the validity of this regulation from the Li-Haldane conjecture. As the Li-Haldane conjecture is believed to be universally satisfied across the chiral gapped phase, such a regularization scheme is also generically applicable. 
This picture has the following consequences [Section~\ref{sec:universal-corner-contributions}]. 

First, in the corner contribution $f(\theta)$ in Eq.~\eqref{eq:SA-corner-intro}, there exists a universal piece that is of the same form as the groundstate entanglement entropy of an interval of the 1d edge CFT on a circle. Explicitly,
\begin{equation}\label{eq:corner-universal-intro}
     f(\theta) = f_{\text{univ}}(\theta) + c', \text{ with }f_{\text{univ}}(\theta) = \frac{c_{\tot}}{6} \log(\sin(\theta/2)), 
\end{equation}
where $c_{\tot}$ is the minimal total central charge of the edge CFT, $\theta$ is the opening angle of the corner and $c'$ is the non-universal piece from the short-range entanglement in regime $l \sim a$. 
Corner contributions to the entanglement entropy in chiral gapped states were first identified by \cite{Rodriguez:2010dm}, and further studied by \cite{Sirois:2020zvc,Rozon:2019evk}. The specific form of $f_{\text{univ}}(\theta)$ in Eq.~\eqref{eq:corner-universal-intro} was obtained in \cite{Liu:2023pdz,Liu:2024ulq} using arguments of topological quantum field theory.
It also agrees with the numerical results for integer quantum Hall states studied by \cite{Rodriguez:2010dm, Sirois:2020zvc} when $\theta$ is not small\footnote{In the small angle limit $\theta \to 0$ studied in \cite{Rozon:2019evk,Sirois:2020zvc}, their results show $f(\theta) \sim - 1/\theta$, which does not match Eq.~\eqref{eq:corner-universal-intro}. First of all, $\theta \to 0$ is a very ultraviolet regime, whereas we anticipate that our formula works when $\theta \xi \gg \xi$. Moreover, our definition of corner contribution, which will be given explicitly in the main text, is different from the one in \cite{Rozon:2019evk,Sirois:2020zvc}. In fact, as explained in \cite{Liu:2022mop}, following their definition, the $-1/\theta$ behavior can be understood as a consequence of having a diverging length of the entanglement boundary which is moreover colliding in the limit $\theta \to 0$. We thank Meng Cheng for discussions of this apparent discrepancy.}. 

Here the meaning of ``universal'' of Eq.~\eqref{eq:corner-universal-intro} is similar to that of the TEE in Eq.~\eqref{eq:SA-corner-intro} discussed earlier. For a region finite size region $A$ with a corner, there are two sources of errors: the finite size error and spurious contributions. The finite size error can be reduced by increasing the subsystem size $|A|$. $f(\theta)$ could include a spurious corner contribution that depends on $\theta$, which cannot be reduced simply by increasing $|A|$. For example, one could imagine stacking arbitrary ``garbage'' in the corner region and the corner contribution will be changed. Nevertheless, just as in the result of the universal lower bound of TEE shown in \cite{Kim:2023ydi}, spurious contribution can only increase the corner entanglement and there is a sense that $f_{\text{univ}}(\theta)$ is the ``minimal'' contribution over states related by a finite depth quantum circuit. Explicitly, the word ``minimum'' is reflected by the $c_{\text{tot}}$. That is, for any corner contribution, one can formally write $f(\theta) = \frac{c_{\text{eff}}}{6} \log(\sin(\theta/2))$, where $c_{\text{eff}}$ might also be a function of $\theta$. The $f_{\text{univ}}(\theta)$ is ``minimal'' in the sense that $c_{\text{eff}} \geq c_{\text{tot}}$. 

Moreover, we can have a direct explanation of various topological quantities extracted from entanglement. For example, based on the bulk/edge correspondence, \cite{Kitaev:2005dm} already explains the origin of the TEE term using the CFT partition function. In the same spirit, from the operator bulk/edge correspondence with the corner regulation, we will explicitly derive the relation between the modular commutator and chiral central charge of the edge CFT [Eq.~\eqref{eq:modular-commutator-intro}]. In addition, we can use this tool to obtain a formula for the modular commutator for an ``incomplete disk'' [Fig.~\ref{fig:modular-commutator-intro} (b)], explaining some numerical observations of \cite{Fan:2022ipl}. 

{\bf A logical framework for corner entanglement.}
Now we switch gears to discuss the logical framework [Section~\ref{sec:logical-framework}]. The goal is to attribute the results about universal properties discussed earlier to a set of local entanglement conditions (axioms) on a quantum state $\ket{\Psi}$. In this context, we do \emph{not} assume the operator bulk/edge correspondence.

To state the axiom of corner entanglement, we first define $\fc(\ket{\Psi})$. It is a function of the state that, under the \Aone axiom, measures the amount of entanglement near a corner. $\fc$ can be used as a quantum-information-theoretic diagnostic for edge ungappability from the bulk wavefunction. We will show [Section~\ref{sec:diagnostic-for-ungappable}] that with a representative state of a gapped phase, $\{\fc\}_{\min} > 0$ if and only if the edge is ungappable (i.e. the gapped phase is chiral), where the minimal value of $\fc$ is taken across the states connected to the given representative by a local unitary around the corner. This is an alternative way of detecting the ungappability of the edge in a similar spirit to \cite{Siva:2021cgo}. The intuition is that among the states $\{\ket{\Phi}\}$ in a chiral gapped phase, $\fc(\ket{\Phi})$ has a non-zero minimal value $\fc(\ket{\Phi^\star})$. Under the corner regulation picture, the minimum is equal to the minimal total central charge of the edge CFT in the IR limit. 

Moreover, $\fc(\ket{\Phi})$ is stationary at the state $\ket{\Phi^\star}$ that achieves the minimum. Motivated by this, our axiom for $\ket{\Psi}$ is: $\fc(\ket{\Psi})$ is stationary. This stationarity condition can be locally checked by a vector fixed-point equation, whose relation to the one for 1+1D CFT \cite{Lin:2023pvl} and 2+1D chiral edge \cite{kim2024conformalgeometryentanglement} is manifest through the operator bulk/edge correspondence around the corner. 

The stationarity condition plays the following important roles in studying the universal corner entanglement. (1) As a starter, it can rule out the spurious corner contributions mentioned earlier, because $\fc$ of that state depends on $\xi$ with non-vanishing first derivative, and one can always tune the correlation length of the state by an adiabatic evolution. (2) Moreover, one can derive Eq.~\eqref{eq:corner-universal-intro} with $c_{\tot}$ replaced by $\fc$ of the given state. In addition, as we will explain, the opening angle $\theta$ is measured intrinsically from the entanglement. This in fact is a better measure of the opening angle, especially in the cases where the system is on a lattice or the rotational symmetry around the corner is unclear. (3) Furthermore, with the stationarity condition, we can derive the formula for modular commutator of the incomplete disk in Fig.~\ref{fig:modular-commutator-intro}(b). In summary, the stationairty condition results in a conformal geometry of corners -- a universal and intrinsic measurement of the angle of a corner modulo global conformal transformation. We will numerically verify these axioms as well as their logical conclusions [Section~\ref{sec:num}].

We end this introduction by posing a question 
that explains the title of the paper:
Does a TQFT encode all the universal properties of a chiral gapped groundstate? 
The TQFT description is valid for length scales $l \gg \xi$, in which the universal physics is independent of the metric of the space on which it lives. In this paper, we show that there is still universal physics at length scales $a \ll l < \xi$. In such a regime, there is a universal dependence of entanglement information on the conformal geometry of a metric near a corner --  it allows us to extract the chiral central charge with Eq.~\eqref{eq:modular-commutator-intro}; one can also obtain the minimal total central charge and detect edge ungappability. All of these properties are independent of any UV regulators. These results suggest that, for the purpose of understanding all the universal properties of a chiral gapped groundstate, it is perhaps not enough to only utilize TQFT. 

{\bf Organization.}
The paper is organized as follows: 

\S\ref{sec:bulk-edge} proposes an operator version of the Li-Haldane conjecture relating the entanglement Hamiltonian of a bulk disk in a chiral gapped state to its edge conformal field theory. We also propose several forms of reconstructed Hamiltonians within the same chiral gapped phase.

The purpose of \S\ref{sec:hypothesis} is to extend this hypothesis about the entanglement Hamiltonian in a chiral gapped state to regions with sharp corners.  
The key idea is that the bulk groundstate can be viewed as being filled with holes.  In \S\ref{sec:universal-corner-contributions} we use the picture established in the preceding sections to derive a number of concrete predictions about the entanglement structure of chiral groundstates.  
In particular, we derive the universal corner contribution to the entanglement entropy \eqref{eq:corner-universal-intro}, and explain how to isolate it from non-universal angle-dependent contributions. We also propose a vector fixed point equation, similar to the one in a 1+1D CFT groundstate \cite{Lin:2023pvl} and near a gapless edge of a 2+1D chiral gapped state \cite{Vir, kim2024conformalgeometryentanglement}. Finally, we use this picture to derive formulae for the modular commutators, consistent with results of \cite{Zou:2022nuj}.

In \S\ref{sec:logical-framework} we encapsulate these results into a logical framework, where we regard the stationarity of $\fc$ as the axiom and derive the rest of the universal properties of corners. 
This allows the extraction of a conformal geometry in analogy with \cite{kim2024conformalgeometryentanglement} that measures the sharp corners in a universal and intrinsic way. Then in \S\ref{sec:diagnostic-for-ungappable}, we define a quantity $(\fc)_{\text{min}}$ that can be interpreted as a minimal total central charge of the edge and is a diagnostic of edge ungappability like the one proposed in \cite{Siva:2021cgo}.

Finally, in \S\ref{sec:num} we provide numerical verifications of our results. 
One result from that section that we highlight here is a method to reconstruct a parent Hamiltonian from the groundstate wavefunction. A similar task is also accomplished in the context of 1+1D CFT \cite{Lin:2023pvl}. We show that the iterated process of reconstructing the Hamiltonian and finding its groundstate decreases the error of all of our RG fixed-point diagnostics.  

Appendix \ref{app:siva} contains a reminder of a sufficient condition for a gapped boundary found in \cite{Zou:2020bly} and \cite{Siva:2021cgo}. Appendix \ref{app:Hrec} explains the relationships between the entanglement Hamiltonian of a disk and various forms of the reconstructed Hamiltonians that follow from the bulk {\bf A1} condition. 
Appendix \ref{app:TT} contains the details of our calculations of modular commutators using the edge CFT.  
Several other appendices sequester details best enjoyed in a dark, quiet room.

\section{Operator bulk/edge correspondence hypothesis}

\label{sec:bulk-edge}
In this section, we discuss an explicit bulk/edge correspondence hypothesis in the context of a generic chiral gapped state $\ket{\Psi}$. The original bulk/edge correspondence, conjectured by Li and Haldane \cite{Li-Haldane}, refers to a matching of two spectra, namely, the low-lying spectrum of the entanglement Hamiltonian $K_A = -\ln \rho_A$ of a reduced density matrix $\rho_A = \Tr_{\overline{A}}\ketbra{\Psi}{\Psi}$ matches the low-lying spectrum of the edge chiral CFT. The precise statement will be provided later. Here we propose an \emph{operator} bulk/edge correspondence, namely that one can construct a sum of local operators from the bulk entanglement Hamiltonians which can be regarded as the edge Hamiltonian density. The operator bulk/edge correspondence is a universal description of the operator content in $K_A$.

\subsection{Preliminaries: bulk entropic conditions and locality of entanglement Hamiltonians}
\phantomsection\label{subsec:local-ent-hmt}
To begin, we set the context and discuss the EB axioms for liquid topological orders, which are two conditions on the entanglement entropies of local regions in the bulk. In a gapped phase, these two axioms are approximately satisfied with errors that decay as the states approach the zero-correlation-length RG fixed point, where it is believed that the state can be described by a TQFT, or pretty much equivalently, a UMTC \cite{Kitaev:2005hzj}. In this section, we mainly use their implication that one can decompose a entanglement Hamiltonian of a disk in the bulk into a sum over local operators. 

Throughout this paper, we will focus on a quantum state $\ket{\Psi}$ on a disk. The microscopic degrees of freedom live on a lattice, and the Hilbert space $\CH$ is a tensor product of local Hilbert spaces $\CH_v$ on each site $v$, $\CH = \otimes_{v}\CH_v$. The state $\ket{\Psi}$ can be thought of as a groundstate of a chiral gapped local Hamiltonian with a correlation length $\xi$.  

Given a region $A$, which is a collection of sites, the entanglement between $A$ and its complement is described by the local reduced density matrix $\rho_A \equiv \Tr_{\overline{A}}\ket{\Psi}\bra{\Psi}$. It can be written as $\rho_A = e^{-K_A}$ where $K_A$ is called the entanglement Hamiltonian. 

We posit two local conditions concerning the entanglement entropies $S(A) \equiv - \Tr \rho_A \log \rho_A$ among a set of local regions, namely \Azero and \Aone. They are two crucial axioms of Entanglement bootstrap that allow one to derive various universal properties of gapped phases of matter \cite{Shi_2020, Shi:2019ngw,Shi:2020jxd, Shi:2020rne, knots-paper, Shi:2018bfb, Shi:2018krj, Shi:2023kwr, Lin:2023pvl,Kim:2023ydi, figure-eight,  kim2024strict, Kim:2024gtp}. The statement of \Azero and \Aone is: 

\begin{figure}[thb]
    \centering
    \begin{tikzpicture}
        \draw[thick] (0, 0) circle (1);
        \draw[thick] (0, 0) circle (0.5);
        \draw[thick] (3, 0) circle (1); 
        \draw[thick] (3, 0) circle (0.5);
        \draw[thick] (3,0.5) -- (3,1); 
        \draw[thick] (3,-0.5) -- (3,-1); 
        \node at (-0.75,0) {$B$}; 
        \node at (0,0) {$C$}; 
        \node at (3,0) {$C$}; 
        \node at (2.25,0) {$B$};
        \node at (3.75,0) {$D$};
        \node at (0, -1.5) {(a)}; 
        \node at (3, -1.5) {(b)}; 
    \end{tikzpicture}
	\caption{\Azero and \Aone partitions. The linear size of the regions are larger than the correlation length $\xi$.}
	\label{fig:A0-A1}
\end{figure}
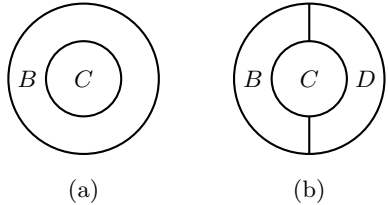

\begin{assumption}[\Azero]
    \label{assumption:A0}
    For any disk-like region in the bulk with a partition $BC$ topologically equivalent to the one in Fig.~\ref{fig:A0-A1} (a), where the typical linear sizes of $B,C$ are larger than $\xi$, we assume 
    \begin{align}
        \label{eq:A0}
         \Delta(B,C)\equiv (S_{BC} + S_C - S_B)_{\ket{\Psi}} \approx 0. 
    \end{align}
\end{assumption}
\begin{assumption}[\Aone]
    \label{assumption:A1}
    For any disk-like region in the bulk with a partition $BCD$ topologically equivalent to the one in Fig.~\ref{fig:A0-A1} (b), where the typical linear sizes of $B,C,D$ are larger than $\xi$, we assume 
    \begin{align}
        \label{eq:A1}
         \Delta(B,C,D) \equiv (S_{BC} + S_{CD} - S_B - S_D)_{\ket{\Psi}} \approx 0. 
    \end{align}
\end{assumption}
These two conditions are approximately satisfied by generic gapped states due to the area law entanglement property. That is, for a region $A$, the entanglement entropy is 
\begin{align}\label{eq:area-law}
    S(A) = \alpha |\partial A| - n\gamma + f_{\text{corner}} + \cdots
\end{align}
where the leading order is the circumference of $A$, denoted as $|\partial A|$, with a non-universal prefactor $\alpha$, the subleading terms include the universal topological entanglement entropy (TEE) $\gamma$ multiplied by the number of disconnected boundaries $n$ and $f_{\text{corner}}$ for the corner contributions. The $\cdots$ term includes the possible corner contributions and the other subleading terms which decays as the size of $A$ increases. The \Azero and \Aone combination is designed such that the the leading order boundary term, TEE term and the corner term is canceled. 

We should comment on the meaning of the $\approx$ symbol in in Eq.~\eqref{eq:A0} and Eq.~\eqref{eq:A1}. For a chiral gapped state on a lattice with finite local dimensions, as shown in \cite{no-go}, the \Aone combination in Eq.~\eqref{eq:A1} will not vanish exactly, but with a small violation. We also believe the same for \Azero combination in Eq.~\eqref{eq:A0}. The source of the violation is from the $\cdots$ term in Eq.~\eqref{eq:area-law} which decays as the subsystem size increases. As a result, the violation of the \Azero and \Aone condition, namely $\Delta(B,C)$ and $\Delta(B,C,D)$, shall decay as a function of the size of the regions, as numerically demonstrated in \cite{Vir}. Once the linear sizes are larger than $\xi$, one can ignore the error and replace $\approx$ with $=$ in Eq.~\eqref{eq:A0} and Eq.~\eqref{eq:A1}. This can also be regarded as the definition of our $\xi$. We also assume that the logical conclusions derived from \Azero and \Aone with the error ignored are applicable. Although there are unphysical pathological counterexamples \cite{2018arXiv180205477S} to this assumption, regarding to the Markov conditions that we will discuss momentarily, we believe a physical system shall have continuous behavior under scaling and therefore it is safe to make such an assumption. One can think of the statements we are about to make, following from \Azero and \Aone with error ignored, as statements for the physics in the scale $l/\xi \to \infty$. In the following, to avoid clutter, we shall just regard the $\approx$ as an equal sign in Eq.~\eqref{eq:A0} and Eq.~\eqref{eq:A1}, as well as all the equations resulting from them that we will introduce below. 

What are implications of \Azero and \Aone? On the high level of the physical picture, these two assumptions imply that, on length scales larger than $\xi$, the state behaves as a renormalization group (RG) fixed-point representative of the chiral gapped phase. This statement is reflected in the following two technical properties: 
\begin{itemize}
    \item Vanishing correlations: If \Azero is satisfied, then for any region $A$ buffered with an annulus $B$ from any region $C$ as shown in Fig.~\ref{fig:implication-from-A0-A1} (a), we have 
    \begin{align}
        I(A:C)_{\ket{\Psi}} = 0 \quad \Leftrightarrow \quad \rho_{AC} = \rho_A \otimes \rho_C \quad \Leftrightarrow \quad K_{AC} = K_A + K_C. 
    \end{align} 
    This is saying that, any regions $A$ and $C$ separated by a distance larger than $\xi$ can be thought of as having no correlations. 
    \item Local recoverability: If \Aone is satisfied, then for any region $A$ buffered with a disk $B$ from any region region $C$ as shown in Fig.~\ref{fig:implication-from-A0-A1} (b), we have \cite{Petz:2002eql}
    \begin{align}
        I(A:C|B)_{\ket{\Psi}} = 0 \quad &\Leftrightarrow \quad \rho_{ABC} = \rho_{BC}^{1/2}\rho_B^{-1/2} \rho_{AB} \rho_B^{-1/2}\rho_{BC}^{1/2} \label{eq:petz} \\ 
        \quad &\Leftrightarrow \quad K_{ABC} = K_{AB} + K_{BC} - K_B. \label{eq:Markov-decompose}  
    \end{align}  
    Here Eq.~\eqref{eq:petz} is the Petz recovery map and we shall refer to Eq.~\eqref{eq:Markov-decompose} as a Markov decomposition. This result is saying that one can obtain the entanglement of a larger region encoded in $\rho_{ABC}$ from the entanglement information of smaller regions encoded in $\rho_{AB},\rho_{BC}$. Therefore, on the intuitive level, one can conclude that the ``information'' encoded in regions with length scale $a \xi$ is the same as the ``information'' encoded in regions with length scale $b \xi$, with $b > a >1$, which is a reflection of the scale invariant and hence RG fixed point.
\end{itemize}
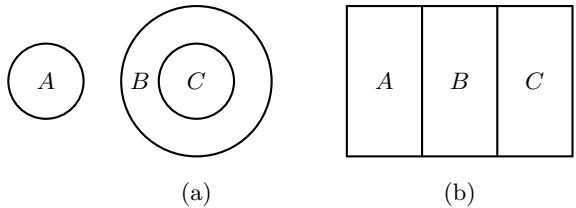
\begin{figure}[htb]
    \centering
    \begin{tikzpicture}
        \draw[thick] (0, 0) circle (1);
        \draw[thick] (0, 0) circle (0.5);
        \draw[thick] (-2, 0) circle (0.5);
        \node at (-2,0) {$A$}; 
        \node at (-0.75,0) {$B$}; 
        \node at (0,0) {$C$}; 
        \node at (0, -1.5) {(a)}; 

        \draw[thick] (2,-1) -- (5, -1) -- (5, 1) -- (2, 1) -- cycle; 
        \draw[thick] (3, -1) -- (3, 1); 
        \draw[thick] (4, -1) -- (4, 1); 
        \node at (2.5,0) {$A$}; 
        \node at (3.5,0) {$B$};
        \node at (4.5,0) {$C$};
        \node at (3.5, -1.5) {(b)}; 
    \end{tikzpicture}
	\caption{Implications from \Azero and \Aone conditions.}
	\label{fig:implication-from-A0-A1}
\end{figure}

An immediate consequence of Eq.~\eqref{eq:Markov-decompose} is that one can do Markov decomposition of a entanglement Hamiltonian $K_D$ of a disk $D$ whose linear size is much larger than $\xi$ \cite{Kim:2021tse}. Inside the disk $D$, we first group the sites into supersites, as shown in Fig.~\ref{fig:sum-Deltas-bdy} where each plaquette is a supersite. We demand the linear size of each supersite is much larger than $\xi$ and one can apply \Azero and \Aone by treating a supersite as a unit site. Then one can apply the Markov decomposition Eq.~\eqref{eq:Markov-decompose} to $K_D$ and obtain \cite{Kim:2021tse}
\begin{align}
    \label{eq:KD-decompose}
    K_D = \sum_{f\in D} K_f - \sum_{e \notin D_{\partial}} K_e + \sum_{s \notin  D_{\partial}} K_s, 
\end{align}
where $f$ stands for a connected 3-supersite, $e$ stands for a connected 2-supersite and $s$ is a single supersite illustrated in Fig.~\ref{fig:sum-Deltas-bdy}. $D_{\partial}$ stands for the layer along the entanglement boundary with 1-supersite thickness. We also define $D_{\inte} \equiv D \backslash D_{\partial}$. We remark that each supersite is \emph{larger} than the minimal length scale where \Azero and \Aone are applicable, for reasons that will be clear later.

\begin{figure}[htb]
    \centering
    \includegraphics[width=0.5\textwidth]{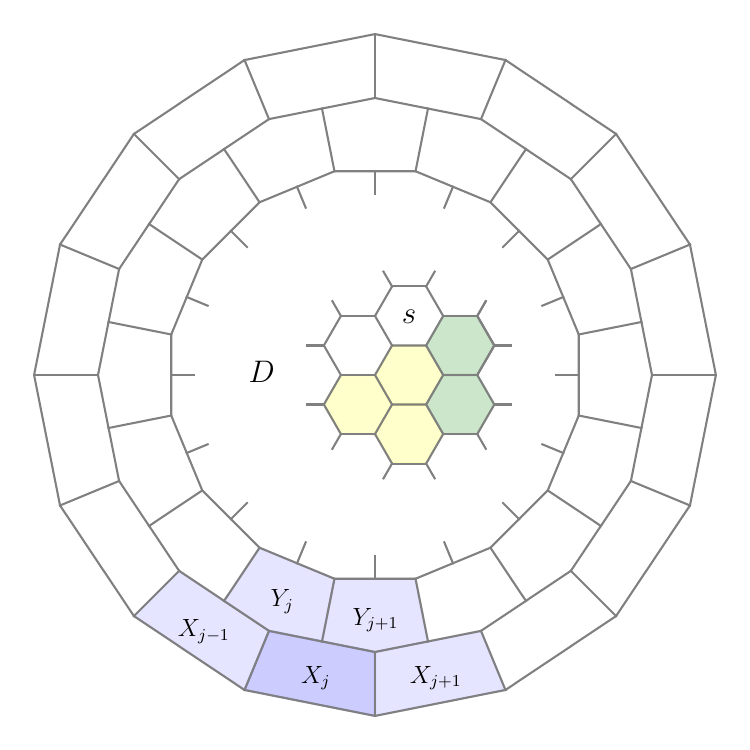}
	\caption{Decomposition of $K_D$ and the support of $K_{D_{\partial}}$ along the entanglement boundary. Each plaquette represents a coarse-grained site (a supersite). The yellow and green regions are examples of a 3-supersite and a 2-supersite, which are denoted as $f$ and $e$ respectively. Notice inside the bulk, the degrees of the supersites are not all equal to 6. $D_{\partial}$ is outer layer $X_1\cdots X_N$, and $D_{\inte} = D \setminus D_{\partial}$.}
	\label{fig:sum-Deltas-bdy}
\end{figure}

The decomposition of $K_D$ in Eq.~\eqref{eq:KD-decompose} can also be written as 
\begin{equation}\label{eq:KD-hrec}
    K_D = \sum_{s \in D_{\inte}} h_s^{\rec} + \sum_{s \in D_{\partial}} h_s^{\partial,\rec} 
\end{equation}
where 
\begin{equation}
    \label{eq:hrec}
     h_s^{\rec} = \frac{1}{2} \overline{\hat{\Delta}_s} - \frac{6-d_s}{6} \overline{\hat{\gamma}_{s}},\quad h_s^{\partial, \rec} = \frac{1}{2} \overline{\hat{\Delta}_s^{\partial}}. 
\end{equation}
Sometimes, we wish to discuss the contribution in the bulk $D_{\inte}$ and boundary $D_{\partial}$ separately. So we define 
\begin{align}\label{eq:KD-separate}
     K_{D_{\inte}} \equiv \sum_{s \in D_{\inte}} h_s^{\rec} ,\quad 
     K_{D_{\partial}} &\equiv \sum_{s \in D_{\partial}} h_s^{\partial,\rec}.
\end{align}
Let us pause to explain the notation. One important object is the operator version of the \Aone combination: 
\begin{align}
    \hat{\Delta}(B,C,D) \equiv K_{BC}+K_{CD} - K_B - K_D. 
\end{align}
We will call $C$ the center region and $BD$ the buffer region. In Eq.~\eqref{eq:hrec}, the operator $\overline{\hat{\Delta}_s}$ denotes the averaging over all possible $\hat{\Delta}$ whose center is $s$ and the buffer is an annulus with thickness of one supersite: 
\begin{equation}
    \overline{\hat{\Delta}_s} = \frac{1}{N_s} \sum_{(B,D)} \hat{\Delta}(B,s,D), 
\end{equation}
where we consider all possible partitions of buffer into two disks $B,D$ and $N_s$ is the number of the partitions. See App.~\ref{app:Hrec} for a more explicit expression. Near the entanglement boundary $\partial D$ (sometimes just denoted as $\partial$ if there is no confusion), the $\overline{\hat{\Delta}_{X_j}^{\partial}}$ centered at $X_j$ is defined as 
\begin{align}
    \overline{\hat{\Delta}_{X_j}^{\partial}} \equiv \frac{1}{3} \[ \hat{\Delta}(X_{j-1}, X_j, Y_{j} Y_{j+1} X_{j+1}) + \hat{\Delta}(X_{j-1} Y_{j}, X_j, Y_{j+1} X_{j+1}) + \hat{\Delta}(X_{j-1}Y_{j} Y_{j+1}, X_j X_{j+1}) \]. 
\end{align}
This operator is the average of all the possible operator versions of boundary \Aone centered at $j$, as discussed in \cite{Shi:2020rne}. It is also involved in the ``conformal ruler'' defined for gapless edge in \cite{kim2024conformalgeometryentanglement}. Here we explicitly demand that the supersites in the outer two layers form a triangle lattice, so that the boundary term $K_{D_{\partial}}$ is a sum over these $\overline{\hat{\Delta}_{X_j}^{\partial}}$. This is always allowed as we are considering a region $D$ with a large enough linear size $|\partial D| \gg \xi$. In $h_{s}^{\rec}$, there is also an averaging of the operator version of the Kitaev-Preskill TEE. For a connected 3-supersite $f = ABC$, we define 
\begin{align}
     \hat{\gamma}(f) = K_{AB} + K_{BC} + K_{CA} - K_A - K_B - K_C - K_{ABC}, 
\end{align}
whose expectation value on the state computes the TEE $\langle \hat{\gamma}(f) \rangle = \gamma$ in Kitaev-Preskill partition \cite{Kitaev:2005dm}. For a supersite $s$, we define the average as 
\begin{align}
    \overline{\hat{\gamma}_s} \equiv \frac{1}{d_s} \sum_{f, s\subset f} \hat{\gamma}(f),
\end{align}
where $f$ runs over all the connected 3-supersites that contain $s$. Here we use $d_s$ to denote the degree of the supersite $s$. Notice the coefficient $\frac{6-d_s}{6}$ can be interpreted as the curvature at the site $s$. 

In Eq.~\eqref{eq:KD-hrec}, since the size of the supersite is larger than $\xi$, one can use the bulk 
\Aone condition to deform the buffer region in $\hat{\Delta}$ and $\hat{\Delta}^{\partial}$ as in Fig.~\ref{fig:deform}. The invariance under the deformation follows the same derivation as in \cite{Vir}, where it is proved that the action of a ``good modular flow generator'' is invariant under deformation. Therefore, the support of $K_{D_{\inte}}$ is away from the entanglement boundary $\partial D$ by a distance $O(\xi)$, while the support of $K_{D_{\partial}}$ is near the entanglement boundary with a thickness of $O(\xi)$. 

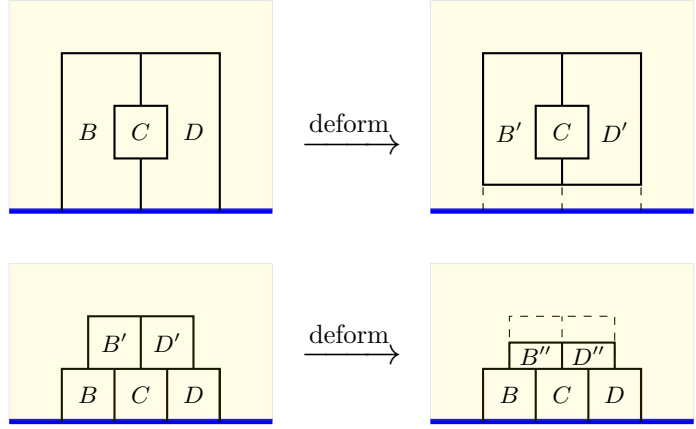
\begin{figure}[htb]
    \centering
  \begin{tikzpicture}[scale = 0.7]
    \begin{scope}
    \draw[color = blue, line width = 2pt] (0,0) -- (5, 0); 
    \draw[fill=yellow, opacity=0.1] (0,0) -- (5,0) -- (5,4) -- (0,4) -- cycle;
    \draw[thick] (2,1) -- (3,1) -- (3,2) -- (2,2) -- cycle;  
    \draw[thick] (1,0) -- (1,3) -- (4,3) -- (4,0);
    \draw[thick] (2.5,0) -- (2.5,1);  
    \draw[thick] (2.5,2) -- (2.5,3);
    \node[font = \Large] at (6.5,1.5) {$ \xrightarrow{\text{deform}} $};

    \node[font = \small] at (1.5, 1.5) {$B$}; 
    \node[font = \small] at (2.5, 1.5) {$C$}; 
    \node[font = \small] at (3.5, 1.5) {$D$}; 
    \end{scope}

    \begin{scope}[xshift = 8cm]
    \draw[color = blue, line width = 2pt] (0,0) -- (5, 0); 
    \draw[fill=yellow, opacity=0.1] (0,0) -- (5,0) -- (5,4) -- (0,4) -- cycle;
    \draw[thick] (2,1) -- (3,1) -- (3,2) -- (2,2) -- cycle;  
    \draw[thick] (1,0.5) -- (1,3) -- (4,3) -- (4,0.5) -- cycle;
    \draw[thick] (2.5,0.5) -- (2.5,1);  
    \draw[thick] (2.5,2) -- (2.5,3);
    \draw[dashed] (2.5,0) -- (2.5, 0.5); 
    \draw[dashed] (1,0) -- (1, 0.5); 
    \draw[dashed] (4,0) -- (4, 0.5); 
    \node[font = \small] at (1.5, 1.5) {$B'$}; 
    \node[font = \small] at (2.5, 1.5) {$C$}; 
    \node[font = \small] at (3.5, 1.5) {$D'$};
    \end{scope}

    \begin{scope}[yshift = -4cm]
    \draw[color = blue, line width = 2pt] (0,0) -- (5, 0); 
    \draw[thick] (1,0) -- (1,1) -- (4,1) -- (4,0);
    \draw[thick] (1.5,1) -- (1.5,2) -- (3.5,2) -- (3.5,1);
    \draw[thick] (2,0) -- (2,1);
    \draw[thick] (3,0) -- (3,1); 
    \draw[thick] (4,0) -- (4,1); 
    \draw[thick] (2.5,1) -- (2.5,2); 
    \draw[fill=yellow, opacity=0.1] (0,0) -- (5,0) -- (5,3) -- (0,3) -- cycle;
    \node[font=\Large] at (6.5,1.5) {$ \xrightarrow{\text{deform}} $};

    \node[font = \small] at (1.5, 0.5) {$B$}; 
    \node[font = \small] at (2.5, 0.5) {$C$}; 
    \node[font = \small] at (3.5, 0.5) {$D$}; 

    \node[font = \small] at (2, 1.5) {$B'$}; 
    \node[font = \small] at (3, 1.5) {$D'$}; 
    \end{scope}

    \begin{scope}[xshift = 8cm,yshift = -4cm]
    \draw[color = blue, line width = 2pt] (0,0) -- (5, 0); 
    \draw[thick] (1,0) -- (1,1) -- (4,1) -- (4,0);
    \draw[thick] (2,0) -- (2,1);
    \draw[thick] (3,0) -- (3,1); 
    \draw[thick] (4,0) -- (4,1); 
    \draw[dashed] (1.5,1) -- (1.5,2) -- (3.5,2) -- (3.5,1);
    \draw[dashed] (2.5,1) -- (2.5,2); 
    \draw[thick] (1.5,1) -- (1.5,1.5) -- (3.5,1.5) -- (3.5,1);
    \draw[thick] (2.5,1) -- (2.5,1.5); 
    \draw[fill=yellow, opacity=0.1] (0,0) -- (5,0) -- (5,3) -- (0,3) -- cycle;

    \node[font = \small] at (1.5, 0.5) {$B$}; 
    \node[font = \small] at (2.5, 0.5) {$C$}; 
    \node[font = \small] at (3.5, 0.5) {$D$}; 

    \node[font = \small] at (2, 1.25) {$B''$}; 
    \node[font = \small] at (3, 1.25) {$D''$}; 
    \end{scope}
  \end{tikzpicture}
  \caption{Deformations of a $\hat{\Delta}$ in $K_{D_{\text{int}}}$ (top) and $\hat{\Delta}^{\partial}$ in $K_{D_{\partial}}$ (bottom), near an entanglement boundary shown by the blue line. With the bulk \Aone condition, one can show that $\hat{\Delta}(B,C,D) = \hat{\Delta(B',C,D')}$ (top) and $\hat{\Delta}(BB',C,DD') = \hat{\Delta}(BB'',C,DD'')$ (bottom).}
\label{fig:deform}
\end{figure} 

Let us explain the physical picture behind this result of $K_D$ decomposition Eq.~\eqref{eq:KD-hrec} and its separation into Eq.~\eqref{eq:KD-separate}. For a gapped state, there are two types of entanglement. One is from short-distance fluctuation of degrees of freedom and the other is from the long-range fluctuation due to anyon loops. Each of these contribution is reflected in Eq.~\eqref{eq:KD-hrec}.  
\begin{itemize}
    \item The short-range fluctuation has range of correlation of order $\xi$. This short-range entanglement contributes to $K_D$ additively around the entanglement boundary as $(K_D)_{\text{short-range}} = \int_{\partial D} dx \hat{O}(x)$, see Fig.~\ref{fig:ent-contribution}. In Eq.~\eqref{eq:KD-hrec}, this is exactly captured by $K_{D_{\partial}}$. 
    \item The long-range fluctuation is due to topological order.
    It is the source of $-\gamma$ term in the entanglement entropy. This long-range entanglement is reflected as sectorization in $\rho_D$. Explicitly, in the Hilbert space $\CH_D$, there exists a set of orthogonal subspaces $\CH_D^{\mathfrak{a}}$ that are distinguished from each other by the presence of anyon $\mathfrak{a}$. Let $\rho_D^{\mathfrak{a}}$ denote the reduced density matrix on $D$ with an anyon $\mathfrak{a}$, then $\rho_D^{\mathfrak{a}}$ has non-zero support only in the sector $\CH_D^{\mathfrak{a}}$. Such a sectorization can be shown by utilizing fidelity, defined for any two quantum states $\sigma, \lambda$ as $F(\sigma, \lambda) = \Tr\(\rho^{1/2}\sigma \rho^{1/2} \)^{1/2}$. First in $D$ consider an annulus $A$ through which anyon string passes. \cite{Shi_2020} shows that the fidelity $F(\rho_A^{\mathfrak{a}}, \rho_A^{\mathfrak{b}})=0$ where $\rho_A^{\bullet} = \Tr_{D\backslash A} \rho_D^{\bullet}$. Since fidelity is a non-negative quantity and it is non-increasing under partial trace, we can conclude $0 = F(\rho_A^{\mathfrak{a}}, \rho_A^{\mathfrak{b}}) \geq F(\rho_D^{\mathfrak{a}}, \rho_D^{\mathfrak{b}}) \geq 0$ and hence $F(\rho_D^{\mathfrak{a}}, \rho_D^{\mathfrak{b}}) = 0$. This implies the operator product $\rho_D^{\mathfrak{a}} \cdot \rho_D^{\mathfrak{b}} = 0$. Because $\rho_D^{\mathfrak{a}}$ has a kernel in $\oplus_{\mathfrak{b} \neq \mathfrak{a}} \CH_D^{\mathfrak{a}}$, when computing $K_D^{\mathfrak{a}} = - \log (\rho_D^{\mathfrak{a}})$, there exists a term $\Lambda \sum_{\mathfrak{b} \neq \mathfrak{a}} P_D^{\mathfrak{b}} $ with $\Lambda = - \log (\epsilon) \to \infty$ for regulating the zero eigenvalues with $\epsilon \to 0$ in the kernel of $K_D^{\mathfrak{a}}$. This infinite projector is the role of $K_{D_{\inte}}$. In Eq.~\eqref{eq:KD-hrec}, if we regard $h_s^{\rec}$ as a (reconstructed) Hamiltonian density, then inserting an anyon in site $s$ will produce an infinite energy. We will show this in App.~\ref{app:Hrec}. In \S\ref{subsec:Hrec} we discuss more about regarding $h_s^{\rec}$ as a Hamiltonian density to reconstruct a Hamiltonian. 
\end{itemize} 

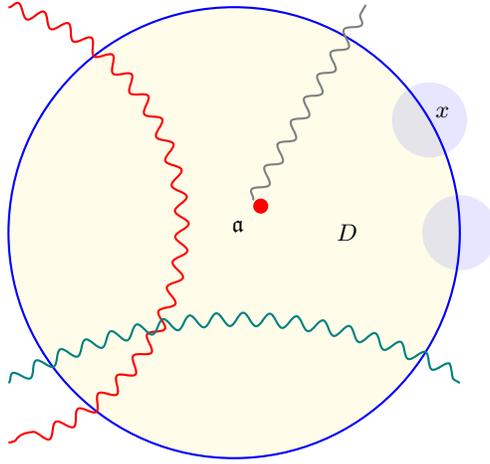
\begin{figure}
  \begin{tikzpicture}
    \filldraw[fill = yellow!10, draw = blue, thick] (0,0) circle (3); 
    \fill[fill = blue, opacity = 0.1] (0:3) circle (0.5);
    \fill[fill = blue, opacity = 0.1] (30:3) circle (0.5);
    \draw[thick, color = red, decorate, decoration=snake] (-3,3) arc (75:-75:3); 
    \draw[thick, color = teal, decorate, decoration=snake] (-3,-2) arc (120:60:6); 
    \draw[thick, color = gray, decorate, decoration=snake] (60:0.5) -- (60:3.5); 
    \fill[fill = red] (45:0.5) circle (0.1); 
    \node at (1.5, 0) {$D$}; 
    \node at (30:3.2) {$x$}; 
    \node at (60:0.1) {$\mathfrak{a}$}; 
  \end{tikzpicture}
  \caption{Two types of entanglement captured by $K_D$. Around the entanglement boundary $\partial D$, there are short-range entanglement contribution. A blue cluster stands for a local fluctuation (e.g. at location $x$) with a finite correlation range. The wiggle lines are anyon strings. The anyon inside $D$ dictates the sector in which the low-lying spectrum of $K_D$ lives.}
  \label{fig:ent-contribution}
\end{figure}

\subsection{Operator bulk/edge correspondence}
What exactly is the nature of $K_{D_{\partial}}$? Here we make a hypothesis that $K_{D_{\partial}}$ is algebraically isomorphic to a one dimensional edge Hamiltonian. In the context where $\ket{\Psi}$ is a representative of a chiral gapped phase, the edge of the system is generically described by a chiral CFT. For simplicity, let's say it is a purely chiral CFT, namely there is only one chiral component. 

We first consider the case where the entanglement boundary is smooth. Explicitly, what that means is that, everywhere along the entanglement boundary, the radius of the curvature is much larger than the correlation length. The consequence of this is that, for every segment in $\partial D$, one can zoom in to an extent that the segment is flat while still in the length scale larger than $\xi$. 

\begin{hypothesis}[operator bulk/edge correspondence for smooth entanglement boundary]
    \label{hypo:op-bulk-edge}
    Let $\CH_{D_{\partial}}$ be the Hilbert space on which the operator $K_{D_{\partial}}$ acts, and let $\CH_{\CFT}$ be the Hilbert space of the chiral CFT that describes the edge of the 2d system; we hypothesize that there exists an isometry $\mathbb{V}: \CH_{\text{CFT}} \to \CH_{D_{\partial}}$ (i.e. $\mathbb{V}^{\dagger}\mathbb{V} = \mathbbm{1}$), such that 
    \begin{align}\label{eq:hypo-bulk-edge-1}
         K_{D_{\partial}} & \eqover{\mathbb{V}}  \beta \int_{\partial D} dx T(x)  \\ 
        \ket{\Psi} &= \mathbb{V} \ket{\beta}_{\text{CFT}} 
    \end{align}
    where $\ket{\beta}_{\text{CFT}}$ is a thermal double state of the CFT with the inverse temperature $\beta$ that can be identified with the correlation length $\xi$ of the 2d system. $O_1 \eqover{\mathbb{V}} O_2$ is a shorthand notation for $O_1 \mathbb{V} = \mathbb{V} O_2$. 
\end{hypothesis}

Before we dive into the technical aspects of this statement, let us first explain the underlying intuition behind it.

We start with the Li-Haldane conjecture and its relations with bulk topological order. Li-Haldane conjecture says that, for a disk $D$ inside the bulk of a chiral gapped state, the low-lying spectrum of $K_D$ matches with the low-lying spectrum of the edge CFT (up to an overall rescaling and a shift). In the future explanation, we will compress this sentence and simply say the spectrum of $K_D$ ``matches'' the spectrum of the edge CFT. The edge CFT is believed to be a rational CFT, meaning that the complete Hilbert space of the CFT admits a finite number of irreducible representations of the chiral algebra. That is $\CH_{\chi \CFT} = \oplus \CV_{\mathfrak{a}}$, where $\CV_{\mathfrak{a}}$ is the carrier space of an irrep with a label $\mathfrak{a}$, which is believed to be in one-to-one correspondence to the label of anyons of the bulk topological order. Explicitly, suppose there is an anyon $a$ inside $D$, the spectrum of the entanglement Hamiltonian $\mathrm{spec}\( K_D^{\mathfrak{a}} \) \approx \frac{\xi}{|\partial D|} \mathrm{spec} \( L_0 \)_{\CV_{\mathfrak{a}}}$\footnote{Here we omit an additive constant shift.}, where $L_0$ is the Virasoro zero mode that encodes the edge CFT spectrum in the sector $\CV_{\mathfrak{a}}$. Notice this is consistent with our earlier conclusion with Fig.~\ref{fig:ent-contribution}, where the low-lying spectrum of $K_D^{\mathfrak{a}}$ is inside the sector $\CH_D^{\mathfrak{a}}$ mentioned above.    

{\bf Argument of Hypothesis~\ref{hypo:op-bulk-edge} from the local decomposition of $K_D$.} 
The hypothesis~\ref{hypo:op-bulk-edge} is a generalization of this Li-Haldane conjecture. It is an attempt to pin down the role of the operator content of $K_D$ in producing the CFT spectrum of $K_D$. In Eq.~\eqref{eq:KD-decompose}, we've shown that, for a gapped state close to the IR fixed point, it can be decomposed into a sum of local operators. If the state satisfies the Li-Haldane conjecture, then the spectrum of such a sum of local operators (i.e. $K_D$) ``matches'' with the spectrum of the parent local Hamiltonian $H$ of this state. This is because for a gapped local Hamiltonian with gapless edge, the low-lying spectrum is dominated by the edge excitation. As a result, the sum of local operator from decomposing $K_D$ can be regarded as a local Hamiltonian for the same gapped phase on the disk $D$. 

Moreover, in $K_D = K_{D_\inte} + K_{D_\partial}$, as we explained earlier, the support of $K_{D_{\partial}}$ is localized around $\partial D$ with a thickness of $O(\xi)$. Therefore, $K_{D_{\partial}}$ can be regarded as the term that governs the boundary excitations, which produce the CFT spectrum. This is roughly the reason why we hypothesize that $K_{D_{\partial}}$ is algebraically isomorphic to the boundary Hamiltonian $\beta\int_{\partial D} dx T(x)$. The coefficient $\beta$ can be identified by matching the spectrum. As for $K_{D_\inte}$, its role, besides to produce gap in the bulk, is to enforce the CFT spectrum to be in the sector that matches the anyon content inside $D$. As we explained at the end of Section~\ref{subsec:local-ent-hmt}, if there is an anyon $a$ inside $D$, states in the sector of other anyons will have infinite energies. 

The intuition described above can be distilled into the following mathematical argument. One can show that, based on \Azero and \Aone condition, $K_{D_\partial}$ fully captures the action of $K_D$ on $\ket{\Psi}$. Explicitly, for any function of an operator $f(\hat{O})$ that is well-defined in terms of power expansion, we have
\begin{equation}\label{eq:KD-action}
    f(K_D)\ket{\Psi} = f(K_{D_\partial} - \gamma \mathbbm{1})\ket{\Psi}. 
\end{equation}
We leave the details of the derivation to App.~\ref{app:property}. In particular, one can take $f(\bullet)$ to be any operator polynomial $\mathrm{poly}(\bullet)$ and obtain 
\begin{align}
    \Tr \[ \rho_D \mathrm{poly}(K_D) \] = \Tr \[ \rho_D \mathrm{poly}(K_{D_\partial} - \gamma \mathbbm{1}) \]. 
\end{align}
This suggests that the eigenvalues of $K_D$ and $K_{D_\partial} - \gamma \mathbbm{1}$ whose corresponding eigenstates are within the support of $\rho_D$ are the same. Mathematically speaking, 
\begin{align}\label{eq:spec-match}
    \mathrm{spec}_{\rho_D}(K_D) = \mathrm{spec}_{\rho_D}(K_{D_\partial} - \gamma \mathbbm{1}),
\end{align}
where $\mathrm{spec}_{\sigma}(O)$ is defined as a set of eigenvalues $\lambda$ of $O$ such that the corresponding eigenvectors $\ket{\lambda}$ satisfy $\tr(\sigma \ket{\lambda}\bra{\lambda}) \neq 0$. We note that in Eq.~\eqref{eq:spec-match} the equal sign should only be approximately satisfied due to finite size error. This is because, the derivation of Eq.~\eqref{eq:KD-action} makes use of exact \Azero and \Aone, which we know is subject to finite size error that decays as the subsystem size increases\footnote{At least we believe so for physical systems. One can regard this statement of decaying error as an assumption.}. Notice the left hand side of Eq.~\eqref{eq:spec-match} ``matches'' with the low-lying chiral edge CFT spectrum, following from the Li-Haldane conjecture, therefore, one can conclude that 
\begin{align}\label{eq:KDpartial-L0}
    \( K_{D_\partial} - \gamma \mathbbm{1} \)_{\CW} \propto \( L_0  + \text{const} \)_{\CV_1} ,
\end{align}
where $\CW$ is a Hilbert space spanned by eigenvectors $\ket{\lambda}$ of $K_{D_\partial} - \gamma \mathbbm{1}$ such that $\Tr(\sigma \ket{\lambda}\bra{\lambda}) \neq 0$, and $\CV_1$ is the identity irrep of the chiral algebra of the edge CFT. If $\rho_D$ contains an anyon $\mathfrak{a}$, then $\CV_1$ should be replaced by the corresponding irrep sector $\CV_{\mathfrak{a}}$. Here we use $(O)_{\CH}$ to denote the operator restricted on the Hilbert space $\CH$, by projecting out the eigenstates that are not in $\CH$. With the local decomposition of $K_{D_\partial} = \sum_{s \in D_\partial} h_s^{\partial, \rec}$, Eq.~\eqref{eq:KDpartial-L0} indicates that one can regard $K_{D_\partial}$ as a realization of a local Hamiltonian of a 1+1D chiral CFT in $\CW$. The $h_s^{\partial, \rec}$ is the Hamiltonian density, which, in the regime $|\partial D| \gg \xi$, can be identified with the stress-energy tensor $T(x)$ of the chiral CFT. That is, 
\begin{align}
     K_{D_\partial} = \sum_{s \in D_\partial} h_s^{\partial, \rec} \eqover{|\partial D| \gg \xi} \beta \int_{\partial D} dx \mathfrak{h}(x) \xrightarrow{\text{restrict to }\CW} \beta \int_{\partial D} dx T(x), 
\end{align}
where the conversion to an integral follows from the usual procedure of obtaining an integral from a Riemann sum. The prefactor $\beta$ can be decided by the Li-Haldane conjecture. Because $\mathrm{spec}(K_D) = \frac{\xi}{|\partial D|} \mathrm{spec}(L_0) + \text{const}$, and because the Virasoro zero mode $L_0$ for a CFT living on $\partial D$ satisfies
\begin{align}
     L_0 = \frac{|\partial D|}{4 \pi^2}\int_{\partial D} dx T(x)~,
\end{align}
one can identify $\beta = \xi/(4\pi^2)$. 

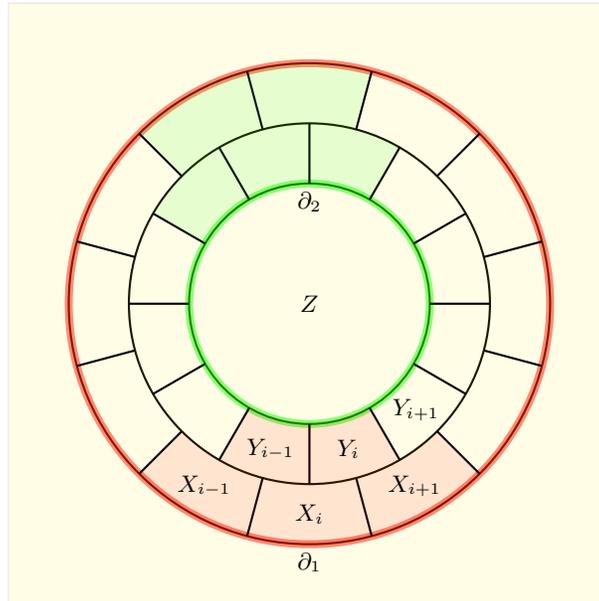
\begin{figure}[htb]
    \centering
    \begin{tikzpicture}[scale = 0.8]
	\def\rone{2cm}
	\def\rtwo{3cm}
	\def\rthree{4cm}
    \draw[thick] (0,0) circle (\rone);
	\draw[color = green, line width = 3, opacity=0.5] (0,0) circle (\rone);
	\draw[thick] (0,0) circle (\rtwo);
    \draw[thick] (0,0) circle (\rthree);
	\draw[color = red, line width = 3, opacity=0.5] (0,0) circle (\rthree);
	\draw[fill=yellow, opacity=0.1] (-5,-5) -- (5,-5) -- (5,5) -- (-5,5) -- cycle;
	\foreach \x in {1,...,12}
	    \draw[thick] (\x*30:\rone)-- (\x*30:\rtwo);

	\foreach \x in {1,...,12}
		\draw[thick] (\x*30+15:\rtwo)-- (\x*30+15:\rthree);

    \draw[fill = red, opacity=0.1] (270-45:3) -- (270-45:4) arc (270-45:270+45:4) -- (270+45:3) arc (270+45:270+30:3) -- (270+30:2) arc (270+30:270-30:2) -- (270-30:3) arc (270-30:270-30-15:3); 

    \draw[fill = green, opacity=0.1] (90-30:2) arc (90-30:90+60:2) -- (90+60:3) arc (90+60:90+45:3) -- (90+45:4) arc (90+45:90-15:4) -- (90-15:3) arc (90-15:90-30:3) -- cycle;

	\node at (270-15:2.5) {$Y_{i-1}$};
	\node at (270+15:2.5) {$Y_{i}$};
	\node at (270+15+30:2.5) {$Y_{i+1}$};

	\node at (270-30:3.5) {$X_{i-1}$};
	\node at (270+0:3.5) {$X_{i}$};
	\node at (270+30:3.5) {$X_{i+1}$};

    \node at (0,0) {$Z$};

    \node at (270:4.3) {$\partial_1$}; 
    \node at (90:1.7) {$\partial_2$}; 
    \end{tikzpicture}
	\caption{Decomposition of an annulus $XY$ in the bulk. $X = X_1X_2\cdots X_N,Y = Y_1Y_2\cdots Y_N$. The red and green regions are examples of $\hat{\Delta}^{\partial_1}$ and $\hat{\Delta}^{\partial_2}$ for the outer and inner boundary.}
	\label{fig:annulus-decomp}
\end{figure}

{\bf Argument of Hypotehsis~\ref{hypo:op-bulk-edge} from dimensional reduction.}
We can also try to understand $K_{D_{\partial}}$ by dimensional reduction. The 1d system is obtained from a 2d cylinder (annulus). Consider an annulus $XY$ in the bulk shown in Fig.~\ref{fig:annulus-decomp} and a state $\sigma_{XY}$ defined as
\begin{align}\label{eq:sigma-XY}
    \sigma_{XY} = \sum_{\mathfrak{a}} \frac{d_{\mathfrak{a}}^2}{\CD^2} \rho^{\mathfrak{a}}_{XY}, \quad \rho^{\mathfrak{a}}_{XY} = 
    \begin{tikzpicture}[baseline={([yshift=-.5ex]current bounding box.center)}, scale=0.6]
    \fill[fill = yellow!10, draw = blue, thick] (0,0) circle (1.6); 
    \fill[fill = white, draw = blue, thick] (0,0) circle (0.8); 
    \draw[dashed] (0,0) circle (1.2);
    \draw[decorate, decoration=snake, thick, gray] (0,0) -- (2.4,0);
    \fill[red] (0,0)  node[above] {${\mathfrak{a}}$} circle (0.1);
    \fill[red] (2.4,0) node[above] {$\bar{\mathfrak{a}}$} circle (0.1); 
    \node[font = \footnotesize] at (180:1) {$Y$}; 
    \node[font = \footnotesize] at (180:1.4) {$X$}; 
  \end{tikzpicture}~~.
\end{align}
Here $\rho^{\mathfrak{a}}_{XY}$ is the reduced density matrix from the state $\ket{\Psi_{\mathfrak{a}}}$ with an anyon string operator passing thought the annulus $XY$. $d_{\mathfrak{a}}$ is the quantum dimension of the anyon, and $\CD = \sqrt{\sum_{\mathfrak{a}} d_{\mathfrak{a}}^2}$ is the total quantum dimension. The state $\sigma_{XY}$ is the maximum entropy state in the information convex set of the annulus $XY$ \cite{Shi_2020}. For any states which are convex combinations of $\{\rho^{\mathfrak{a}}_{XY}\}$, their reduced density matrix on any disk inside $XY$ are the same. 
As pointed out in \cite{Kato:2018bzi}, the entanglement Hamiltonian $K_{XY}^{\sigma}$ of $\sigma_{XY}$ is a sum over local operators. Moreover, one can show that 
\begin{align}\label{eq:K-sigma}
    K_{XY}^{\sigma} =  \sum_i \(\frac{1}{2} \overline{\hat{\Delta}^{\partial_1}_{X_i}} + \sum_i \frac{1}{2} \overline{\hat{\Delta}^{\partial_2}_{Y_i}} \)
\end{align}
where $X = X_1 X_2 \cdots, Y = Y_1 Y_2 \cdots$ is shown in Fig.~\ref{fig:annulus-decomp}. The derivation is similar as in App.~\ref{app:Hrec}. Notice the first sum in the R.H.S. of Eq.~\eqref{eq:K-sigma} is exactly $K_{D_{\partial}}$ with $D = XYZ$ and the second sum, denoted as $\overline{K_{D_{\partial}}}$, is just $K_{D_{\partial}}$ flipped upside down and hence has the opposite chirality of $K_{D_{\partial}}$. Based on the known Li-Haldane conjecture, each entanglement Hamiltonian $K^{\mathfrak{a}}_{XY}$ from $\rho_{XY}^{\mathfrak{a}}$ shall encode the spectrum of the non-chiral CFT in the sector $\CV_{\mathfrak{a}} \otimes \overline{\CV}_{\mathfrak{a}}$. Moreover, based on \cite{Shi_2020}, different $\rho^{\mathfrak{a}}_{XY},\rho^{\mathfrak{b}}_{XY}$ are in orthogonal subspaces and hence the sum in Eq.~\eqref{eq:sigma-XY} can be thought of as a direct sum. As a result, $K^{\sigma}_{XY} = \bigoplus_{\mathfrak{a}} K_{XY}^{\mathfrak{a}}$ shall have the spectrum of the non-chiral CFT with a diagonal form, i.e. in the Hilbert space $\CH_{\text{CFT}} = \bigoplus_{\mathfrak{a}} \CV_{\mathfrak{a}} \otimes \overline{\CV}_{\mathfrak{a}}$. If we regard $X_i Y_i$ as a single site, then $\sigma_{XY}$ can be thought of as a 1+1D CFT thermal state and the entanglement Hamiltonian $K_{XY}^{\sigma}$ can be regarded as the CFT Hamiltonian multiplied by an inverse temperature. That is $K_{XY}^{\sigma} = \beta H_{CFT}^{1d}$. The decomposition in Eq.~\eqref{eq:K-sigma} indicates that, one can regard $h_{i} = \frac{1}{2} \overline{\hat{\Delta}^{\partial_1}_{X_i}}, \bar{h}_{i} = \frac{1}{2} \overline{\hat{\Delta}^{\partial_2}_{Y_i}}$ as the Hamiltonian density for chiral and anti-chiral modes. 
If $|\partial D| \gg \xi$, the sum can also be written as an integral and one can obtain 
\begin{align}
    \sum_i h_i \xrightarrow{|\partial D|\gg \xi} \beta \int_{\partial D} dx T(x), \quad \sum_i \bar{h}_i \xrightarrow{|\partial D|\gg \xi} \beta \int_{\partial D} dx \bar{T}(x)
\end{align} 
where $T(x),\bar{T}(x)$ is identified with the holomorphic and anti-holomorphic stress-energy tensor for the chiral and anti-chiral modes of the CFT\footnote{We comment that in the lattice Hilbert space $\CH_{XY}$, $h_i$ and $\bar{h}_i$ do not commute. However, in the subspace that is the support of $\sigma_{XY}$, they indeed commute. As we mentioned earlier, the support of $\sigma_{XY}$ is diagonal and corresponds to the low-lying subspace of $K_{XY}^{\sigma}$. This is consistent with the fact that $T(x)$ and $\bar{T}(x)$ commute in the non-chiral CFT Hilbert space.}.  

There are also some subtle details of this hypothesis: 

First of all, what is the role of $\mathbb{V}$? It in fact has two functionalities: one is to remove the kernel in $K_{D_{\partial}}$ and the other is to select the CFT sectors. Below we explain these two points in details. (1) One might be tempted to study operator $K_{D_{\partial}}$ on its own. For example, one might expect the spectrum of $K_{D_{\partial}}$ itself matches the full spectrum of the edge CFT in $\CH_{\chi \CFT} = \oplus_{\mathfrak{a}} \CV_{\mathfrak{a}}$. This is incorrect. As we explained earlier, in the Hilbert space $\CH_{D_{\partial}}$, one can always construct the anti-chiral component. Hence, in such a Hilbert space, the operator $K_{D_{\partial}}$ is better to be understood as $\int dx  (\beta T(x) + 0 \bar{T}(x))$. That is, there is a huge kernel in $K_{D_{\partial}}$, which will produce a huge degeneracy. Therefore, one cannot regard $K_{D_\partial}$ as a realization of a 1+1D chiral CFT Hamiltonian \emph{in the Hilbert space $\CH_{\chi \CFT}$}. This is also consistent with the no-go theorem \cite{Hellerman:2021fla} which says a realization of a 1+1D chiral CFT Hamiltonian in a tensor product Hilbert space is impossible. This is in fact also the reason why having an isometry $\mathbb{V}$ in Eq.~\eqref{eq:hypo-bulk-edge-1} is necessary. The isometry plays the role of removing such kernel and degeneracy. (2) Moreover, followed from the cylindrical picture, $K_{D_\partial}$ itself has non-trivial support in all the anyon sectors. Within the scope of this paper, we only study $K_{D_\partial}$ acting on a state $\ket{\Psi}$ without anyon in $D$, hence, the second functionality of the isometry is to keep the support of $K_{D_\partial}$ only in the vacuum sector. 

Secondly, one might wonder can we regard $h_s^{\partial,\rec}$ for $s \in D_{\partial}$ as $T(x)$ up to a prefactor $\beta$? The answer is roughly yes, but not exactly. Indeed, suppose $|\partial D| \gg \xi$, then one can regard $h_s^{\partial,\rec}/\beta = \mathfrak{h}(x)$ as a local operator with $x = s \xi$ and write the sum into an integral. However, $\mathfrak{h}(x)$ constructed this way has a support in a region centered at $x$ and with a range $\xi$. Intuitively, this is similar to a ``smeared'' stress-energy tensor up to a prefector, i.e. $\beta \int dx f(x) T(x)$, where $f(x)$ has a compact support at $x$ with a range $\xi$. Roughly, the reason is that if one were to study the correlation function such as $\langle \mathfrak{h}(x) \mathfrak{h}(y)\rangle$, it does not have the singular behavior as in $\langle T(x) T(y)\rangle$ when $x \to y$. We leave a more detailed explanation of this point in the future work. 

Lastly, in a system with finite local dimension, the hypothesis~\ref{hypo:op-bulk-edge} is up to finite size error. The Hilbert space $\CH_{\CFT}$ for the isometry $\mathbb{V}$ is only the part spanned by the low-lying states of the edge CFT. The finite size error is expected to be reduced as the size of the entanglement boundary $|\partial D|$ increases.  

\subsection{Reconstructed Hamiltonian}
\label{subsec:Hrec}
As we mentioned earlier, it is plausible that one can regard the decomposition of $K_D$ in Eq.~\eqref{eq:KD-hrec} as a local Hamiltonian on the disk $D$. The reason is based on Li-Haldane conjecture, or more generally, a spectrum bulk/edge correspondence, namely, the low-lying spectrum of $K_D$ matches the low-lying spectrum of the edge degrees of freedom up to an overall rescaling and a shift. Assuming the state is close to the zero correlation RG fixed point, so that \Aone condition is applicable, the spectrum of $K_{D_{\inte}} + K_{D_{\partial}}$ will approximately equal that of $K_D$ and hence have the same low-lying spectrum of the original Hamiltonian for $\ket{\Psi}$ up to an overall rescaling and shift. 

Let us now discuss $h_s^{\rec}$ as a Hamiltonian density\footnote{For states satisfying exact \Aone, a commuting projector parent Hamiltonian was reconstructed from the state in \cite{Kim:2024amo}.}. We can assume the system is on a closed manifold, so that we do not need to consider the boundary term. From the gapped groundstate $\ket{\Psi}$ on such a manifold with a triangulation, one can construct 
\begin{align}
    \label{eq:Hrec-Delta-Euler}
    H_{\rec}^{\Delta, \euler} &= \sum_s h_s^{\rec}. 
\end{align}
Moreover, there are other possible choices: 
\begin{align}
    H_{\rec}^{\euler} &= \sum_f K_f - \sum_e K_e + \sum_s K_s \label{eq:Hrec-Euler} \\ 
    H_{\rec}^{\Delta} &= \sum_s \overline{\hat{\Delta}_s} \label{eq:Hrec-Delta}. 
\end{align}
In a gapped groundstate that approximately satisfies \Aone, we have $H_{\rec}^{\Delta, \euler} \approx H_{\rec}^{\euler}$\footnote{For a state that satisfies the exact Markov decomposition condition Eq.~\eqref{eq:Markov-decompose}, this is an equality. 
In a lattice model, where the reduced density matrices generally have kernels, the exact decomposition is violated by terms associated with these kernels; the action of the two forms of $H_\text{rec}$ on the complement of the kernel, which is the low-lying spectrum, will still agree, as we have verified numerically.}.
When we evaluate their expectation value in such a state, they both equal $- \chi \gamma$, where $\chi$ is the Euler character of the manifold. $H_{\rec}^{\Delta}$ is obtained from $H_{\rec}^{\Delta,\euler}$ by removing the averaging over the operator TEE term. If the system is on a torus, where one can coarse-graine the sites such that the supersites form a regular triangular lattice, then $H_{\rec}^{\Delta} = H_{\rec}^{\Delta, \euler}$. This simply follows from the definition and does not rely on any property of the state.

The merits of $H_{\rec}^{\Delta, \euler}$ and $H_{\rec}^{\euler}$, for the purpose of the study of bulk/edge correspondence, is that their local density terms are both equal to those obtained from the Markov decompositions of an entanglement Hamiltonian of a disk, if the state satisfies \Aone. This connection is true for any triangulation of the disk. Hence if we want to study $H_{\rec}^{\Delta, \euler}$ and $H_{\rec}^{\euler}$ on some open manifold, one can infer the spectrum with bulk/edge correspondence. 

The merits of $H_{\rec}^{\Delta}$ is that it is a positive operator for any input state because of the operator weak monotonicity theorem, $\hat{\Delta} \geq 0$ \cite{Lin:2022jtx}. The groundstate of $H_{\rec}^{\Delta}$ with zero energy is guaranteed to satisfy \Aone. Hence, one can conclude that, if a groundstate of $H_{\rec}^{\Delta}$ has zero energy, then it is a gapped state. This operator being positive also enables us to run gradient descent with an objective function $g(\ket{\psi}) = \langle \psi | (H_{\rec}^{\Delta})^{\psi} | \psi \rangle$, which will reduce the error of \Aone and bring the state closer the zero-correlation-length fixed point. 

We conjecture that, when the input state $\ket{\psi}$ is close to a zero-correlation-length RG fixed point, then all $(H_{\rec})^{\psi}$ are local Hamiltonians within the same gapped phase and their groudstates have a smaller correlation length. Explicitly, here we use the violation of \Aone as a characterization of the correlation length. With this conjecture, one can design a process, namely ``find groundstate -- reconstruct Hamiltonian -- find groundstate -- reconstruct Hamiltonian -- find groundstate -- ...'', to progressively drive the state closer to the fixed point, where the UV contributions to entanglement quantities are reduced. We will give an argument for this statement in Section~\ref{subsec:error-reduction}, and numerically verify and utilize this thereafter. 

\section{Entanglement of corners}

\label{sec:hypothesis}

The previous discussion takes place on length scales much larger than the correlation length $\xi$, such that \Azero and \Aone are applicable. In this section, we are going to discuss the entanglement of corner regions, where the typical linear size is smaller than the correlation length. We note that in order to discuss any possible universal properties of the corner entanglement, the correlation length $\xi$ has to be larger than the lattice spacing, so that there are enough degrees of freedom around a corner to emerge any universal properties. 
Roughly speaking, the physical picture of a corner region $A$ (defined below) behind the results in this section is that one can extrapolate the Li-Haldane conjecture to the limit $\xi/|\partial A| \to \infty$, so that the state $\rho_A$ can be approximately described by the CFT groundstate.

\subsection{Hypothesis about entanglement of corners}

We define a \emph{corner region} as a disk-like region whose typical linear size is much less than $\xi$ but still much larger than the lattice spacing. When one studies entanglement of a disk $A$, whose entanglement boundary $\partial A$ contains sharp corners, then one needs to take the corner entanglement into consideration. By a \emph{sharp corner}, we mean a location on the entanglement boundary $\partial A$ whose radius of curvature is much less than the correlation length. As we mentioned earlier, the contribution in $K_A$ from the degrees of freedom along the entanglement boundary $\partial A$ is within an annulus with thickness $\xi$. Therefore, if there is a sharp corner in $\partial A$, then it is going to wrap around a corner region as shown in Fig.~\ref{fig:corner-region}. 

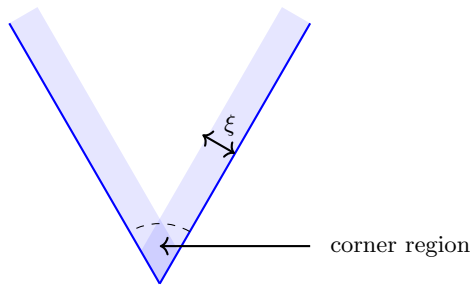
\begin{figure}[htb]
    \centering
      \begin{tikzpicture}
    \draw[thick, color = blue] (0,0) -- (60:4);
    \draw[thick, color = blue] (0,0) -- (120:4); 
    \fill[color = blue, opacity = 0.1] (0,0) -- (120:4) -- ($(120:4) + (30:{0.5 / 2 * sqrt(3)})$) -- ($(0,0)+(60:0.5)$); 
    \fill[color = blue, opacity = 0.1] (0,0) -- (60:4) -- ($(60:4) + (150:{0.5 / 2 * sqrt(3)})$) -- ($(0,0)+(120:0.5)$); 
    \draw[->, thick] (2,0.5) -- (0,0.5); 
    \draw[dashed] ($(0,0)+(60:0.8)$) arc (60:120:0.8);
    \node at (3.2, 0.5) {corner region}; 
    \draw[<->, thick] (60:2) -- ($(60:2) + (150:0.5)$);  
    \node at ($(60:2.3) + (150:0.25)$) {$\xi$}; 
  \end{tikzpicture}
	\caption{A corner region near a sharp corner along the entanglement boundary (blue line).}
	\label{fig:corner-region}
\end{figure}

The operator content in $K_A$ supported on such a region will be of a different form from those near a smooth segment along the entanglement boundary. Explicitly, one can first apply the derivation in Section~\ref{subsec:local-ent-hmt} and obtain a decomposition $K_A = K_{A_{\text{int}}} + K_{A_{\partial}}$ where $K_{A_\inte}$ and $K_{A_\partial}$ take the same form as sums of $h_{s}^{\rec}$ and $h_{s}^{\partial,\rec}$ respectively. However, the operator $h_{s}^{\partial,\rec}$ for the site $s$ along a smooth segment and the one at the sharp corner $c$ will be different [Fig.~\ref{fig:Deltas-smooth-corners}].

For smooth segments, which we define as segments where the radius of curvature is much larger than $\xi$ everywhere, one can zoom in such that the segment is flat. Therefore, the local operator, e.g. $h_{s_1}^{\partial,\rec},h_{s_2}^{\partial,\rec}$ in Fig.~\ref{fig:Deltas-smooth-corners}, are just related by a translation along $\partial A$. If the system is uniform, then these two operators are the same but just with a different support of location, while near the sharp corner $c$, the $h_s^{\partial, \rec}$ takes a different form. Therefore, we need to treat it differently. Below we give a regulation of the corner region that highlights the universal properties of the corner entanglement. 

\begin{figure}[htb]
    \centering
    \begin{tikzpicture}
    \filldraw[thick, draw = blue, fill = yellow!10] (0,0) -- (180-30:3) arc [start angle=180-30-90, end angle=180+30+90, radius={3*tan(30)}] -- cycle; 
    \node at (-3,0) {$A$};  

    \draw[dashed] ({-3/cos(30)},{-3*tan(30)}) circle (0.5);
    \fill[color = blue] ({-3/cos(30)},{-3*tan(30)}) circle (0.08);
    \node at ({-3/cos(30)},{-3*tan(30)-0.3}) {$s_1$};

    \draw[dashed] ({-3/cos(30)-3*tan(30)},0) circle (0.5);
    \fill[color = blue] ({-3/cos(30)-3*tan(30)},0) circle (0.08);
    \node at ({-3/cos(30)-3*tan(30)-0.3},0) {$s_2$};

    \draw[dashed] (0,0) circle (0.5);
    \fill[color = red] (0,0) circle (0.08);

    \begin{scope}[xshift = -3cm, yshift = -3.2cm, scale = 0.3]
    \draw[color = blue, thick] (0,0) -- (5, 0); 
    \draw[thick] (1,0) -- (1,1) -- (4,1) -- (4,0);
    \draw[thick] (1.5,1) -- (1.5,2) -- (3.5,2) -- (3.5,1);
    \draw[thick] (2,0) -- (2,1);
    \draw[thick] (3,0) -- (3,1); 
    \draw[thick] (4,0) -- (4,1); 
    \draw[thick] (2.5,1) -- (2.5,2); 
    \draw[fill=yellow, opacity=0.1] (0,0) -- (5,0) -- (5,3) -- (0,3) -- cycle;
    \fill[color = blue] (2.5,0) circle (0.2);
    \node at (2.5, -1) {$s_1$};
    \end{scope}

    \begin{scope}[xshift = -6cm, yshift = -1cm, scale = 0.3,rotate = -90]
    \draw[color = blue, thick] (0,0) -- (5, 0); 
    \draw[thick] (1,0) -- (1,1) -- (4,1) -- (4,0);
    \draw[thick] (1.5,1) -- (1.5,2) -- (3.5,2) -- (3.5,1);
    \draw[thick] (2,0) -- (2,1);
    \draw[thick] (3,0) -- (3,1); 
    \draw[thick] (4,0) -- (4,1); 
    \draw[thick] (2.5,1) -- (2.5,2); 
    \draw[fill=yellow, opacity=0.1] (0,0) -- (5,0) -- (5,3) -- (0,3) -- cycle;
    \fill[color = blue] (2.5,0) circle (0.2);
    \node at (2.5, -1) {$s_2$};
    \end{scope}

    \begin{scope}[xshift= 0.5cm, yshift = -1.2cm, scale = 0.3, rotate = 90]
    \draw[thick] (90-30:2) arc (90-30:90+30:2); 
    \draw[thick] (90-30:3) arc (90-30:90+30:3); 
    \draw[thick, color = blue] (90-30:4) -- (0,0) -- (90+30:4); 
    \draw[thick] (90-15:2) -- (90-15:3); 
    \draw[thick] (90-30:2) -- (90-30:3);
    \draw[thick] (90:2) -- (90:3);
    \draw[thick] (90+30:2) -- (90+30:3);
    \draw[thick] (90+15:2) -- (90+15:3); 
    \fill[color = yellow, opacity = 0.1] (0,0) -- (90-30:4) arc (90-30:90+30:4) -- cycle; 
    \fill[color = red] (0,0) circle (0.2);
    \node at (-1, 0) {$c$}; 
    \end{scope}
  \end{tikzpicture}
	\caption{$h_{s}^{\partial,\rec}$ along the entanglement boundary. In the smooth part, such as at locations $s_1,s_2$, because the radius of curvature is much larger than $\xi$, one can zoom in such that the segment is flat. If the state has translation symmetry, then $h_{s_1}^{\partial,\rec},h_{s_2}^{\partial,\rec}$ are just related by a translation along the entanglement boundary. However, $h_{c}^{\partial,\rec}$ at the sharp corner is a different operator.}
	\label{fig:Deltas-smooth-corners}
\end{figure}
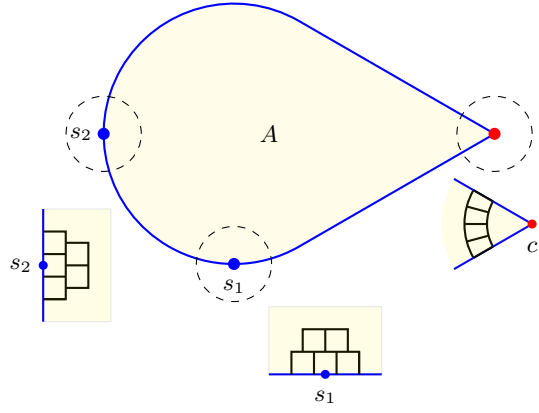

\begin{hypothesis}
    \label{hypo:corner}
    Consider a disk $A$ with a sharp corner in $\partial A$. We can decompose $\partial A = (\partial A)_{\smooth} \cup (\partial A)_{\corner}$ into the smooth part $(\partial A)_{\smooth}$ and corner part $(\partial A)_{\corner}$. We hypothesize that there exists an isometry $\mathbb{V}: \CH^{\chi \CFT}_{(\partial A)_{\smooth}}\otimes \CH^{\chi CFT}_{\text{circle}} \to \CH_{A_{\partial}}$ (i.e. $\mathbb{V}^{\dagger}\mathbb{V} = \mathbbm{1}$) such that
    \begin{align} \label{eq:Kpartial-hypo}
         K_{A_{\partial}} \eqover{\mathbb{V}}  \xi \int_{x \in (\partial A)_{\smooth}} dx T(x) + K^{\chi CFT}_{a}  ,~~~
        \ket{\Psi} = \mathbb{V} \ket{\varphi}
    \end{align}
    where $K^{\chi CFT}_{a}$ is the entanglement Hamiltonian from the edge CFT groundstate of an interval, $a$, with opening angle equal to the angle of the sharp corner. 
    Here $\ket{\varphi}$ behaves as a chiral CFT thermal double state in $(\partial A)_{\smooth}$ and the chiral CFT groundstate near the corner.  
    Explicitly, this means that [Fig.~\ref{fig:hypo-corner}]: 
    \begin{enumerate}
        \item $\langle \varphi| T(x)T(y)\cdots | \varphi \rangle = \Tr [\rho_{(\partial A)_{\smooth}} T(x)T(y)\cdots], x,y,\cdots \in (\partial A)_{\smooth}$ with the reduced density matrix $\rho_{(\partial A)_{\smooth}}$ from a chiral CFT thermal state. 
        \item $\langle \varphi| K_a^{\chi \CFT} K_b^{\chi\CFT} \cdots | \varphi \rangle = \langle \Omega|K_a^{\chi \CFT} K_b^{\chi\CFT} \cdots | \Omega\rangle_{\text{circle}}, a,b,\cdots \in \text{circle}$, with the chiral CFT groundstate $|\Omega\rangle_{\text{circle}}$ on the circle.
    \end{enumerate}
\end{hypothesis}
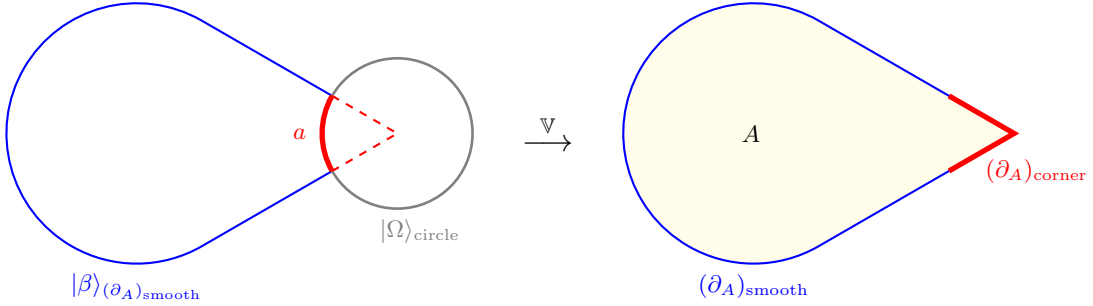
\begin{figure}[htb]
    \centering
    \begin{tikzpicture}
    \draw[thick, draw = blue] (0,0) -- (180-30:3) arc [start angle=180-30-90, end angle=180+30+90, radius={3*tan(30)}] -- cycle; 
    \filldraw[draw = gray, fill = white, line width = 1 pt] (0,0) circle (1);
    \draw[color = red, thick, dashed] (180-30:1) -- (0,0) -- (180+30:1); 
    \draw[color = red, thick, line width = 2 pt] ($(0,0) + (180-30:1)$) arc (180-30:180+30:1); 
    
    \node at ({-3/cos(30)},{-3*tan(30)-0.3}) {\color{blue} $|\beta\rangle_{(\partial_A)_{\text{smooth}}}$};
    \node at (0.3,-1.3) {\color{gray} $|\Omega\rangle_{\text{circle}}$};
    \node at (-1.3,0) {\color{red} $a$};

    \node[font = \large] at (2,0) {$\overset{\mathbb{V}}{\longrightarrow}$};

    \begin{scope}[xshift = 8.2cm]
      \filldraw[thick, draw = blue, fill = yellow!10] (0,0) -- (180-30:3) arc [start angle=180-30-90, end angle=180+30+90, radius={3*tan(30)}] -- cycle; 
      \node at (-3.5,0) {$A$};  
      \node at ({-3/cos(30)},{-3*tan(30)-0.3}) {\color{blue} $(\partial_A)_{\text{smooth}}$};
      \draw[color = red, line width = 2 pt] (180-30:1) -- (0,0) -- (180+30:1); 
      \node at (0.3,-0.5) {\color{red} $(\partial_A)_{\text{corner}}$};
    \end{scope}
  \end{tikzpicture}
	\caption{Regulating the corner region with a hole. The boundary CFT groundstate lives on the boundary of the hole.}
	\label{fig:hypo-corner}
\end{figure}

The essence of this Hypothesis~\ref{hypo:corner} is that one can regard a disk $D$ in the bulk whose size is much less than $\xi$ as an empty hole with a gapless edge. That is, once we zoom in to small length scales, one can view a chiral gapped groundstate as a ``gas of holes''. This picture is the basis of the familiar coupled-wire construction of chiral states \cite{coupled-wires, neupert2014wire, teo2014luttinger, meng2015coupled, Qi:2012ngv}, and more directly, of the construction of chiral states \cite{Brehm:2021wev, Sopenko:2023utk,Cheng:2023kxh} in terms of a lattice of chiral CFTs on the circle. The idea is that the universal properties of the wavefunction are unchanged by removing a small-enough disk;
its boundary hosts a CFT on the small circle bounding the hole, whose spectrum is gapped with a level-spacing of the inverse radius of the hole.  For a small-enough hole, the CFT is well-approximated by its groundstate.  

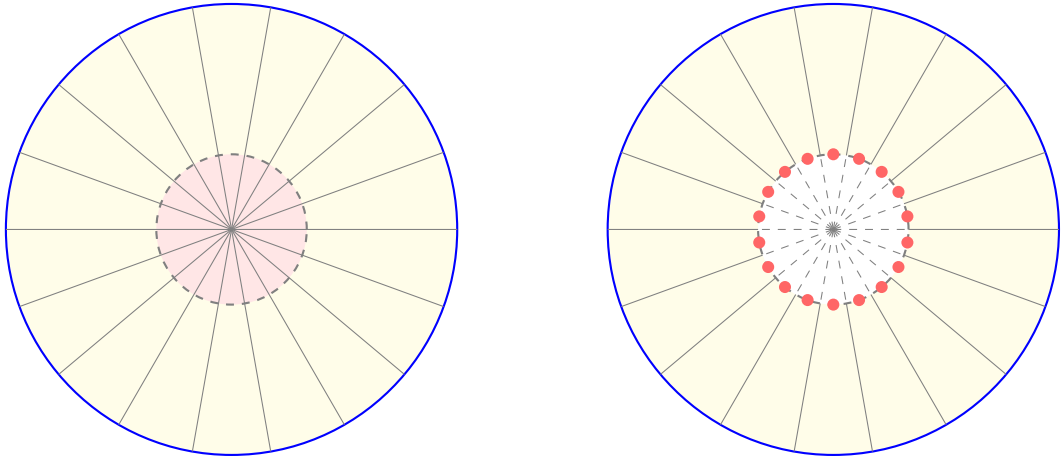
\begin{figure}[htb]
  \begin{tikzpicture}
    \filldraw[color = yellow!10, draw = blue, thick] (0,0) circle (3);
    \filldraw[color = red!10, draw = gray, thick, dashed] (0,0) circle (1); 
    \foreach \x in {1,...,18}
		  \draw[gray] (0,0) -- (20*\x : 3); 
    \begin{scope}[xshift = 8cm]
      \filldraw[color = yellow!10, draw = blue, thick] (0,0) circle (3);
      \filldraw[color = white, draw = gray, thick, dashed] (0,0) circle (1); 
      \foreach \x in {1,...,18}{
		    \draw[gray, dashed] (0,0) -- (20*\x :1);
        \draw[gray] (20*\x :1) -- (20*\x : 3); 
        \fill[red!60] (10+20*\x : 1) circle (0.08); 
      }
    \end{scope}
  \end{tikzpicture}
  \caption{Making a hole by ``blocking''. In the left figure, we divide a small disk $D$ shown by the red region into several pie-shape regions. Then we block each ``pie slice'' into a single site shown by the red dot in the right figure.}
  \label{fig:corner-blocking}
\end{figure}

This picture of ``gas of holes'' can be understood in terms of a blocking procedure and Li-Haldane  conjecture. Consider a disk $D$ of radius $R_D$ shown in Fig.~\ref{fig:corner-blocking} (left). We assume $R_D$ is much larger than the lattice spacing but smaller than the correlation length. We divide the disk $D$ into a set of pie-slices, then block each pie slice into a single site living on $\partial D$. Then we apply Li-Haldane conjecture, namely, $\rho_D = e^{ - { \xi \over R_D } L_0 }$, where $L_0$ is an operator that encodes the spectrum of the chiral CFT with integer spacing. When the radius of the disk $D$ is small enough ($R_D/\xi \ll 1$), 
the reduced density matrix is effectively the projector onto the groundstate of the boundary CFT. The fact that the entanglement between $D$ and its complement decreases rapidly as $e^{- 1/R_D}$ when we shrink the region is why we can freely remove small holes. 

To use this picture to understand the corner entanglement structure, we simply use our freedom to make holes to replace the tip of the corner with a hole of radius $\epsilon \sim \xi$ as shown in Fig.~\ref{fig:hypo-corner}. 
Now this part of the entanglement boundary of $D$ is replaced by an actual gapless boundary, hosting the edge CFT in its groundstate. This contributes the familiar $K^{\chi\CFT}_a$ entanglement Hamiltonian of an interval in a CFT groundstate.  

A second argument for this `gas of holes' picture is based on the fact that the low-energy theory of a chiral gapped state is a (very special) 2+1d conformal field theory.  One of the elements of the conformal group is the inversion transformation.  
By an inversion, we can map an arbitrarily small hole to a physical gapless boundary taken to be arbitrarily large and smooth. We have a good understanding of the entanglement structure of such a gapless edge \cite{Vir,kim2024conformalgeometryentanglement} that can now be applied to the neighborhood of the hole.

Both of these arguments imply that one can just regard the corner region as part of a gapless physical boundary as illustrated in Fig.~\ref{fig:hypo-corner}. Therefore, one can apply all of the technology about the entanglement of the gapless edge in \cite{Vir,kim2024conformalgeometryentanglement} to the neighborhood of a small hole, which motivates the logical framework described in the next subsection.

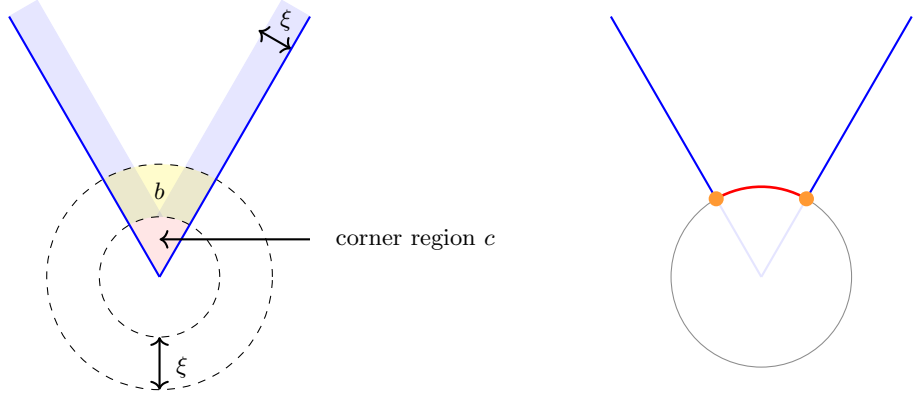
\begin{figure}[htb]
    \centering
    \begin{tikzpicture}
    \fill[color = blue, opacity = 0.1, draw = blue] (0,0) -- (120:4) -- ($(120:4) + (30:{0.5 / 2 * sqrt(3)})$) -- ($(0,0)+(60:0.5)$); 
    \fill[color = blue, opacity = 0.1, draw = blue] (0,0) -- (60:4) -- ($(60:4) + (150:{0.5 / 2 * sqrt(3)})$) -- ($(0,0)+(120:0.5)$); 
    \draw[<->, thick] (60:3.5) -- ($(60:3.5) + (150:0.5)$); 
    \node at ($(60:3.8) + (150:0.25)$) {$\xi$}; 

    \fill[color = red!10] (0,0) -- (60:0.8) arc (60:120:0.8) -- cycle;

    \fill[color = yellow, opacity = 0.2] (60:0.8) arc (60:120:0.8) -- (120:1.5) arc (120:60:1.5) -- cycle;

    \draw[dashed] (0,0) circle (0.8);
    \draw[dashed] (0,0) circle (1.5);
    \draw[thick, color = blue] (0,0) -- (60:4);
    \draw[thick, color = blue] (0,0) -- (120:4); 

    \draw[<->, thick] (-90:0.8) -- (-90:1.5); 
    \node at ($(-90:1.2) + (0:0.3)$) {$\xi$};
    
    \node at (0,1.15) {$b$};

    \draw[->, thick] (2,0.5) -- (0,0.5); 
    \node at (3.4, 0.5) {corner region $c$}; 

    \begin{scope}[xshift = 8cm]
      \draw[thick, color = blue] (0,0) -- (60:4);
      \draw[thick, color = blue] (0,0) -- (120:4);
      \filldraw[color = white, draw = gray, opacity = 0.9] (0,0) circle (1.2);  
      \draw[color = red, line width = 1 pt] (60:1.2) arc (60:120:1.2); 
      \fill[color = orange!80] (60:1.2) circle (0.1); 
      \fill[color = orange!80] (120:1.2) circle (0.1); 
    \end{scope}
  \end{tikzpicture}
	\caption{(Left) The corner is buffered away from the rest of the entanglement boundary. (Right) The junctions, denoted by the two dots, between the entanglement boundary and the hole of CFT.}
	\label{fig:buffer-junction}
\end{figure}

\subsection{Further remarks}
In this subsection, we are going to remark on several possible questions about the Hypothesis~\ref{hypo:corner}. 

Firstly, how to understand the ``hole'' around the corner? Is it a real hole? In a generic chiral gapped state on a lattice, there isn't really a hole around the corner unless one constructs it that way as in \cite{Brehm:2021wev, Sopenko:2023utk,Cheng:2023kxh}. However, the hole is meant to be a regulation scheme of the corner geometry in which the universal entanglement properties of the corner is manifest in a infrared (IR) and continuum limit. Let us unpack this sentence: to begin with, the IR limit (RG fixed point) of a chiral gapped phase has to be described by a system where any local regions contains infinite amount of degrees of freedom, and hence the IR limit is described by a quantum field theory with correlation length $\xi \to 0$. Under such a limit when the system is scale invariant, for a disk of radius $R$, there is only two possible limit $R/\xi \to \infty$ or $R/\xi \to 0$. A corner region, which is defined as a region with linear size less than $\xi$ as we mentioned earlier, is about the second kind. For such a disk, it is described by a pure state $\rho = e^{-\xi / R L_0} \to \ket{0}\bra{0}$, the groundstate of $L_0$ which encodes a chiral CFT spectrum. A corner region is part of such a disk. Furthermore, we wish to capture the universal feature in the entanglement of the corner region as a function of the opening angle. There could be non-universal corner contribution existing as well, but we do not want to include them in our hypothesized picture. Later we shall mention a scheme to remove non-universal contributions so that the hypothesis can be tested in a model not at the IR fixed point. When under such IR limit, there isn't really a difference among the following three scenarios: (1) the hole is there when one constructed the system as in \cite{Brehm:2021wev, Sopenko:2023utk, Cheng:2023kxh}; (2) on the original system, one pokes a hole by tracing out a disk of radius $R/\xi \to 0$; (3) there is no hole, but we imagine there is a hole to obtain entanglement relations about corners. 

Secondly, we remark that one can also regard a point on the smooth part of the entanglement boundary as a corner with angle $\pi$, even thought its contribution is often included in the area law part so that one can set $S_{\corner}(\theta = \pi) = 0$. The corner regulation picture in Hypothesis~\ref{hypo:corner} is necessary when one wants to consistently discuss the entanglement near a point that is divided into more then two pieces (i.e. a junction of entanglement segments). For example, in a T-shape division where one of the entanglement boundary segment is flat, one still needs to regulate it so that one can discuss the other entanglement of the $\pi/2$ corners in the meantime. 

Lastly, we comment on the connecting junctions between the corner and the smooth segment of the entanglement boundary in Fig.~\ref{fig:buffer-junction}. 
A technical question one might ask is: in Hypothesis~\ref{hypo:corner}, why can we study the smooth part of the entanglement boundary and the corner region separately? It is because of the cluster property, more concretely, due to the \Azero condition~\ref{assumption:A0}. Let us consider an entanglement boundary $(\partial A)_{+}$ thickened by the correlation length $\xi$ with a corner region $c$. In Fig.~\ref{fig:buffer-junction} (left), because of the \Azero condition and its consequence Eq.~\eqref{eq:A0}, there is no correlation between $(\partial A)_{+} \setminus (bc)$ and the corner region $c$, namely, the reduced density matrix from $\ket{\Psi}$ is a product state $\rho_{(\partial A)_{+}\setminus b} = \rho_{(\partial A)_{+}\setminus (bc)} \otimes \rho_{c}$. Therefore, as long as we are studying the correlators of operators supported only in $(\partial A)_{+} \setminus (bc)$ or the corner region $c$ (or the inner disk of Fig.~\ref{fig:buffer-junction} (left) that contains $c$), we can treat them individually. As readers will see in the following sections, when we study corner contributions, we indeed manage to design operators which only have non-trivial action on the inner disk that contains corner regions. In the expression for $K_{A_{\partial}}$ (Eq.~\eqref{eq:Kpartial-hypo}) in Hypothesis~\ref{hypo:corner}, indeed there are local operator contributions in the connecting regions $b$. This can be included as the boundary UV operator in $K_a^{\chi \CFT}$ on the junction shown in Fig.~\ref{fig:buffer-junction} (right). Such UV operators also have to be canceled, when we extract universal properties of the corner entanglement, which can be done by taking certain linear combinations of entanglement Hamiltonians with different sharp corners such as those considered in the next section.

\section{Universal properties regarding corner entanglement}
\label{sec:universal-corner-contributions}

Having introduced the Hypothesis~\ref{hypo:corner}, we now use it to derive various universal properties that one can verify on a chiral gapped wavefunction. 

\subsection{Corner entanglement entropies}  

Following from Hypothesis~\ref{hypo:corner}, one can compute the entanglement entropy $S_A = \bra{\Psi}K_A \ket{\Psi}$ and conclude that around the corner region there is a contribution $f(\theta)$ of the form 
\begin{align}\label{eq:f-corner}
     f(\theta) = \frac{c_{\tot}}{6} \log \sin\( \frac{\theta}{2}\) + f_{\text{non-uni}},
\end{align}
where $\theta$ is the opening angle. The derivation is explained below in details. 

With the regulation in Fig.~\ref{fig:hypo-corner}, we replace the sharp corner with a hole and envision the edge CFT groundstate lives on the boundary of the hole. The entanglement from the sharp corner segment $(\partial A)_{\corner}$ is now described by the entanglement of an interval $a$ in the edge CFT groundstate on the circle. More precisely, consider a disk $D$ of radius less than $\xi$ but much larger than the lattice spacing, centered at the tip of the corner. We use the entanglement between $a$ and its complement $\bar{a}$ [Fig.~\ref{fig:corner-to-inteval} (right)] in the edge CFT groundstate on the circle to describe the entanglement between $A_{\corner}$ and $D \backslash A_{\corner}$ [Fig.~\ref{fig:corner-to-inteval} (left)]. This is in line with the picture of blocking [Fig.~\ref{fig:corner-blocking}], that here we block $A_{\corner}$ into $a$ and $D\backslash A_{\corner}$ into $\bar{a}$. Under such a procedure, the entanglement between $A_{\corner}$ and $D \backslash A_{\corner}$ is the same as that between $a,\bar{a}$. 

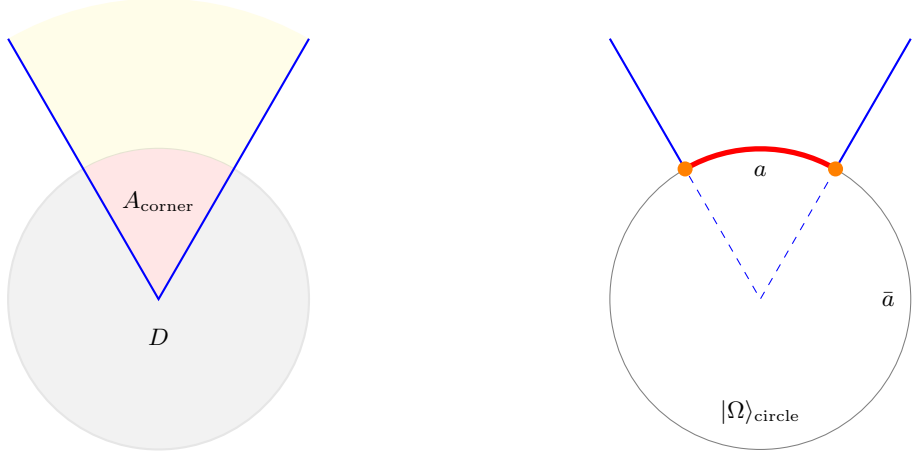
\begin{figure}[htb]
\centering
  \begin{tikzpicture}
    \fill[color = yellow!10] (0,0) -- (90-30:4) arc (90-30:90+30:4) -- cycle; 
    \filldraw[color = gray, opacity = 0.1, draw = black, thick] (0,0) circle (2); 
    \fill[color = red!10] (0,0) -- (90-30:2) arc (90-30:90+30:2) -- cycle;
    \draw[blue, thick] (90+30:4) -- (0,0) -- (90-30:4);  
    \node at (90:1.3) {$A_{\text{corner}}$};
    \node at (90:-0.5) {$D$}; 
    \begin{scope}[xshift = 8cm]
      \draw[blue, thick] (90+30:4) -- (90+30:2); 
      \draw[blue, thick] (90-30:4) -- (90-30:2); 
      \draw[blue, dashed] (90+30:4) -- (0,0) -- (90-30:4);
      \draw[gray] (0,0) circle (2); 
      \draw[red, line width = 2 pt] (90-30:2) arc (90-30:90+30:2); 
      \fill[orange] (90-30:2) circle (0.1);
      \fill[orange] (90+30:2) circle (0.1);
      \node at (90:1.7) {$a$};
      \node at (0:1.7) {$\bar{a}$};
      \node at (90:-1.5) {$|\Omega\rangle_{\text{circle}}$}; 
    \end{scope}
  \end{tikzpicture}
  \caption{We use the entanglement beween $a$ and $\bar{a}$ on the edge CFT groundstate on a circle (right) to describe the entanglement between $A_{\corner}$ and $D\backslash A_{\corner}$ (left).}
  \label{fig:corner-to-inteval}
\end{figure}

Recall that for a CFT groundstate on a circle of length $L$, the entanglement entropy of an interval of length $\ell$ is of the form \cite{Calabrese:2004eu}
\begin{align}\label{eq:S-CFT}
     S(\ell) = \frac{c_{\tot}}{6} \log\( \frac{L}{\pi a'} \sin\( \frac{\pi \ell}{L} \) \) + c'_1,
\end{align}
where $a'$ is the cutoff scale, $c'_1$ is a non-universal constant and $c_{tot} = c + \bar{c}$ is the total central charge. With our regulation scheme, one can obtain Eq.~\eqref{eq:f-corner} by applying Eq.~\eqref{eq:S-CFT} with $\theta = \pi \ell /L$. Here we separate out the angle dependence and combine the non-universal piece into the additive constant $f_{\text{non-uni}}$. In a generic representative wavefunction of a chiral gapped phase, there could be ultraviolet (UV) contributions that are governed by physics on the scale of a few lattice spacings; this is the source of $f_{\text{non-uni}}$. Therefore, in the regulation picture [Fig.~\ref{fig:corner-to-inteval} (right)], it is plausible to assume that $f_{\text{non-uni}}$ depends \emph{only} on the junction points. Such a UV contribution does not have to be uniform everywhere, and therefore $f_{\text{non-uni}}$ could potentially vary with the location of the junctions.  

How to test Eq.~\eqref{eq:f-corner} on a generic chiral gapped wavefunction? We need to single out the corner contribution $f(\theta)$ and remove the non-universal piece in it. For this purpose, we design two linear combinations of entanglement entropies on a ``corner conformal ruler''\footnote{This name is borrowed from \cite{kim2024conformalgeometryentanglement}.} $\mathbb{D} = (A,A',B,C,C')$ in Fig.~\ref{fig:corner-conf-ruler} and derive that
\begin{align}\label{eq:Delta-c-eta}
    \Delta(AA',B,CC') = S_{AA'B} + S_{CC'B} - S_{AA'} - S_{CC'} =  - \frac{c_{\tot}}{6} \ln (\eta_g) \\ \label{eq:I-c-eta}
    I(A:C|B) = S_{AB} + S_{BC} - S_{ABC} - S_{B} =  - \frac{c_{\tot}}{6} \ln (1-\eta_g),
\end{align}
where 
\begin{align}\label{eq:eta-g}
     \eta_g = \frac{l_a l_c}{l_{ab} l_{bc}}
\end{align}
is the geometric cross-ratio, defined with the chord distance $l_{\bullet} =R \sin(\theta_{\bullet}/2)$ as shown in Fig.~\ref{fig:corner-conf-ruler}. One can verify that the area law contribution in the entanglement entropies are canceled in these linear combinations. Moreover, the UV contributions from the scale of few lattice spacing on the junctions in the regulation picture are also canceled. Hence, such relations can capture the universal entanglement entropy contributions from the corner in a robust manner.  

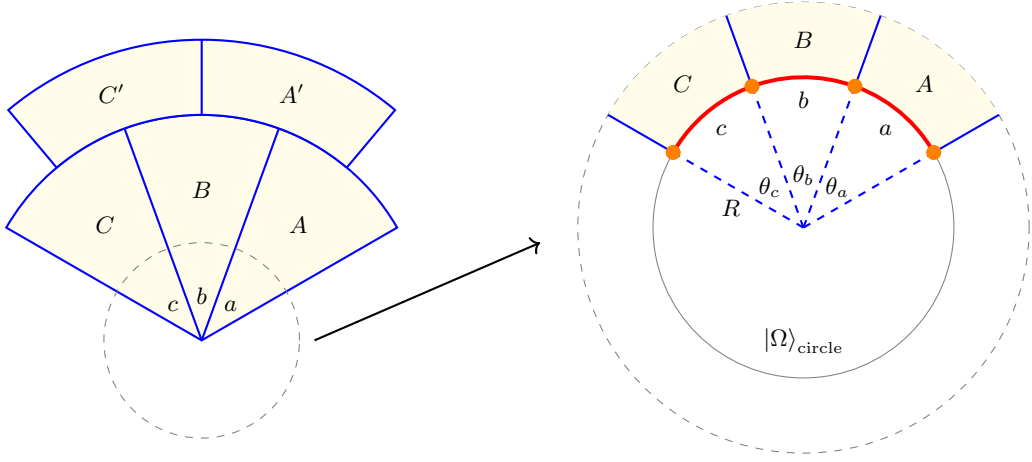
\begin{figure}[htb]
    \centering
    \begin{tikzpicture}
    \filldraw[thick, fill = yellow!10, draw = blue] (0,0) -- (90-60:3) arc (90-60:90+60:3) -- cycle; 
    \filldraw[thick, fill = yellow!10, draw = blue] (90-40:3) -- (90-40:4) arc (90-40:90+40:4) -- (90+40:3) arc (90+40:90-40:3); 
    \draw[thick, color = blue] (0,0) -- (90-20:3); 
    \draw[thick, color = blue] (0,0) -- (90+20:3); 
    \draw[thick, color = blue] (90:3) -- (90:4); 
    \node at (90-40:0.6) {$a$}; 
    \node at (90:0.6) {$b$}; 
    \node at (90+40:0.6) {$c$}; 
    \node at (90-40:2) {$A$}; 
    \node at (90:2) {$B$}; 
    \node at (90+40:2) {$C$}; 
    \node at (90-20:3.5) {$A'$};  
    \node at (90+20:3.5) {$C'$};  

    \draw[gray, dashed] (0,0) circle (1.3);
    \draw[->, thick] (1.5,0) -- (4.5,1.3);

    \begin{scope}[xshift = 8cm, yshift = 1.5cm]
    \draw[gray, dashed] (0,0) circle (3);
    \fill[color = yellow!10] (0,0) -- (90-60:3) arc (90-60:90+60:3) -- cycle;
    \filldraw[color = white, draw = gray] (0,0) circle (2); 

    \draw[blue,thick] (0,0) edge[dashed] (90-60:2) (90-60:2) -- (90-60:3);
    \draw[blue,thick] (0,0) edge[dashed] (90-20:2) (90-20:2) -- (90-20:3);
    \draw[blue,thick] (0,0) edge[dashed] (90+20:2) (90+20:2) -- (90+20:3);
    \draw[blue,thick] (0,0) edge[dashed] (90+60:2) (90+60:2) -- (90+60:3);
    \draw[color = red,line width = 1.5 pt] (90-60:2) arc (90-60:90+60:2);

    \fill[color = orange] (90-60:2) circle (0.1);
    \fill[color = orange] (90-20:2) circle (0.1);
    \fill[color = orange] (90+20:2) circle (0.1);
    \fill[color = orange] (90+60:2) circle (0.1);

    \node at (0, -1.5) {$\ket{\Omega}_{\text{circle}}$};
    \node at (90-40:2.5) {$A$}; 
    \node at (90:2.5) {$B$}; 
    \node at (90+40:2.5) {$C$}; 
    \node at (90-40:1.7) {$a$}; 
    \node at (90:1.7) {$b$}; 
    \node at (90+40:1.7) {$c$};  
    \node at (90-40:0.7) {$\theta_{a}$}; 
    \node at (90:0.7) {$\theta_{b}$}; 
    \node at (90+40:0.7) {$\theta_{c}$}; 
    \node at (90+75:1) {$R$};
  \end{scope}
  \end{tikzpicture}
	\caption{Corner conformal ruler (left) and a zoom-in to the corner regions regulated by hole with a CFT groundstate $\ket{\Omega}_{\text{circle}}$ on the boundary (right).}
	\label{fig:corner-conf-ruler}
\end{figure}

\subsection{Corner vector fixed-point equation}
\phantomsection\label{subsec:stationary-VFE-implied}
On the corner conformal ruler Fig.~\ref{fig:corner-conf-ruler}, we can further obtain a vector fixed point equation (VFPE), namely 
\begin{align}\label{eq:VFE-implied}
     K_{\mathbb{D}}(\eta_g)  \ket{\Psi} = \frac{c_{\tot}}{6} h(\eta_g) \ket{\Psi},
\end{align}
where 
\begin{align}
    K_{\mathbb{D}}(\eta) &= \eta \hat{\Delta}(AA',B,CC') + (1-\eta) \hat{I}(A:C|B) \\ 
     \hat{\Delta}(AA',B,CC') &= K_{AA'B} + K_{CC'B} - K_{AA'} - K_{CC'}  \\ 
     \hat{I}(A:C|B) &= K_{AB} + K_{BC} - K_{ABC} - K_{B}. 
\end{align}
This can be derived by combining the Hypothesis~\ref{hypo:corner} with the results in \cite{Vir,Lin:2023pvl}. For a purely chiral\footnote{By ``purely chiral'', we mean there is only one chiral component. This is only for simplicity. To generalize the argument with anti-chiral component, one can just replace the holomorphic stress-energy tensor $T(x)$ with the addition $T(x) + \bar{T}(x)$.} CFT groundstate on a circle with length $L$, the entanglement Hamiltonian of an interval $[x_a,x_b]$ is of the form 
\begin{align}
     K^{\chi CFT}_{[x_a,x_b]} &= \int dx \beta_{[x_a,x_b]}(x) T(x) + \kappa_{[x_a,x_b]}\mathbbm{1}, \\ 
    \beta_{[x_a,x_b]}(x) &= 2\Theta(x-x_a) \Theta(x_b-x) \frac{\sin((x-x_a)/2)\sin((x_b-x)/2)}{\sin((x_b-x_a)/2)} \label{eq:coolness}
\end{align}
where $\Theta(x)$ is the Heaviside step function and $\kappa_{[x_a,x_b]}$ is a c-number such that the expectation value on the groundstate $\langle K^{\chi CFT}_{[x_a,x_b]} \rangle = S_{[x_a,x_b]}^{\chi CFT}$\footnote{Because of the Casimir energy, $\langle T(x)\rangle \neq 0$, $\kappa_{[x_a,x_b]}$ is not exactly equal to the entanglement entropy.}. One can then verify with Hypothesis~\ref{hypo:corner}, that 
\begin{align}
     K_{\mathbb{D}}(\eta_g) &\eqover{\mathbb{V}} \eta_g \( K^{\chi CFT}_{ab}  + K^{\chi CFT}_{bc} - K^{\chi CFT}_a - K^{\chi CFT}_c\) + (1-\eta_g)(K^{\chi CFT}_{ab} + K^{\chi CFT}_{bc} - K^{\chi CFT}_b - K^{\chi CFT}_{abc})  \\ 
    &= \int dx \left[ \eta_g \( \beta_{ab} (x) + \beta_{bc}(x) - \beta_a(x) - \beta_c(x)\) + (1-\eta_g)(\beta_{ab}(x) + \beta_{bc}(x) - \beta_b(x) - \beta_{abc}(x)) \right] T(x) + \alpha \mathbbm{1} \label{eq:int-beta-T}\\ 
    & = \alpha \mathbbm{1}. 
\end{align}
where $O_1 \eqover{\mathbb{V}} O_2$ is a short hand notation for $O_1 \mathbb{V} = \mathbb{V}O_2$ and $\alpha$ is a c-number. With $\eta_g$ computed using Eq.~\eqref{eq:eta-g} by applying the Hypothesis~\ref{hypo:corner}, one can show that the weight function in front of $T(x)$ in Eq.~\eqref{eq:int-beta-T} vanishes as discovered in \cite{Lin:2023pvl}. By taking the expectation value, one can obtain the quantity $\alpha = c_{\tot}/6 \cdot h(\eta_g)$ where $h(x) = - x \log(x) - (1-x) \log(1-x)$. 

\subsection{Modular commutators with corners}
One can also use Hypothesis~\ref{hypo:corner} to derive the formula proposed in \cite{Kim2021,Kim:2021tse} relating the modular commutator with the chiral central charge. Moreover, one can derive a formula for modular commutators on an ``incomplete disk'' [Fig.~\ref{fig:J-incomplete} (left)] that are studied numerically in \cite{Fan:2022ipl}. One can also explain the results about modular commutators with both bulk and edge contributions in \cite{Vir,kim2024conformalgeometryentanglement}. 

\subsubsection{Preparation: modular commutator in chiral CFT}
As a preparation, we first discuss the computation of the modular commutator in a chiral CFT groundstate $\ket{\Omega}$ on a circle of length $L$. This question was studied in \cite{Zou:2022nuj}. Upon a close examination, we obtained a slightly generalized version of the formula derived in \cite{Zou:2022nuj}, with a regulation for possible discontinuities as we will mention below. 

Consider two generic intervals $[a,b]$ and $[c,d]$ on the circle. We found that the modular commutator $J_{[a,b],[c,d]}$ is of the form 
\begin{align}\label{eq:J-chiral-CFT}
     J_{[a,b],[c,d]} \equiv \ii \langle \Omega | [K_{[a,b]}^{\chi CFT}, K_{[c,d]}^{\chi CFT}] | \Omega \rangle = \frac{\pi c_{-}}{6}(1- 2 \tilde{\eta}(a,c,b,d)) F_{[a,b],[c,d]}, 
\end{align}
where $\tilde{\eta}$ as a function of $a,c,b,d$ is defined as 
\begin{align}
     \tilde{\eta}(a,c,b,d) = \frac{\sin\( \frac{\pi (c-a)}{L} \) \sin\( \frac{\pi (d-b)}{L} \)}{\sin\( \frac{\pi (b-a)}{L} \) \sin\( \frac{\pi (d-c)}{L} \)}, 
\end{align}
and $F_{[a,b],[c,d]}$ is defined as 
\begin{align}
     F_{[a,b],[c,d]} = \Theta_{[c,d]}(b) - \Theta_{[c,d]}(a)
\end{align}
with $\Theta_{[a,b]}(x)$ being the $2\pi$-periodic extension of the product of Heaviside step functions $\Theta(x-a)\Theta(b-x)$, explicitly, 
\begin{align}\label{eq:Theta-x}
    \Theta_{[a,b]}(x) \equiv \left\{ \begin{aligned}
        &1,& &\text{ if } x \in (a,b) \\ 
        &1/2,&  &\text{ if } x = a \text{ or }x = b \\
        &0,& & \text{ otherwise} 
    \end{aligned} \right., \quad x \in S^1. 
\end{align}
In App.~\ref{app:TT}, we give two derivations of Eq.~\eqref{eq:J-chiral-CFT}, one from the operator product expansion (OPE) of the stress-energy tensor, and one using the Virasoro algebra. From the derivation one can see the origin of this function $F_{[a,b],[c,d]}$. 

The main difference between our result Eq.~\eqref{eq:J-chiral-CFT} and the result in \cite{Zou:2022nuj} is this factor $F_{[a,b],[c,d]}$.  Including this factor, the result can be applied for more general cases. The authors of \cite{Zou:2022nuj} consider the case $a<c<b<d$, in which case $F_{[a,b],[c,d]} = -1$, so we obtain the same result as the one in \cite{Zou:2022nuj}. In our result Eq.~\eqref{eq:J-chiral-CFT}, we do not demand any relations between the two intervals $[a,b]$ and $[c,d]$. In particular, there is a case where end points of the two intervals coincide, such as $a = c < b < d$, and our result provide a regulation for such cases. 

Let us explain more why the suitable regulation is $\Theta_{[a,b]}(x = a) = \Theta_{[a,b]}(x = b) = \frac{1}{2}$ in $F_{[a,b],[c,d]}$. This in fact comes from both the OPE derivation and Virasoro algebra derivation of Eq.~\eqref{eq:J-chiral-CFT} in App.~\ref{app:TT}. The OPE computation is rather technical. The Virasoro algebra computation gives a clear indication that this is the right choice. Since the Virasoro generators are the Fourier modes of the stress-energy tensor, as one can imagine, the function $\Theta_{[a,b]}(x)$ in the computation arises in the form of its Fourier series. This will result in Eq.~\eqref{eq:Theta-x} because the Fourier series of a function $f(x)$ with a jump discontinuity $f(x \to x_{*}^+)\neq f(x \to x_{*}^-)$ converges to the middle point $(f(x \to x_{*}^+) + f(x \to x_{*}^-))/2$ at the discontinuous point $x_*$. Therefore, we have the following results for such special cases: 
\begin{equation}
    J_{[a,b],[c,d]} = \left\{ \begin{aligned} 
        &-\frac{\pi c_{-}}{12},&&\text{ if } a = c < b < d \\ 
        &\frac{\pi c_{-}}{12},&& \text{ if } a<c<b = d. 
    \end{aligned} \right.
\end{equation}

\subsubsection{Modular commutator from corners}
Now let us compute the modular commutator 
\begin{align}
     J(A,B,C) \equiv \ii \langle \Psi | [K_{AB},K_{BC}] |\Psi\rangle,
\end{align}
where $ABC$ is shown in Fig.~\ref{fig:J-incomplete}.

\begin{figure}[htb]
  \begin{tikzpicture}[scale = 0.8]
    \filldraw[color = yellow!10, draw = black, thick] (0,0) -- (90:4) arc (90:90+360-30:4) -- cycle; 
    \draw[thick] (0,0) -- (90+120:4); 
    \draw[thick] (0,0) -- (90+240:4); 
    \node[font = \LARGE] at (90+60:2.5) {$A$}; 
    \node[font = \LARGE] at (90+60+120:2.5) {$B$}; 
    \node[font = \LARGE] at (90+60+240-15:2.5) {$C$};
    
    \node[font = \large] at (90+60:0.5) {$\theta_{a}$}; 
    \node[font = \large] at (90+60+120:0.5) {$\theta_{b}$}; 
    \node[font = \large] at (90+60+240-15:0.5) {$\theta_{c}$};

    \begin{scope}[xshift = 10cm]
    \filldraw[color = yellow!10] (0,0) -- (90:4) arc (90:90+360-30:4) -- cycle;  
    \node[font = \LARGE] at (90+60:2.5) {$A$}; 
    \node[font = \LARGE] at (90+60+120:2.5) {$B$}; 
    \node[font = \LARGE] at (90+60+240-15:2.5) {$C$};
    \draw[color = blue, line width = 2 pt] (0,0) -- (90:4) arc  (90:90+240:4) -- cycle;
    \draw[color = red, line width = 4 pt, opacity = 0.5] (0,0) -- (90+120:4) arc  (90+120:90+330:4) -- cycle; 
    \filldraw[color = white, draw = gray] (0,0) circle (1); 

    \draw[blue, line width = 2 pt] (90:1) arc (90:90+240:1); 
    \draw[red, line width = 4 pt, opacity = 0.5] (90+120:1) arc (90+120:90+330:1); 

    \filldraw[color = white, draw = gray] (90+120:4) circle (1);
    \draw[blue, line width = 2 pt] (90+120:4) ++ (-{acos(1/8)+30}:1) arc (-{acos(1/8)+30}:{acos(1/8)+30}:1);
    \draw[red, line width = 4 pt, opacity = 0.5] (90+120:4) ++ (30:1) arc (30:{-acos(1/8)+30}:1);
    
    \filldraw[color = white, draw = gray] (90+240:4) circle (1);
    \draw[blue, line width = 2 pt] (90+240:4) ++ (150:1) arc (150:{acos(1/8)+150}:1);

    \draw[red, line width = 4 pt, opacity = 0.5] (90+240:4) ++ ({150-acos(1/8)}:1) arc ({150-acos(1/8)}:{acos(1/8)+150}:1);

    \draw[dashed] ({90+120-2*asin(1/8)}:4) arc ({90+120-2*asin(1/8)}:{90+120+2*asin(1/8)}:4);
    \draw[dashed] (90+120:4-1) -- (90+120:4);

    \draw[dashed] ({90+240-2*asin(1/8)}:4) arc ({90+240-2*asin(1/8)}:{90+240+2*asin(1/8)}:4);
    \draw[dashed] (90+240:4-1) -- (90+240:4);

    \draw[dashed] (0,0) -- (90:1); 
    \draw[dashed] (0,0) -- (90+120:1); 
    \draw[dashed] (0,0) -- (90+240:1); 
    \draw[dashed] (0,0) -- (90+330:1); 

    \node[font = \large] at (90+60:1.4) {$a_1$};
    \node[font = \large] at (-0.4, 2.5) {$a_2$}; 
    \node[font = \large] at (90+60:3.6) {$a_3$};
    \node[font = \large] at ($(90+120:4) + (90-20:1.4)$) {$a_4$};
    \node[font = \large] at ($(90+120:2) + (0,0.4)$) {$a_5$}; 
    
    \node[font = \large] at (90+60+120:1.4) {$b_1$}; 
    \node[font = \large] at ($(90+120:4) + (-15:1.4)$) {$b_2$};
    \node[font = \large] at (-90:3.6) {$b_3$};
    \node[font = \large] at ($(90+240:4) + (180+15:1.4)$) {$b_3$};
    \node[font = \large] at ($(90+240:2) + (0,-0.4)$) {$b_4$};
    
    \node[font = \large] at (90+60+240-15:1.4) {$c_1$};
    \node[font = \large] at ($(90+240:4) + (90+15:1.4)$) {$c_2$};
    \node[font = \large] at (90+60+240-15:3.6) {$c_3$};
    \node[font = \large] at ($(90+330:3) + (0.1,-0.6)$) {$c_4$};
  \end{scope}
  \end{tikzpicture}
  \caption{Modular commutator on incomplete disk (left) and its corner regulations (right).}
  \label{fig:J-incomplete}
\end{figure}
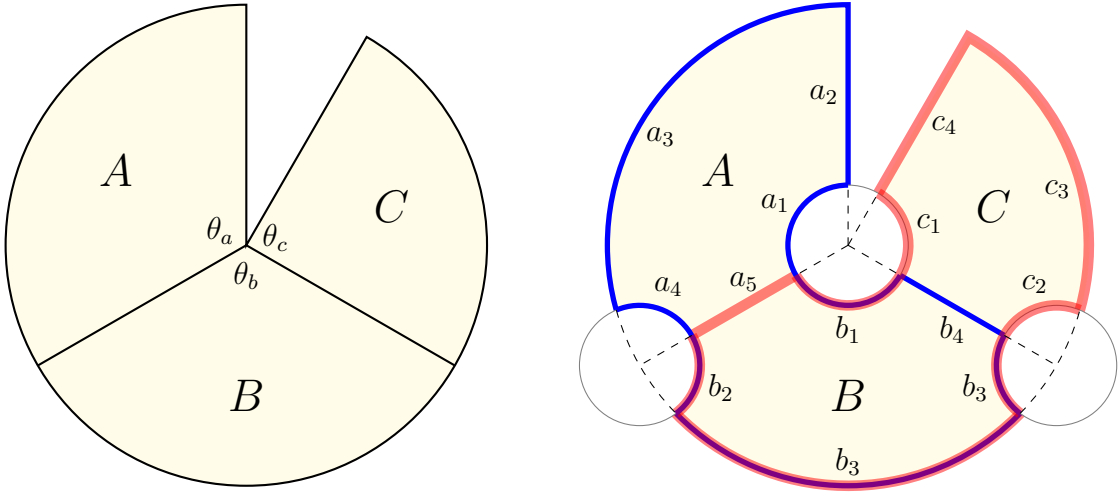

Using Hypothesis~\ref{hypo:corner}, we can convert the computation to the CFT computation in Fig.~\ref{fig:J-incomplete}. We first replace the $K_{AB}$ and $K_{BC}$ with 
\begin{align}
     K_{AB} &\eqover{\mathbb{V}} \beta \int_{x \in a_2 \cup a_3 \cup b_3 \cup b_4} dx T(x) + K_{a_1b_1}^{\chi\CFT} + K_{a_4b_2}^{\chi\CFT} + K_{b_3}^{\chi\CFT} \\ 
     K_{BC} &\eqover{\mathbb{V}} \beta \int_{x \in a_5 \cup b_3 \cup c_3 \cup c_4 } dx T(x) + K_{b_1c_1}^{\chi\CFT} + K_{b_2}^{\chi\CFT} + K_{b_3c_2}^{\chi\CFT}. 
\end{align}
Then, the commutator is 
\begin{align}
     \ii \langle \Psi| [K_{AB}, K_{BC}] | \Psi\rangle = J_{a_1b_1,b_1c_1} + J_{a_4b_2,b_2} + J_{b_3,b_3c_2},
\end{align}
where $J_{x,y} = \ii \bra{\Omega} [K_x^{\chi\CFT}, K_y^{\chi\CFT}]\ket{\Omega}$ for two intervals $x$ and $y$. 
Using the result Eq.~\eqref{eq:J-chiral-CFT}, the values are 
\begin{align}\label{eq:J1}
    J_{a_1b_1,b_1c_1} = \frac{\pi c_{-}}{6}(2 \eta - 1) \\ 
    J_{a_4b_2,b_2} = J_{b_3,b_3c_2} = \frac{\pi c_{-}}{12},\label{eq:J23}
\end{align}
and therefore 
\begin{align}
    \label{eq:J-eta}
    J(A,B,C) = \frac{\pi c_{-}}{3} \eta(a,b,c), 
\end{align}
with $\eta(a,b,c)$ being the cross-ratio computed from the opening angles $\theta_a,\theta_b,\theta_c$ shown in Fig.~\ref{fig:J-incomplete} (left):
\begin{align}
    \eta(a,b,c) = \frac{\sin(\theta_a/2) \sin(\theta_c/2)}{\sin((\theta_a + \theta_b)/2) \sin((\theta_b + \theta_c)/2)}. 
\end{align}
This result Eq.~\eqref{eq:J-eta} gives an analytic expression for the modular commutators on an incomplete disk that are studied in \cite{Fan:2022ipl}. In this paper, we also did a similar numerics to verify Eq.~\eqref{eq:J-eta} in Section~\ref{sec:num}, which shows excellent agreement when the angles $\theta_a,\theta_b,\theta_c$ are large enough. 

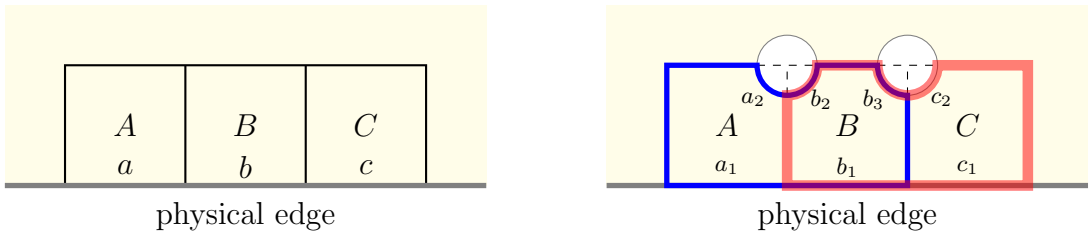
\begin{figure}[htb]
    \centering
    \begin{tikzpicture}[scale = 0.8]
    \fill[color = yellow!10] (0,0) -- (8,0) -- (8,3) -- (0,3) -- cycle; 
    \draw[thick] (1,0) -- (1,2) -- (1+6,2) -- (1+6,0); 
    \draw[thick] (1+2,0) -- (1+2,2); 
    \draw[thick] (1+4,0) -- (1+4,2); 
    \draw[line width = 2 pt, color = gray] (0,0) -- (8,0);
    \node[font = \large] at (2,1) {$A$};
    \node[font = \large] at (2+2,1) {$B$};
    \node[font = \large] at (2+4,1) {$C$};
    \node[font = \large] at (2,0.3) {$a$};
    \node[font = \large] at (2+2,0.3) {$b$};
    \node[font = \large] at (2+4,0.3) {$c$};
    \node[font = \large] at (4,-0.5) {physical edge};

    \begin{scope}[xshift = 10cm]
      \fill[color = yellow!10] (0,0) -- (8,0) -- (8,3) -- (0,3) -- cycle; 
      \draw[line width = 2 pt, color = gray] (0,0) -- (8,0);
      \node[font = \large] at (2,1) {$A$};
      \node[font = \large] at (2+2,1) {$B$};
      \node[font = \large] at (2+4,1) {$C$};
      \node[font = \large] at (4,-0.5) {physical edge};

      \filldraw[color = white, draw = gray] (1+2,2) circle (0.5); 
      \draw[dashed] (1+2-0.5,2) -- (1+2+0.5,2);
      \draw[dashed] (1+2,2-0.5) -- (1+2,2);

      \filldraw[color = white, draw = gray] (1+4,2) circle (0.5); 
      \draw[dashed] (1+4-0.5,2) -- (1+4+0.5,2);
      \draw[dashed] (1+4,2-0.5) -- (1+4,2);

      \draw[blue, line width = 2 pt] (1,0) -- (1,2) -- (1+2-0.5,2) arc (180:360:0.5) -- (1+4-0.5,2) arc (180:270:0.5) -- (1+4,0) -- cycle;
      
      \draw[red, line width = 4 pt, opacity = 0.5] (3,0) -- (3,2-0.5) arc (270:360:0.5) -- (3+2-0.5,2) arc (180:360:0.5) -- (3+4,2) -- (3+4,0) -- cycle;

      \node at (2,0.3) {$a_1$}; 
      \node at ($(3,2) + ({180+45}:0.8)$) {$a_2$}; 

      \node at (2+2,0.3) {$b_1$}; 
      \node at ($(3,2) + ({270+45}:0.8)$) {$b_2$}; 
      \node at ($(5,2) + ({180+45}:0.8)$) {$b_3$};
      
      \node at (2+2+2,0.3) {$c_1$}; 
      \node at ($(5,2) + ({270+45}:0.8)$) {$c_2$}; 
    \end{scope}
  \end{tikzpicture}
	\caption{Modular commutator near the physical edge}
	\label{fig:J-edge}
\end{figure}

Similarly, one can also compute $J(A,B,C)$ near the physical gapless edge with $A,B,C$ shown in Fig.~\ref{fig:J-edge}. Applying Hypothesis~\ref{hypo:corner}, the modular commutator is a sum of three modular commutators in chiral CFT groundstate: 
\begin{align}
    J(A,B,C) = J_{a_1b_1,b_1c_1} + J_{a_1b_2,b_2} + J_{b_3,b_3c_3},
\end{align}
where $J_{a_1b_1,b_1c_1}$ is from the physical gapless edge, and $J_{a_1b_2,b_2}, J_{b_3,b_3c_3}$ are from the sharp corners. With the similar computation, we obtain $J(A,B,C) = \frac{\pi c_{-}}{3} \eta(a,b,c)$, with $\eta(a,b,c)$ being the cross-ratios for the contiguous integrals $(a,b,c)$ in Fig.~\ref{fig:J-edge} (left). This result is consistent with the results that were derived based on several locally checkable assumptions and numerically verified in \cite{kim2024conformalgeometryentanglement}.

One can hope to apply the Hypothesis~\ref{hypo:corner} to understand the result in \cite{Park:2024hix}.

\section{Logical framework for corner entanglement}

\label{sec:logical-framework}
The results obtained above can be summarized into a logical framework. In the logical framework, we start with a set of locally-checkable conditions as axioms, and one can show that the universal properties discussed earlier are logical consequences of these axioms. These axioms, as we will explain, can be treated as a set of conditions for a low-energy RG fixed point. This understanding indicates that the properties following from the axioms will emerge at the low-energy fixed point and hence we obtain a natural explanation for why these universal properties emerge. 

Besides providing us an explanation of the emergence of the universal properties, there are also other practical uses for such a logical framework. For example, in practice, because it packs up the universal properties into a set of locally checkable conditions, one can simply verify these conditions instead of verifying these properties one by one. In other words, the verification of the axioms can quantify to what extent these universal properties are satisfied on a given wavefunction\footnote{Usually, on a generic wavefunction in the phase, the universal properties are only satisfied up to some errors due to UV degrees of freedom.}. Theoretically, this logical framework can give us a clear logical dependence among these universal properties. Utilizing this framework, one can turn physical intuition and physical pictures into concrete statements and sharpen our understanding of these universal properties. 

The axioms of the logical framework are the following. First we will assume \Azero and \Aone in the bulk [Assumption~\ref{assumption:A0} and Assumption~\ref{assumption:A1}]. These two axioms are posited on the length scale larger than the correlation length and as we explained in Section~\ref{subsec:local-ent-hmt}, these two conditions indicate that on such a length scale, one can treat the state as the zero-correlation-length RG fixed point representative of a gapped phase. Then we will make axioms about corners. We will first define a quantity $\fc$ that under the right condition can be interpreted as the minimal total central charge discussed in \cite{Siva:2021cgo}. We will posit the condition $\fc \neq 0$ and a stationarity condition. This condition naturally follows from the fact that the state does not admit gapped boundary. 
We will also posit a genericity condition to rule out some trivial states. 

Based on these axioms, we can derive a universal measure of the angles of sharp corner modulo global conformal transformations. In terms of this measure, we can obtain the results mentioned in Section~\ref{sec:universal-corner-contributions}, namely (1) the corner entanglement entropy formula, (2) the corner vector fixed point equation and (3) modular commutator for incomplete disk. 

Before we dive into the details, we comment that the logical formalism described below is quite similar to the one in \cite{kim2024conformalgeometryentanglement} for the gapless physical edge. In \cite{kim2024conformalgeometryentanglement}, the main focal object is three contiguous intervals $a,b,c$ along the physical edge. The proofs of the statements from the axioms described in \cite{kim2024conformalgeometryentanglement} in fact is applicable here. One can simply transplant the proofs by imagining the corner as a small hole with a  gapless edge\footnote{Put differently, for the purpose of the logical relations in the framework, it does not matter whether $a,b,c$ are three contiguous intervals along the gapless edge or three contiguous sharp corners.}. Because of this, we shall simply summarize the axioms and the logical conclusions, and focus more on the physical understandings. For the detailed proofs we refer to \cite{kim2024conformalgeometryentanglement}.  

We remark that there is a tight relation between $\fc$ and edge ungappability. We will postpone the discussion of this relation in the next section. 

\subsection{Corner axioms, definition of $\fc$ and its stationarity}
The axioms \Azero and \Aone were already discussed in Section~\ref{sec:bulk-edge}. In this section, we will focus on the axioms regarding the corner entanglement.

\begin{figure}[thb]
    \centering
    \begin{tikzpicture}
    \filldraw[thick, fill = yellow!10, draw = blue] (0,0) -- (90-60:3) arc (90-60:90+60:3) -- cycle; 
      \filldraw[thick, fill = yellow!10, draw = blue] (90-40:3) -- (90-40:4) arc (90-40:90+40:4) -- (90+40:3) arc (90+40:90-40:3); 
      \filldraw[thick, fill = yellow!20, draw = gray] (0,0) -- (90-60:0.8) arc (90-60:90+60:0.8) -- cycle;
      \draw[thick, color = blue] (0,0) -- (90-20:3); 
      \draw[thick, color = blue] (0,0) -- (90+20:3); 
      \draw[thick, color = blue] (90:3) -- (90:4); 
      \draw[thick, color = blue] (0,0) -- (90-60:3);
      \draw[thick, color = blue] (0,0) -- (90+60:3);
      \node at (90-40:2) {$A$}; 
      \node at (90:2) {$B$}; 
      \node at (90+40:2) {$C$}; 
      \node at (90-20:3.5) {$A'$};  
      \node at (90+20:3.5) {$C'$};  
      \node at (90-40:0.5) {$a$}; 
      \node at (90:0.5) {$b$}; 
      \node at (90+40:0.5) {$c$}; 

    \begin{scope}[xshift = 8cm]
      \filldraw[thick, fill = yellow!10, draw = blue] (0,0) -- (90-60:3) arc (90-60:90+60:3) -- cycle; 
      \filldraw[thick, fill = yellow!10, draw = blue] (90-40:3) -- (90-40:4) arc (90-40:90+40:4) -- (90+40:3) arc (90+40:90-40:3); 
      \filldraw[thick, fill = yellow!20, draw = gray] (0,0) -- (90-60:0.8) arc (90-60:90+60:0.8) -- cycle;
      \draw[thick, color = blue] (0,0) -- (90-20:3); 
      \draw[thick, color = blue] (0,0) -- (90+20:3); 
      \draw[thick, color = blue] (90:3) -- (90:4); 
      \draw[thick, color = blue] (0,0) -- (90-60:3);
      \draw[thick, color = blue] (0,0) -- (90+60:3);
      \node at (90-40:2) {$A$}; 
      \node at (90:2) {$B$}; 
      \node at (90+40:2) {$C$}; 
      \node at (90-25:3.5) {$A'$};  
      \node at (90+25:3.5) {$C'$};  
      \node at (90-40:0.5) {$a$}; 
      \node at (90:0.5) {$b$}; 
      \node at (90+40:0.5) {$c$}; 

      \filldraw[color = white] (90-15:4.1) arc (90-15:90+15:4.1) -- (90+15:3.5) arc (90+15:90-15:3.5) -- cycle;
      \draw[blue, thick] (90-15:4) -- (90-15:3.5) arc (90-15:90+15:3.5) -- (90+15:4); 
      \draw[blue, dashed] (90-15:4) arc (90-15:90+15:4);  
      \draw[blue, dashed] (90:3.5) -- (90:4);
      \fill[yellow!10] (90-20+10:3.2) arc (90-20+10:90-20-10:3.2) -- (90-20-10:2.5) arc (90-20-10:90-20+10:2.5) -- cycle;
      \draw[blue, thick] (90-20-10:3) -- (90-20-10:2.5) arc (90-20-10:90-20+10:2.5) -- (90-20+10:3); 
      \draw[blue, dashed] (90-20-10:3) arc (90-20-10:90-20+10:3);
    \end{scope}
  \end{tikzpicture}
	\caption{Corner conformal ruler (left) and its deformations (right). Each region is enclosed by the solid blue lines.}
	\label{fig:corner-conf-ruler-deform}
\end{figure}
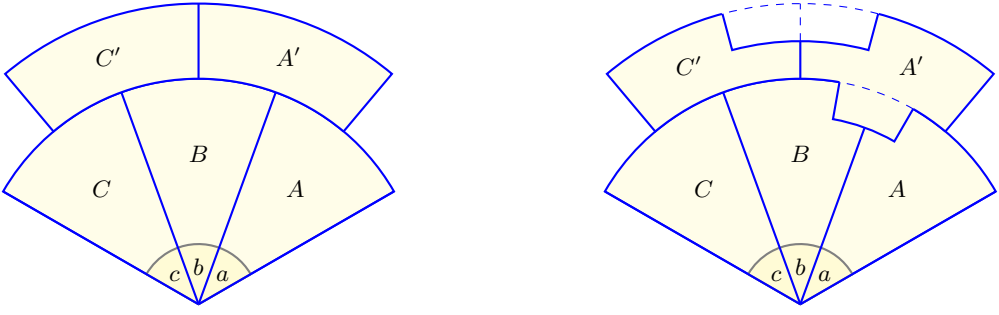

Consider a combination of regions $\dune = (A,A',B,C,C')$ shown in Fig.~\ref{fig:corner-conf-ruler-deform}, which we refer to as a ``(corner) conformal ruler'' following \cite{kim2024conformalgeometryentanglement}.
We define $\fc, \eta$ as the solutions to the equations 
\begin{equation}\label{eq:def-c-eta}
\begin{aligned}
    &e^{-6\Delta(\dune)/\fc(\dune)} + e^{-6 I(\dune)/\fc(\dune)} = 1 \\ 
    &\eta(\dune) \equiv e^{-6\Delta(\dune)/\fc(\dune)}
\end{aligned} ~~ \Rightarrow ~~ (\fc(\dune),\eta(\dune)),
\end{equation}
where 
\begin{equation}
    \begin{split}
        \Delta(\dune) &= \Delta(AA',B,CC')  \\ 
        &= (S_{AA'B}+S_{CC'B}-S_{AA'}-S_{CC'})_{\ket{\Psi}} \\ 
        I(\dune) &= I(A:C|B) \\ 
         &= (S_{AB}+S_{BC}-S_{ABC}-S_{B})_{\ket{\Psi}}.
    \end{split}
\end{equation}
Use $a,b,c$ to denote the three corner regions $A,B,C$ respectively shown in Fig.~\ref{fig:corner-conf-ruler-deform}. With the bulk {\bf A1} condition, one can show that $\fc(\dune),\eta(\dune)$ depend only on $(a,b,c)$. This is because bulk {\bf A1} condition allows one to deform the region combinations $(A,A',B,C,C')$ without changing the value of $\Delta(\dune),I(\dune)$, as long as the corner regions $(a,b,c)$ and the topology of the partition is unchanged. The detailed argument for this fact is the same as the one given in \cite{kim2024conformalgeometryentanglement} with appropriate substitutions. 

Before we introduce the axiom for $\fc$, let us first explain the physical picture of $\fc$. 
\begin{itemize}
    \item $\fc$ is meant to detect the corner contribution to the entanglement entropy, i.e. the $f(\theta)$ term in Eq.~\eqref{eq:SA-corner-intro}, in a gapped wavefunction. This point can be explained from two viewpoints. 
    Logically speaking, based on the \Aone condition, one can freely deform the support of $\Delta$ and $I$ as long as the corner region $(a,b,c)$ remains unchanged, as shown in Fig.~\ref{fig:corner-conf-ruler-deform}\footnote{The topology of the partition in a corner conformal ruler should be maintained under the deformation.}. This indicates that only the entanglement near the corners matters in $\Delta, I$. 
    From a different viewpoint, if one assumes the area law [Eq.~\eqref{eq:SA-corner-intro}] in the first place, then $\Delta$ and $I$ are non-zero only if there exists a corner-dependent term $f(\theta)$. Based on the definition of $\fc$, if $\Delta = I = 0$, then $\fc = 0$ \cite{kim2024conformalgeometryentanglement}. Therefore, having $\fc \neq 0$ indicates the existence of a corner-dependent term in the entanglement entropy. 
    \item $\fc$ also has a close relation to edge theory via bulk/edge correspondence. If we assume the Hypothesis~\ref{hypo:corner}, this $\fc$ matches the (minimal) total central charge of the edge CFT. The word ``minimal'' is referring to the fact that the edge CFT, among all the possible RG fixed-point edge systems, has the minimal total central charge. In other words, with the bulk/edge correspondence picture, we expect that $\fc$ is related to the most stable boundary condition one could assign to the edge of the system. If the system admits a gapped boundary, which can be regarded as an ``empty'' CFT, one can conclude $\fc = 0$, which is consistent with the fact that there are no sharp corner contributions at the IR fixed-point. If the system does not admit a gapped boundary, this means that there exists an edge CFT whose total central charge cannot be reduced by any local perturbations near the edge. We expect that $\fc$ computes the total central charge when the system is at the IR fixed point. In the next section, we will discuss the relation between $\fc$ and the edge theory without assuming Hypothesis~\ref{hypo:corner}. 
\end{itemize}

Based on the two physical pictures of $\fc$ mentioned above, it is plausible that, if the wavefunction is the IR fixed point representative of a chiral gapped phase, $\fc$ as a function of states in the phase will reach the minimal value. That is, $\fc(\ket{\Psi}) = \mathrm{min}\{\fc(\ket{\Phi})\}$, where $\ket{\Phi}$ are representative wavefunctions in a chiral gapped phase. With this motivation, we assume the following stationarity condition: 

Consider any norm-preserving perturbation of $\ket{\Psi}$ of the form $\ket{\Psi}\to \ket{\Psi}+\epsilon\ket{\Psi'}$ with $\epsilon$ being an infinitesimal real number and $\langle \Psi|\Psi'\rangle = 0$. Let $\delta\fc$ denotes the resulting variation of $\fc(\dune)_{\ket{\Psi}}$ in linear order of $\epsilon$, we assume:
\begin{assumption}[Stationarity]
    For any $\dune$ in the bulk, $\fc(\dune)_{\ket{\Psi}}$ is stationary in the sense that 
    \begin{equation}
        \delta \fc = 0,
    \end{equation}
    for any norm-preserving perturbation of $\ket\Psi$. 
\end{assumption}

This assumption can be equivalently stated using a vector fixed-point equation, because of the following theorem: 
\begin{theorem}[Vector fixed-point equation] 
$\fc(\dune)_{\ket{\Psi}}$ is stationary if and only if 
    \begin{equation}\label{eq:VFE-logic}
        \Big[\eta(\dune)\hat{\Delta}(\dune) + (1-\eta(\dune))\hat{I}(\dune)\Big]\ket{\Psi} \propto \ket{\Psi}, 
    \end{equation}
    where 
    \begin{equation}
        \begin{split}
            \hat{\Delta}(\dune) &= \hat{\Delta}(AA',B,CC')  \\ 
            &= K_{AA'B}+K_{CC'B}-K_{AA'}-K_{CC'} \\ 
            \hat{I}(\dune) &= \hat{I}(A:C|B) \\ 
             &= K_{AB}+K_{BC}-K_{ABC}-K_{B}.
        \end{split}
    \end{equation}
\end{theorem} 
This vector fixed-point equation condition can also be understood with the physical picture from Hypothesis~\ref{hypo:corner} as explained in Section~\ref{subsec:stationary-VFE-implied}. 
But we emphasize that within this section, the Hypothesis~\ref{hypo:corner} is only considered as a guidance or the physical motivations behind this logical framework. For any conclusions that we are going to make, they logically do not depend on the validity of the Hypothesis~\ref{hypo:corner}. Moreover, one can even attempt to include the Hypothesis~\ref{hypo:corner} as a logical consequence of the axioms in this framework. We shall discuss this as a future direction in the discussion section. 

The next axiom about the corner entanglement is the genericity condition: 
\begin{assumption}[genericity] \label{assumption:genericity}
    For any partition $(A,B,C,D,X,Y,Z)$ as shown in Fig.~\ref{fig:genericity}, which is a combination of two corner conformal rulers, the following three vectors 
    \begin{align}
        \hat{\Delta}(AX,B,CY)\ket{\Psi}, \hat{\Delta}(BY,C,DZ)\ket{\Psi}, \ket{\Psi}
    \end{align}
    are linearly independent. 
\end{assumption}

\begin{figure}[htb]
    \centering
    \begin{tikzpicture}
    \filldraw[color = yellow!10,draw = blue, thick] (0,0) -- (90-80:3) arc (90-80:90+80:3) -- cycle; 
    \filldraw[color = yellow!10, draw = blue, thick] (90-60:3) -- (90-60:4) arc (90-60:90+60:4) -- (90+60:3) arc (90+60:90-60:3); 
    \draw[blue, thick] (0,0) -- (90-40:3); 
    \draw[blue, thick] (0,0) -- (90:3);
    \draw[blue, thick] (0,0) -- (90+40:3);  
    \draw[blue, thick] (90-20:3) -- (90-20:4); 
    \draw[blue, thick] (90+20:3) -- (90+20:4); 
    \node at (90-60:2) {$A$}; 
    \node at (90-20:2) {$B$}; 
    \node at (90+20:2) {$C$}; 
    \node at (90+60:2) {$D$}; 
    \node at (90-60:0.7) {$a$}; 
    \node at (90-20:0.7) {$b$}; 
    \node at (90+20:0.7) {$c$}; 
    \node at (90+60:0.7) {$d$}; 
    \node at (90-40:3.5) {$X$}; 
    \node at (90:3.5) {$Y$}; 
    \node at (90+40:3.5) {$Z$}; 
  \end{tikzpicture}
	\caption{Regions for genericity condition}
	\label{fig:genericity}
\end{figure}
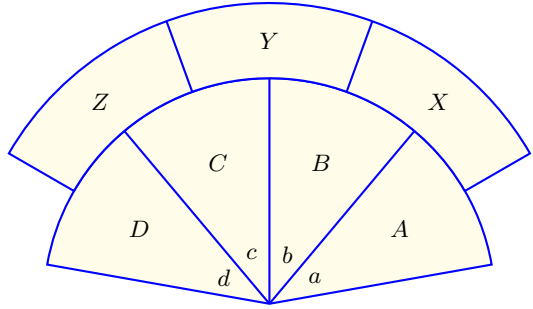

This condition is saying that $\hat{\Delta}$ for a conformal ruler has a non-trivial action on the corner regions. Take $\hat{\Delta}(AX,B,CY)$ for example, if $\hat{\Delta}(AX,B,CY)\ket{\Psi}$ is a linear combination of $\hat{\Delta}(BY,C,DZ)\ket{\Psi}$ and $\ket{\Psi}$, then this indicates that $\hat{\Delta}(AX,B,CY)$ has no action on the corner $a$. In fact, if the system has a non-zero chiral central charge, then this genericity condition is automatically satisfied, as shown in \cite{kim2024conformalgeometryentanglement}. 

\subsection{Logical consequences}

Having introduced the setup and assumptions, we now discuss their logical consequences. 
First of all, we can draw a set of conclusions that are similar to those in \cite{kim2024conformalgeometryentanglement}. We enumerate them as follows: 
\begin{itemize}
    \item $\fc$ defined in Eq.~\eqref{eq:def-c-eta} is a non-zero constant independent of the choices of sharp corner combinations $(a,b,c)$.
    \item $\{\eta(a,b,c)\}$ defined in Eq.~\eqref{eq:def-c-eta} forms a valid set of cross-ratios.
    \item If $c_{-}\neq 0$, then one can show that $J(A,B,C) = \frac{\pi c_{-}}{3} \eta(a,b,c)$, with $\eta(a,b,c)$ defined in Eq.~\eqref{eq:def-c-eta}.  
\end{itemize}
Similarly as in \cite{kim2024conformalgeometryentanglement}, this further enables us to construct a map $\varphi$ that gives a measure to the sharp corner $a$. Explicitly, given a set of sharp corners $\{a\}$, we can assign an angle $a \to \varphi_a$, such that the geometric cross-ratios from these angles match with the cross-ratios computed from the state: 
\begin{equation}\label{eq:eta-geo}
    \begin{split}
    \eta(a,b,c) &= \eta_g(\varphi_a,\varphi_b,\varphi_c)  \\ 
    &= \frac{\sin(\varphi_a/2)\sin(\varphi_c/2)}{\sin((\varphi_a+\varphi_b)/2)\sin((\varphi_b+\varphi_c)/2)}. 
\end{split}
\end{equation}
With these results, one can further derive the formulae for the sharp corner contribution to entanglement entropies and entanglement Hamiltonians. 

{\bf Discuss the contribution to $S_{A,\text{corner}}$.}
We first discuss the formula for entanglement entropies. Due to the definition of $(\fc,\eta)$ in Eq.~\eqref{eq:def-c-eta} as well as the result in Eq.~\eqref{eq:eta-geo}, for a $\dune=(A,A',B,C,C')$ that contains three sharp corners $(a,b,c)$ we obtain 
\begin{equation}\label{eq:Delta-corner}
    \begin{split}
    \Delta(\dune) &= -\frac{\fc}{6}\ln(\eta(a,b,c)) = - \frac{\fc}{6}\ln (\eta_g(\varphi_a,\varphi_b,\varphi_c)) \\ 
    & =f(\varphi_{ab})+f(\varphi_{bc})-f(\varphi_a)-f(\varphi_c),
\end{split}
\end{equation}
where $f(\varphi_{\bullet})$ denotes the contribution to the entanglement entropies from a sharp corner $\bullet$. The last line of Eq.~\eqref{eq:Delta-corner} is because only the sharp corner contributions remains in the linear combinations of entropies in $\Delta(\dune)$. Let's consider a particular $\dune$ such that $\varphi_a = \varphi_c = \theta$ is finite and $\varphi_b = d\theta$ is infinitesimal, then we obtain 
\begin{equation}
    2(f(\theta+d\theta) - f(\theta)) = - \frac{\fc}{6} \ln (\eta_g(\theta,d\theta,\theta))
\end{equation}
From this equation, we can obtain a differential equation for $f(\theta)$: 
\begin{equation}
    \frac{df(\theta)}{d\theta} = \frac{d}{d\theta}\ln(\sin(\theta)/2),
\end{equation}
which gives 
\begin{equation}
    f(\theta) = \frac{\fc}{6}\ln\left(\sin(\theta/2)\right) + C_0,
\end{equation}
where $C_0$ is the integration constant. Since we expect that when $\theta = \pi$, there is no sharp corner contribution to the entanglement entropy, i.e. $f(\pi)=0$, we infer that $C_0 = 0$. This is because when $\pi = 0$, one usually regards this as a smooth segment rather than a sharp corner. One could assign $f(\pi) \neq 0$, then as a result, such a contribution will exist for every smooth segment on $\partial A$, which can be thought of part of as the area law contribution in $S(A)$. From this point of view, one can think of setting $C_0 = 0$ as a result of $f(\theta = \pi)$ being absorbed in the area law term. Finally we obtain the formula of the sharp corner contribution to the entanglement entropy 
\begin{equation}\label{eq:corner-final}
    f(\theta) = \frac{\fc}{6}\ln\left(\sin(\theta/2)\right). 
\end{equation}

{\bf Discuss the contribution to $K_{A,\text{corner}}$.}
At last, we mention that from these axioms, one can decompose the entanglement Hamiltonian $K_A$ with a sharp corner $a$ to a linear combination of those with smaller sharp corners, when they are acting on the state. This utilizes the vector fixed point equation Eq.~\eqref{eq:VFE-logic}, as it relates the action of the entanglement Hamiltonians on regions with different sharp corners. Suppose there are sufficiently many degrees of freedom for a corner $a$ that allow one to make such a decomposition lots of times, then the resulting linear combinations from decomposing $K_A\ket{\Psi}$ can be written as an integral 
\begin{align}
     K_A \ket{\Psi} = \int d \varphi_x \beta_{\varphi_a} (\varphi_x) \CO(x) \ket{\Psi}, 
\end{align}
where $\CO(x)$ is a certain linear combination of $K_{\bullet}$ whose supports have sharp corners around $x$ as shown in Fig.~\ref{fig:decompose-corner}. Here $\beta_{\varphi_a}(\varphi_x)$ is the same coolness function as in Eq.~\eqref{eq:coolness}, with $\varphi_x$ being the angle on a circle whose values are decided by the map $\varphi$ mentioned above. The explicit derivation will be provided in \cite{strength-of-vfpe}.

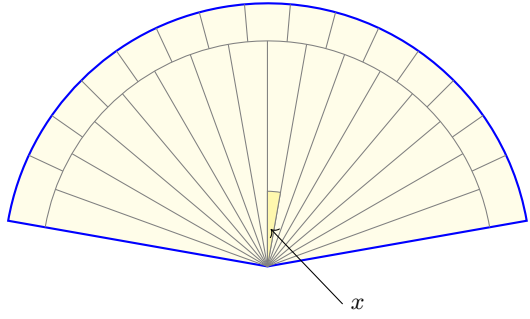
\begin{figure}[htb]
  \begin{tikzpicture}
    \filldraw[color = yellow!10, draw = blue, thick] (0,0) -- (90-80:3.5) arc (90-80:90+80:3.5) -- cycle;
    \draw[gray] (90-80:3) arc (90-80:90+80:3);
    \foreach \x in {1,...,15}
		  \draw[gray] (0,0) -- (90-80+10*\x : 3);
    \foreach \x in {1,...,14}
		  \draw[gray] (90-80+5+10*\x : 3) -- (90-80+5+10*\x : 3.5);
    
    \filldraw[color = yellow!40, draw = gray] (0,0) -- (90-10:1) arc (90-10:90:1) -- cycle; 

    \draw[<-] (90-5:0.5) -- (1,-0.5); 
    \node at (1.2,-0.5) {$x$}; 
  \end{tikzpicture}
  \caption{Decomposition of $K_A \ket{\Psi}$ into an integral of local operators supported on small corner regions. The yellow region is $A$ and its corner region is $a$. We decompose the corner region $a$ into many small corner regions such as $x$.}
  \label{fig:decompose-corner}
\end{figure}

\section{$\fc$ as a diagnostic for ungappable boundary}
\label{sec:diagnostic-for-ungappable}

In this section, we will focus on the relations between $\fc$ and edge ungappability. We will define a value $(\fc)_{\text{min}}$ on a generic representative wavefunctions of a gapped phase, and show that $(\fc)_{\text{min}}=0$ if and only if the phase admits a gapped boundary. Such a $(\fc)_{\text{min}} \neq 0$ reflects a notion of robustness of the corner contribution, because the existence of corner dependent terms in the entanglement entropy of regions with sharp corners is a necessary condition for $(\fc)_{\text{min}} \neq 0$. In the following we are going to show that the robustness of the corner contribution is closely related to the robustness of the gapless edge.

We remark that in the following theorem statements and their proofs, we will focus more on delivering the underlying physical picture. There are several mathematical technical caveats, which we will comment on in the end. 

\subsection{Edge ungappability}

We first define $(\fc)_{\text{min}}$. Consider a generic representative state $\ket{\Psi}$ of a 2+1D gapped phase that satisfies \Azero and \Aone with error vanishing as the subsystem sizes are increased; on a conformal ruler $\mathbb{D}$ whose overall linear size is much larger than the correlation length, we define 
\begin{equation}\label{eq:def-ctot-min}
    (\fc)_{\text{min}} \equiv \mathrm{min} \{ \fc(\mathbb{D})_{\ket{\Phi}}| \ket{\Phi} = U_{XY}\ket{\Psi},\forall\, U_{XY} \},
\end{equation}
where $U_{XY}$ is an arbitrary unitary on a disk $XY$ shown in Fig.~\ref{fig:disentangler} (left). We require the linear sizes of $X$ and $Y$ to be fixed and larger than the correlation length. Here we use the common definition of gapped phases of matter, that is: two states in the thermodynamic limit are in the same gapped phase if they are related to each other by a constant-depth local unitary circuit. Therefore, the set of states $\{ U_{XY}\ket{\Psi}\}_{U_{XY}}$ are all representatives of the same gapped phase. 

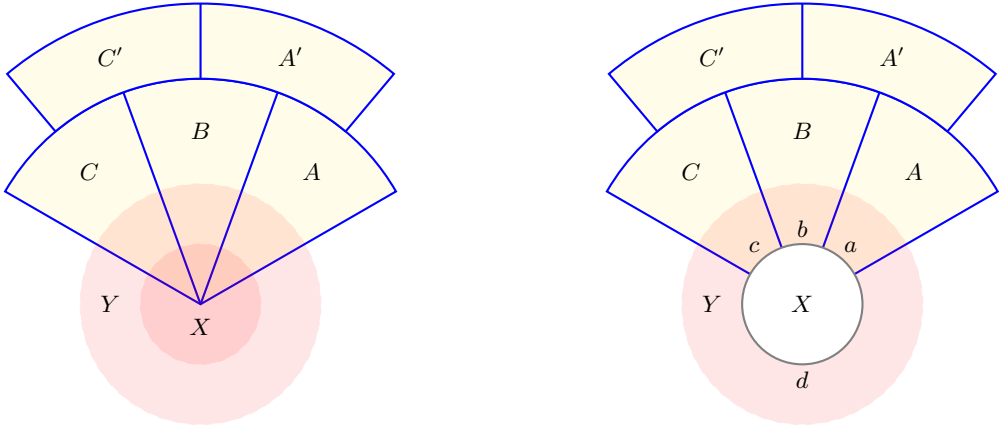
\begin{figure}
  \begin{tikzpicture}
    \filldraw[thick, fill = yellow!10, draw = blue] (0,0) -- (90-60:3) arc (90-60:90+60:3) -- cycle; 
    \filldraw[thick, fill = yellow!10, draw = blue] (90-40:3) -- (90-40:4) arc (90-40:90+40:4) -- (90+40:3) arc (90+40:90-40:3); 
    \draw[thick, color = blue] (0,0) -- (90-20:3); 
    \draw[thick, color = blue] (0,0) -- (90+20:3); 
    \draw[thick, color = blue] (90:3) -- (90:4); 
    \node at (90-40:2.3) {$A$}; 
    \node at (90:2.3) {$B$}; 
    \node at (90+40:2.3) {$C$}; 
    \node at (90-20:3.5) {$A'$};  
    \node at (90+20:3.5) {$C'$};  

    \filldraw[color = red,opacity = 0.1, dashed] (0,0) circle (1.6);
    \filldraw[color = red,opacity = 0.1, dashed] (0,0) circle (0.8);
    
    \node at (0, -0.3) {$X$}; 
    \node at (-1.2, 0) {$Y$}; 

    \begin{scope}[xshift = 8cm]
      \filldraw[thick, fill = yellow!10, draw = blue] (0,0) -- (90-60:3) arc (90-60:90+60:3) -- cycle; 
      \filldraw[thick, fill = yellow!10, draw = blue] (90-40:3) -- (90-40:4) arc (90-40:90+40:4) -- (90+40:3) arc (90+40:90-40:3); 
      \draw[thick, color = blue] (0,0) -- (90-20:3); 
      \draw[thick, color = blue] (0,0) -- (90+20:3); 
      \draw[thick, color = blue] (90:3) -- (90:4); 
      \node at (90-40:2.3) {$A$}; 
      \node at (90:2.3) {$B$}; 
      \node at (90+40:2.3) {$C$}; 
      \node at (90-20:3.5) {$A'$};  
      \node at (90+20:3.5) {$C'$};  

      \filldraw[color = red,opacity = 0.1, dashed] (0,0) circle (1.6);
      \filldraw[color = white, draw = gray, thick] (0,0) circle (0.8);

      \node at (0, 0) {$X$}; 
      \node at (-1.2, 0) {$Y$}; 

      \node at (90-40:1) {$a$}; 
      \node at (90:1) {$b$}; 
      \node at (90+40:1) {$c$};
      \node at (-90:1) {$d$};
    \end{scope}
  \end{tikzpicture}
  \caption{(Left) Region $XY$ around the corner regions of a conformal ruler. $XY$ is a region where \Azero approximately holds. (Right) Conformal ruler for an actual boundary of intervals $(a,b,c)$. $d$ is the complement of $abc$ on the boundary. The boundary can be obtained by blocking as in Fig.~\ref{fig:corner-blocking} or applying the disentangler $W_{XY}$.}
  \label{fig:disentangler}
\end{figure}

This $(\fc)_{\text{min}}$ can be used to diagnose whether the edge is gappable:
\begin{theorem}
    \label{thm:ctot-min-gapped}
    For a gapped phase, $(\fc)_{\text{min}} = 0$ if and only if the gapped phase admits a gapped boundary condition. 
\end{theorem}
\begin{proof}
    \begin{itemize}
        \item We first show the $\Rightarrow$ direction. We can apply the blocking procedure as in Fig.~\ref{fig:corner-blocking} so that the corner regions $A \cap X, B \cup X, C \cap X$ in Fig.~\ref{fig:disentangler} (left) become an actual boundary intervals $a,b,c$ in Fig.~\ref{fig:disentangler} (right). Since this is merely a blocking, the entanglement entropies of various regions on the conformal ruler does not change. Therefore, the $\Delta_{(a,b,c)}$ and $I_{(a,b,c)}$ for the boundary conformal ruler defined in \cite{kim2024conformalgeometryentanglement} will be the same as the original $\Delta(\mathbb{D}),I(\mathbb{D})$ of the corners. Therefore, $\fc = 0$ implies $\Delta_{(a,b,c)} = 0$ or $I_{(a,b,c)} = 0$. This indicates the state on the boundary intervals $(a,b,c)$ or $(b,c,d)$ satisfies the quantum Markov condition \cite{kim2024conformalgeometryentanglement}, or in the entanglement bootstrap framework, satisfies \Aone condition of a gapped domain wall \cite{Shi:2020rne}. This implies that we have created a gapped boundary. 
        \item Now let us show the $\Leftarrow$ direction. We first can construct a disentangler $W_{XY}$ to create an actual hole in $X$ [Fig.~\ref{fig:disentangler} (right)]. The proof of the existence of such a disentangler is in App.~\ref{app:disentangler}. Since $W_{XY}\ket{\Psi}$ is a local unitary on $XY$, the state with the hole $W_{XY}\ket{\Psi}$ is still within the same gapped phase and in the set on which we are computing the minimum of $\fc$ in Eq.~\eqref{eq:def-ctot-min}. If the state admits a gapped boundary, then one can further apply some unitary on region $Y$ such that the edge state on interval $(a,b,c)$ or $(b,c,d)$ again satisfies Markov condition or the \Aone condition of gapped domain wall mentioned earlier. This essentially amounts to tune the boundary state to satisfy the 1d version of the area law. On such a state, $\Delta(\mathbb{D}) = \Delta_{(a,b,c)}$ or $I(\mathbb{D}) = I_{(a,b,c)}$ is zero and hence $\fc = 0$. We thus showed that $(\fc)_{\text{min}} = 0$. 
    \end{itemize}
\end{proof}

There are several caveats in this (physical) proof: First, we implicitly assume \Azero and \Aone in the argument. They are used in the following places: (1) We used \Azero to argue existence of the disentangler $U_{XY}$. This is because \Azero constrains the correlation to be finite. If \Azero is violated a lot, then one cannot disentangle $X$ simply by a local unitary $U_{XY}$ as in Fig.~\ref{fig:disentangler}. (2) We use \Aone to guarantee that $\fc(\ket{\Psi})$ is only about the corner and $\fc(U_{XY}\ket{\Psi})$ is only about the edge $\partial X$. Without \Aone, there could be contribution in $\fc$ from the bulk. Moreover, \Aone guarantees a sense of uniformity of the state $\ket{\Psi}$. Without \Aone, on some gapped state $\ket{\Psi_0}$, one can simply stacks a pure chiral gapped state or even more exotic degrees of freedom on a region $R$ where we are computing $\fc$, i.e. $AA'BCC' \subset R$, then locally one cannot learn anything universal about the original state $\ket{\Psi_0}$. Even though such a stacking will not decrease $(\fc)_{\text{min}}$, its value being zero or not has nothing to do with edge ungappability of the phase. 

The second caveat is that, technically we use an interval on the edge that forms a quantum Markov chain as an indication of the edge being gapped. One might think of a ``counterexample'' that on a 1+1D CFT groundstate upon some boost transformation as in \cite{Casini:2017roe}, the state is a quantum Markov chain for a certain decomposition. However, in this state with our definition $\fc$ is not zero but will still be the total central charge of that CFT. Since we define $\fc,\eta$ from the entanglement of the state, such a boost transformation effectively amounts to changing $\eta$. In fact, here we use the gapped domain wall entanglement bootstrap (EB) axiom as a definition of gapped boundary \cite{Shi:2020rne}. This is consistent with the common definition of a gapped boundary, because if the gapped domain wall EB axioms are satisfied, then one can show that a set of anyons can condense on the boundary \cite{Shi:2020rne}. Technically speaking, here we only show that $(\fc)_{\text{min}} = 0$ implies the boundary \Aone condition is satisfied on a certain interval, while the boundary \Azero condition in \cite{Shi:2020rne} is missing in order to complete the argument above into a mathematical proof. 

\subsection{Non-zero correlation length}
It is proved in \cite{no-go} that in a Hilbert space with finite local dimension, there is no zero correlation representative wavefunction of a chiral gapped phase, where the chirality is reflected by the non-zero chiral central charge or electric Hall conductance. One reflection of the non-zero correlation length is the violation of the strict area law together with the corner contribution. Explicitly, on such a representative state, for a disk $A$ with a sharp corner of angle $\theta$, 
\begin{align}\label{eq:area-law-corner}
    S_A = \alpha |\partial A| - \gamma + f(\theta) + \cdots,
\end{align}
where $f(\theta)$ is the corner dependent term that satisfies $f(\theta) = f(2\pi - \theta)$. The no-go theorems in \cite{no-go} shows that the subleading term $\cdots$ cannot be zero as long as $S_A$ is finite. As we mentioned earlier in this paper, chiral gapped phases also have a ``robust'' corner contribution, reflected by $(\fc)_{\text{min}}\neq 0$ which is closely related to the edge ungappability. It is plausible that, instead of assuming non-zero chiral central charge or electric Hall conductance, one can also reach the same no-go statement from the ``robust'' corner contribution or edge ungappability. Below we will give a hand-waving argument of this statement. 

\begin{figure}[htb]
  \begin{tikzpicture}
    \filldraw[color = yellow!10, draw = blue, thick] (0,0) circle (2); 
    \draw[blue, thick] (0,0) -- (-90:2); 
    \draw[blue, thick] (0,0) -- (30:2);
    \draw[blue, thick] (0,0) -- (150:2);
    \filldraw[color = yellow!10, draw = blue, thick] (15:2) -- (15:3) arc (15:180-15:3) -- (180-15:2) arc (180-15:15:2); 
    \draw[->, red] (-90:0.5) arc (-90:-90+360:0.5);  
    \node at (-30:1) {$A$}; 
    \node at (90:1) {$B$}; 
    \node at (210:1) {$C$}; 
    \node at (90:2.5) {$D$};
    \draw[|<->|] ($(0.15,0)+(0,0)$) -- ($(0.15,0)+(-90:2)$);  
    \node at ($(0.3,0)+(-90:1)$) {$R$}; 
  \end{tikzpicture}
  \caption{A corner conformal ruler. The angle of the corner at the triple intersection point of $A,B,C$ reaches $2\pi$. $R$ is the radius of the disk $ABC$.}
  \label{fig:corner-2pi}
\end{figure}
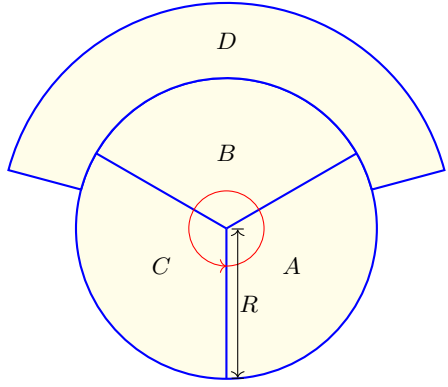

Let us consider a corner conformal ruler under the limit that the corner at the triple intersection point of $A,B,C$ is reaching $2\pi$ [Fig.~\ref{fig:corner-2pi}]. We will estimate $\fc$ computed from this corformal ruler with the radius $R$ of the disk $ABC$ is large. Assuming Eq.~\eqref{eq:area-law-corner}, we can estimate 
\begin{align}\label{eq:scaling-Delta-I}
    \Delta(\mathbb{D}) \sim e^{-R / \xi_{\Aone}},\quad I(\mathbb{D}) \sim 2 \alpha R,
\end{align}
Notice now the $\Delta(\mathbb{D})$ is just the usual bulk \Aone linear combination in Assumption~\ref{assumption:A1}. In $\Delta(\mathbb{D})$, the area law, TEE and the corner terms are all canceled and we assume the error is exponentially decaying $\geq e^{-R / \xi_{\Aone}}$, which defines a length scale $\xi_{\Aone}$. In many cases, such as we will numerically show in the next section, the bulk \Aone combination is indeed exponentially decay. There are also cases as reported in \cite{Vir} that the \Aone combination can algebraically decay within the same phase. That will give a larger $\Delta(\mathbb{D})$ and as a result, as shown in \cite{kim2024conformalgeometryentanglement} we will obtain a larger $\fc$ solved from Eq.~\eqref{eq:def-c-eta}. Since here we are interested in the minimal $\fc$ within the phase, we shall therefore mainly focus on the cases of exponential decay. For $I(\mathbb{D})$, the area law term isn't canceled and hence the leading order term is $2 \alpha R$, where $\alpha$ is the coefficient in Eq.~\eqref{eq:area-law-corner}. 

Now we can solve for $\fc(\Delta, I)$ follows from the definition Eq.~\eqref{eq:def-c-eta}
\begin{align}\label{eq:ctot-xi}
    \fc(\Delta(\mathbb{D}), I(\mathbb{D})) \sim 12 \alpha \xi_{\Aone}. 
\end{align}
We leave the explicit derivation in App.~\ref{app:ctot-xi}. 

This relation indicates that, if $\fc$ has a non-zero minimal value within the phase, then $\xi_{\Aone}$ has to be non-zero unless $\alpha$ is infinite. Indeed, in the computation of entanglement entropy in the continuum field theory such as \cite{Wong:2017pdm}, the area law coefficient is related to the a UV cutoff by $\alpha \sim \frac{1}{\epsilon}$, which approaches to infinite as $\epsilon \to 0$. This indicates an infinite dimensional local Hilbert space is necessary in order to have a zero correlation RG fixed point. 

We comment that one can also use the setup in Fig.~\ref{fig:corner-2pi-2}, where we define $\fc$ as the solution of 
\begin{equation}
    e^{-6 I(A:C|B)/\fc} + e^{-6 I(A:D|B)/\fc} = 1.
\end{equation}
This version focus more on the violation of the Markov condition $I(A:D|B) \sim e^{-\ell / \xi_{\text{Markov}}}$. By studying the limit with $\ell \to \infty$, one can obtain the same conclusion as Eq.~\eqref{eq:ctot-xi}, with $\xi_{\Aone}$ replaced with $\xi_{\text{Markov}}$. 

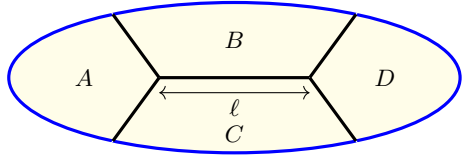
\begin{figure}[htb]
\begin{tikzpicture}
    \filldraw[fill = yellow!10, draw = blue, very thick, name path = ellipse] (0,0) ellipse (3 and 1); 
    \draw [white, name path = ab] (-1,0) -- (150:2); 
    \draw[very thick, name intersections={of=ellipse and ab}] (-1,0) -- (intersection-1);
    \draw [white, name path = ac] (-1,0) -- (-150:2); 
    \draw[very thick, name intersections={of=ellipse and ac}] (-1,0) -- (intersection-1);
    \draw[very thick] (-1,0) -- (1,0); 
    \draw [white, name path = cd] (1,0) -- (30:2); 
    \draw[very thick, name intersections={of=ellipse and cd}] (1,0) -- (intersection-1);
    \draw [white, name path = bc] (1,0) -- (-30:2); 
    \draw[very thick, name intersections={of=ellipse and bc}] (1,0) -- (intersection-1);

    \draw[<->] (-1,-.2) -- (1,-.2);
    \node at (0, -0.4) {$\ell$}; 
    \node at  (-2,0) {$A$};
    \node at  (2,0) {$D$};
    \node at  (0,0.5) {$B$};
    \node at  (0,-0.75) {$C$};
\end{tikzpicture}
\caption{Another setup that is equivalent to Fig.~\ref{fig:corner-2pi}}
\label{fig:corner-2pi-2}
\end{figure}

\section{Numerics}
\label{sec:num}
In this section, we numerically verify the results mentioned in the previous sections, mainly in the example of the $p+\ii p$ superconductor (SC) groundstate. Explicitly, we are going to test 
(1) violation of \Azero and \Aone [Assumption~\ref{assumption:A0} and Assumption~\ref{assumption:A1}]; 
(2) Computation of $\fc$ and $\eta$ from the corner conformal ruler [Eq.~\eqref{eq:def-c-eta}]; 
(3) Vector fixed point equation [Eq.~\eqref{eq:VFE-implied}]; 
(4) Modular commutator with incomplete disks [Eq.~\eqref{eq:J-eta}]; 
(5) Finite size error reduction with $H_{\rec}$. 

\subsection{Setup and test of \Azero and \Aone}
We shall study the groundstate of a $p+\ii p$ SC on a square lattice, with the Hamiltonian of the form 
\begin{align}\label{eq:pip-ham}
     H = \sum_{\rr, \aa} \[ -t a_{\rr}^{\dagger} a_{\rr + \aa} + \Delta a_{\rr}^{\dagger} a_{\rr + \aa}^{\dagger} e^{\ii \aa \cdot \AA} + h.c. \] - \sum_{\rr} (\mu - 4t)a_{\rr}^{\dagger} a_{\rr},
\end{align}
where $\rr = (x,y)$ denotes the location a site on the square lattice, $\aa$ is the lattice vector that takes value of $(1,0)$ and $(0,1)$ for the positive $x$ and $y$ direction respectively. The gauge field is $\AA = (0, \pi/2)$. In all the tests mentioned below, we shall set $t = 1.0$. We will vary the superconducting gap $\Delta$ and the chemical potential $\mu$ within the chiral gapped phase. Explicitly, we will choose $(\Delta,\mu) = (1.0, 1.2), (1.2, 1.4),(1.3,1.2)$, whose correlation lengths will be determined later. Throughout all the tests, we shall set anti-periodic boundary condition in both $x$ and $y$ direction, so that there is no gapless edge and no flux threaded. We shall use $N_x, N_y$ to denote the number of sites (complex fermion sites, to be more precise) in $x$ and $y$ direction, so that the total number of sites is $N_x \times N_y$. We will specify $N_x, N_y$ later. 

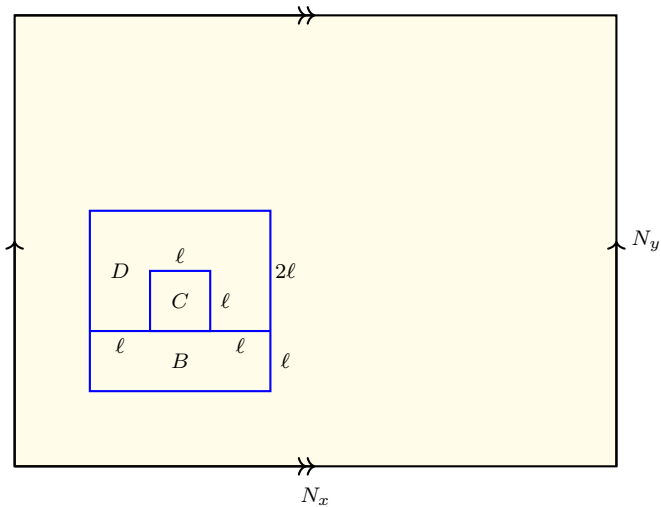
\begin{figure}[htb]
    \centering
    \begin{tikzpicture}
    \def\l{0.8}
    \filldraw[color = yellow!10,draw = black,thick] (0,0) -- (8,0) -- (8,6) -- (0,6) -- cycle; 
    \draw[->>, thick] (0,0) -- (4, 0); 
    \draw[->>, thick] (0,6) -- (4, 6);
    \draw[->, thick] (0,0) -- (0,3); 
    \draw[->, thick] (8,0) -- (8,3);
    \draw[color = blue, thick] (1,1) -- ({1+3*\l},1) -- ({1+3*\l},{1+3*\l}) -- (1,{1+3*\l}) -- cycle; 
    \draw[blue, thick] (1+\l,1+\l) -- (1+\l+\l,1+\l) -- (1+\l+\l,1+\l+\l) -- (1+\l,1+\l+\l) -- cycle; 
    \draw[thick, blue] (1,1+\l) -- ({1+3*\l},1+\l); 
    \node[font = \footnotesize] at ({1+\l/2}, 1+\l-0.2) {$\ell$}; 
    \node[font = \footnotesize] at ({1+\l+\l/2}, 1+\l+\l+0.2) {$\ell$}; 
    \node[font = \footnotesize] at ({1+2*\l+\l/2}, 1+\l-0.2) {$\ell$}; 
    \node[font = \footnotesize] at ({1+2*\l+0.2}, {1+\l+\l/2}) {$\ell$}; 
    \node[font = \footnotesize] at ({1+3*\l+0.2}, {1+\l/2}) {$\ell$}; 
    \node[font = \footnotesize] at ({1+3*\l+0.2}, {1+\l+\l}) {$2\ell$}; 

    \node[font = \footnotesize] at ({1+\l+\l/2}, {1+\l/2}) {$B$};
    \node[font = \footnotesize] at ({1+\l+\l/2}, {1+\l+\l/2}) {$C$}; 
    \node[font = \footnotesize] at ({1+\l/2}, {1+2*\l}) {$D$};
    
    \node[font = \footnotesize] at (4, -0.4) {$N_x$};
    \node[font = \footnotesize] at (8+0.4, 3) {$N_y$};
  \end{tikzpicture}
	\caption{Subsystems used to test \Azero and \Aone.}
	\label{fig:EB-regions-tests}
\end{figure}

We first test the entanglement bootstrap assumption [Assumption~\ref{assumption:A0} and Assumption~\ref{assumption:A1}], and then determine the correlation length. In these tests, we choose $N_x = N_y = 60$, and the subsystems for \Azero and \Aone is shown in Fig.~\ref{fig:EB-regions-tests}, whose typical linear size is $\ell$. We computed the $\Delta(C,BD,\emptyset)$ and $\Delta(B,C,D)$ as a function of $\ell$. The result is shown in Fig.~\ref{fig:err-A0-A1}
\begin{figure}[htp]
    \centering
    \includegraphics[width = 0.45\columnwidth]{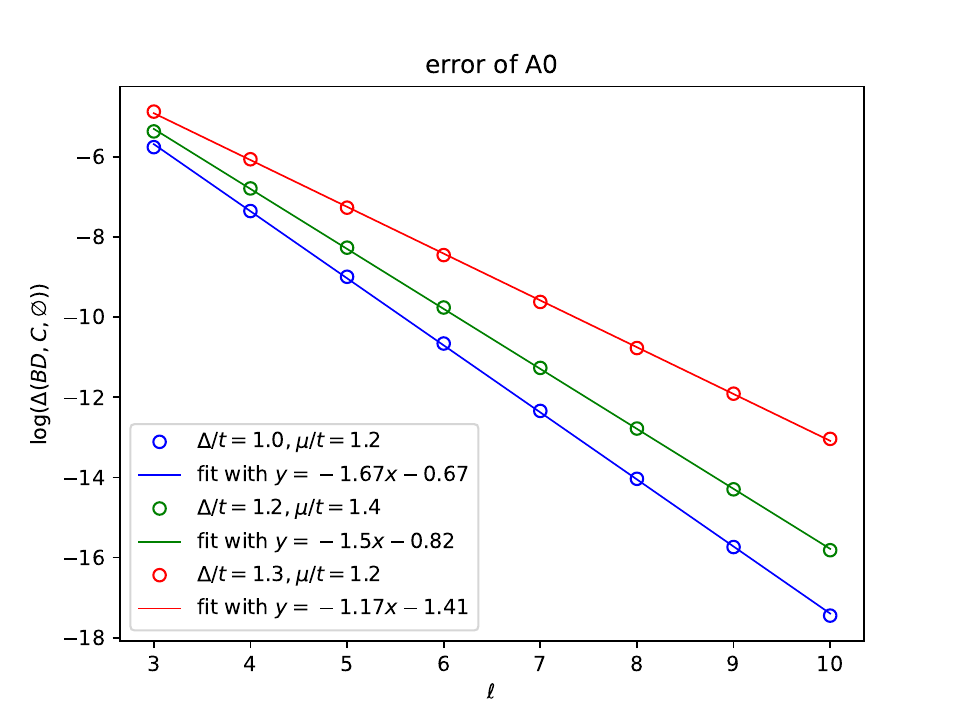}
    \includegraphics[width = 0.45\columnwidth]{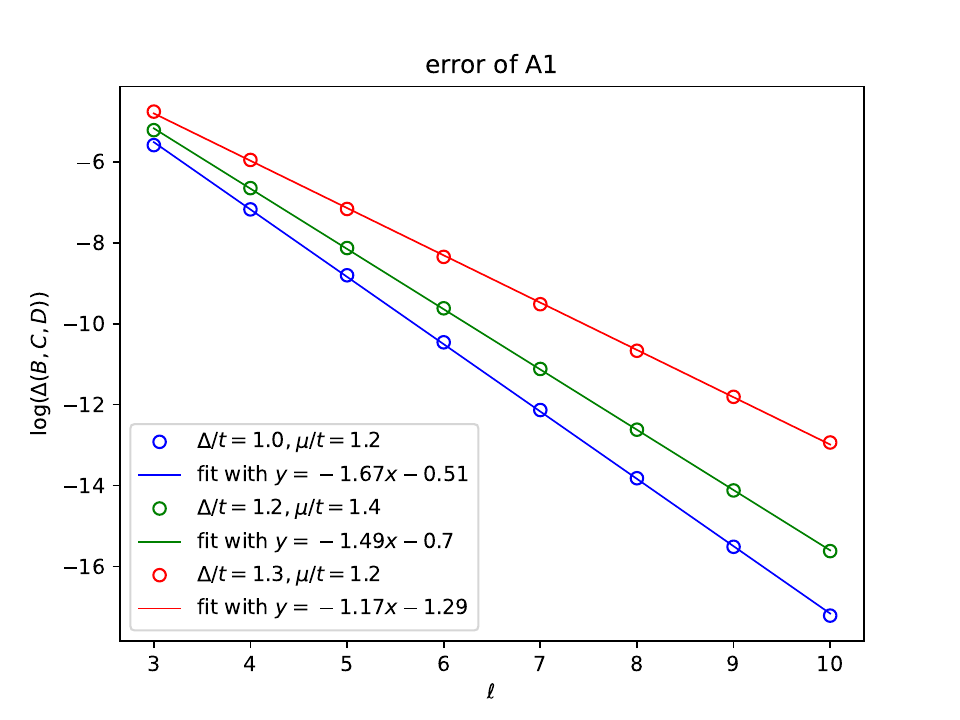}
	\caption{Error of \Azero and \Aone.}
	\label{fig:err-A0-A1}
\end{figure}
By fitting $\log(\Delta(C,BD,\emptyset))$ and $\log(\Delta(B,C,D))$ as shown in the figure, one can see that $\Delta(C,BD,\emptyset)$ and $\Delta(B,C,D)$ exponentially decay to zero as one increase the subsystem sizes. This indicates, even if \Azero and \Aone cannot be satisfied exactly for chiral states with finite local dimensions, the discrepancy decays as one increases the size of the subsystems. Therefore, if we zoom out to a certain scale, one can ignore the error of \Azero and \Aone and still derive conclusions, which is expected to hold at the IR fixed point\footnote{The result in \cite{no-go} suggests that the local Hilbert space dimension, i.e. the dimension of the Hilbert space associated with a local region, has to be infinite at the IR fixed point.}.

We can also infer the correlation length from \Azero, this is because $\Delta(BD,C,\emptyset)$ upper bound any correlation functions between region $C$ and any region $A$ buffered away from $C$ by $BD$. See the numerical section in \cite{Vir} for detail explanations. In the numerical tests, we use the same setup as shown in Fig.~\ref{fig:EB-regions-tests}, except the size of $C$ is fixed. When we  scale up the thickness $\ell$ of the buffer region $BD$. as shown in Fig.~\ref{fig:A0-correlation}, $\Delta(BD,C,\emptyset)$ decays as $\alpha e^{-2\ell/\xi_*}$\footnote{Here we put the factor of $2$ in the exponent because technically \Azero upper bounds the correlation function squared.}, then we can conclude $\xi_* \geq \xi$, where $\xi$ is the correlation length from $\langle \CO_A \CO_C \rangle - \langle \CO_A \rangle \langle \CO_C \rangle \sim \beta e^{-\ell/\xi}$ for any operators $\CO_A,\CO_C$ supported on $A,C$. Therefore, we can use $\xi_*$ as a correlation length. From the slopes of the linear fit in Fig.~\ref{fig:err-A0-A1} (left), we obtain the correlation length is $\xi_* = 1.06, 1.17,  1.45$ for $(\Delta,\mu) = (1.0, 1.2), (1.2, 1.4), (1.3,1.2)$ respectively. 
\begin{figure}[htb]
    \centering
    \includegraphics[width = 0.5\columnwidth]{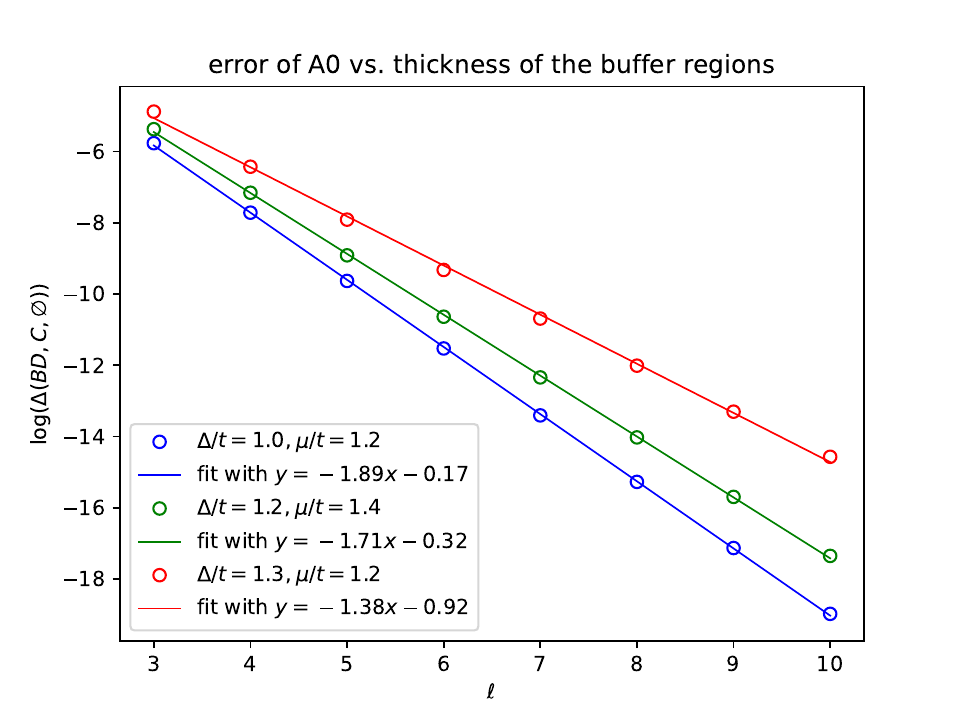}
	\caption{Obtaining an upper bound of the correlation length. In the numeircs, region $C$ is fixed and we scale up the thickness $\ell$ of the buffer region $BD$ [Fig,~\ref{fig:EB-regions-tests}].}
	\label{fig:A0-correlation}
\end{figure}

\subsection{Tests about corner entanglements}
Now we are going to test the hypothesis~\ref{hypo:corner} and its predictions. Explicitly, we shall test (1) computation of $\fc$ and $\eta$ from $\Delta, I$; (2) vector fixed-point equation around a corner; (3) the modular commutator with incomplete disk $J(A,B,C) = \pi c_{-}/3 \cdot \eta$. 

\begin{figure}[htb]
    \centering
    \includegraphics[width = 0.45\columnwidth]{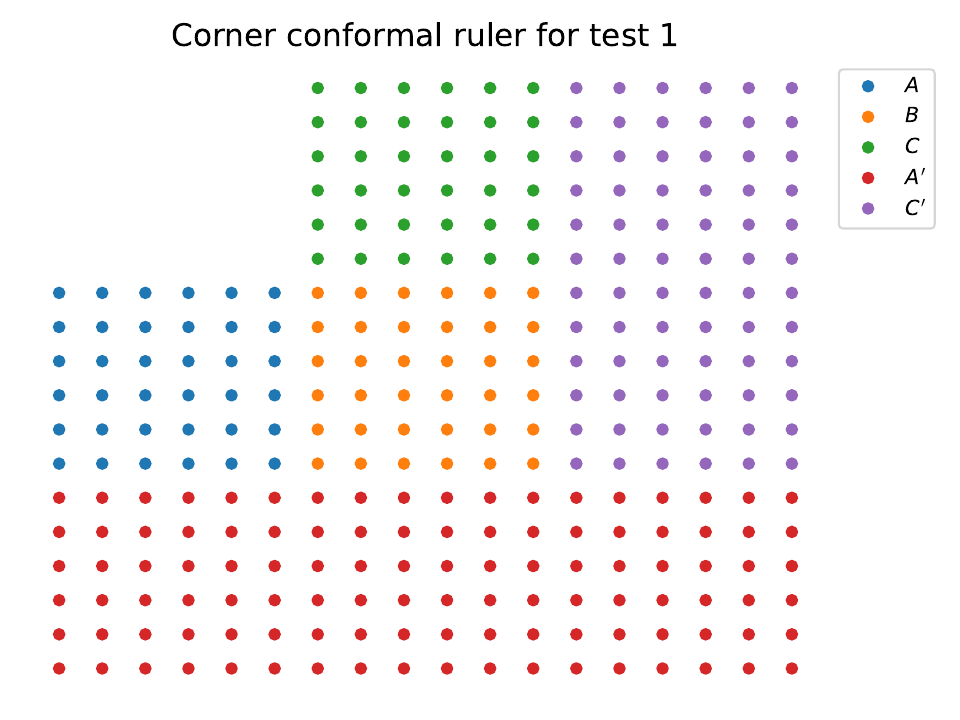}
    \includegraphics[width = 0.45\columnwidth]{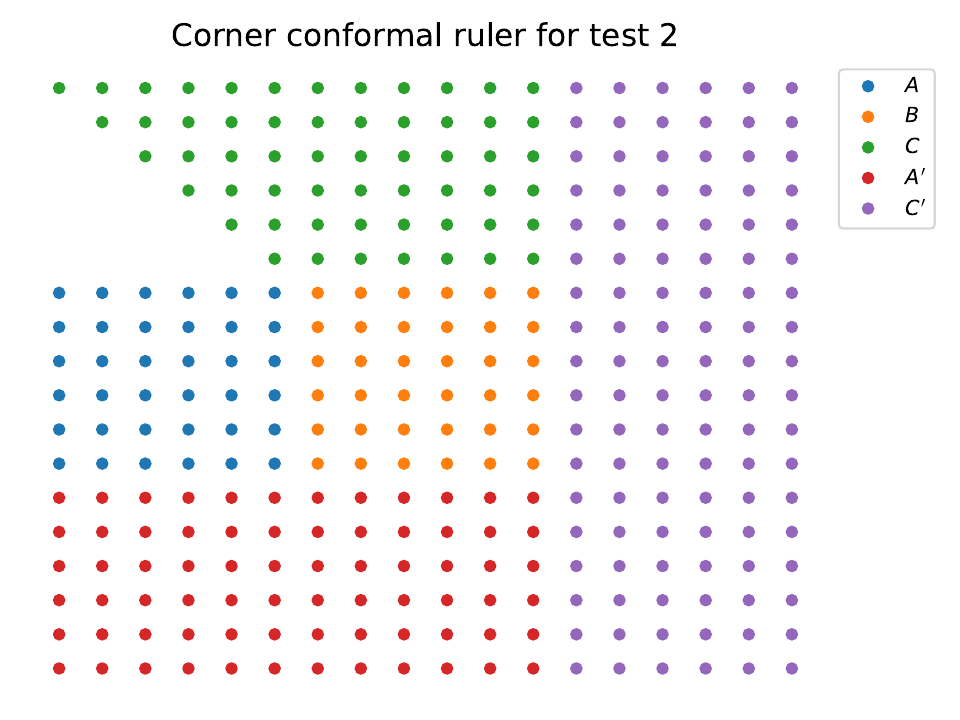}
	\caption{Corner conformal rulers for test 1 (left) and 2 (right).}
	\label{fig:corner-conf-ruler-test12}
\end{figure}

For the tests in this subsection, we set the system sizes to be $N_x = N_y = 40$. The subsystems for the ``corner conformal rulers'' are shown in Fig.~\ref{fig:corner-conf-ruler-test12}. Using this setup, we first compute $\fc$ and $\eta$ from $\Delta(AA',B,CC')$ and $I(A:C|B)$. Then we compute $J(A,B,C)$ which we expected to be $\pi c_{-}/3 \eta$ with $c_{-} = 1/2$ for $p+\ii p$ SC. To test this formula, we extract $\eta_J \equiv 3J(A,B,C)/(c_{-}\pi)$ and compare it with the $\eta$ obtained from $\Delta,I$. The results are listed in Table~\ref{tab:result-1} and Table~\ref{tab:result-2}. One can see that, as the correlation length is decrease, $\fc$ is approaching to $1/2$ and $\eta$ and $\eta_J$ become closer. 

\begin{table}[htb]
\begin{tabular}{|l|l|l|l|}
\hline
$(\Delta,\mu)$ and $\xi$ & (1.0, 1.2), $\xi = 1.06$ & (1.2, 1.4) , $\xi = 1.17$ & (1.3, 1.2) , $\xi = 1.45$ \\ \hline
$\fc$ & 0.569 & 0.645 & 0.691 \\ \hline
$\eta$ & 0.49993 & 0.49985 & 0.49948 \\ \hline
$\eta_J$ & 0.49998 & 0.49994 & 0.49965 \\ \hline
\end{tabular}
\caption{Result for test 1.}
\label{tab:result-1}
\end{table}

\begin{table}[htb]
\begin{tabular}{|l|l|l|l|}
\hline
$(\Delta,\mu)$ and $\xi$ & (1.0, 1.2), $\xi = 1.06$ & (1.2, 1.4) , $\xi = 1.17$ & (1.3, 1.2) , $\xi = 1.45$ \\ \hline
$\fc$ & 0.595 & 0.695 & 0.764 \\ \hline
$\eta$ & 0.96546 & 0.95805 & 0.94551 \\ \hline
$\eta_J$ & 0.96371 & 0.95616 & 0.94476 \\ \hline
\end{tabular}
\caption{Result for test 2}
\label{tab:result-2}
\end{table}

These results also emphasize the necessity to use conformal ruler as a measure of the angles. In a lattice, the geometric corners are not quite well defined. For a square lattice, in test 1, it is plausible to identity the corners for $A,B,C$ as $\pi/2$. The geometric cross-ratio computed using Eq.~\eqref{eq:eta-geo} is $eta_g = 1/2$, which agrees with the $\eta$ obtained using entanglement entropies $\Delta,I$ or using modular commutators $J(A,B,C)$. However, in test 2, the angle for the corner is hard to decide. From the geometry, one might think the angles for the sharp corners in $A,B,C$ are $\pi/2,\pi/2,3\pi/4$ respectively, then the geometric cross-ratio will be $\eta_g = 0.707$. This turns out to be incorrect. To obtain a consistent result, one should use $\eta$ computed from $\Delta,I$. In summary, an important lesson from this test is that, in the corner entanglement entropy formula $S_{\corner} = \frac{\fc}{6}\log(\sin(\varphi/2))$, the angle $\varphi$ should be decided using the conformal ruler. 

To test the vector fixed-point equation, we compute $\sigma(K_{\mathbb{D}}(x)) = \sqrt{\langle K_{\mathbb{D}}(x)^2 \rangle - \langle  K_{\mathbb{D}}(x) \rangle^2 }$, where 
\begin{align}
     K_{\mathbb{D}}(x) = x \hat{\Delta}(AA',B,CC') + (1-x) \hat{I}(A:C|B)
\end{align}
for various $x$. The results are shown in Fig.~\ref{fig:err-VFE}. To make comparison, we also plot the values of $\eta$ from $\Delta, I$ and $\eta_J$. Similar to the results above, the error of the vector fixed-point equation decreases as the correlation length decreases, and the location of the minimum is approaching to $\eta$ from $\Delta, I$.

\begin{figure}[htb]
    \centering
    \includegraphics[width = 0.45\columnwidth]{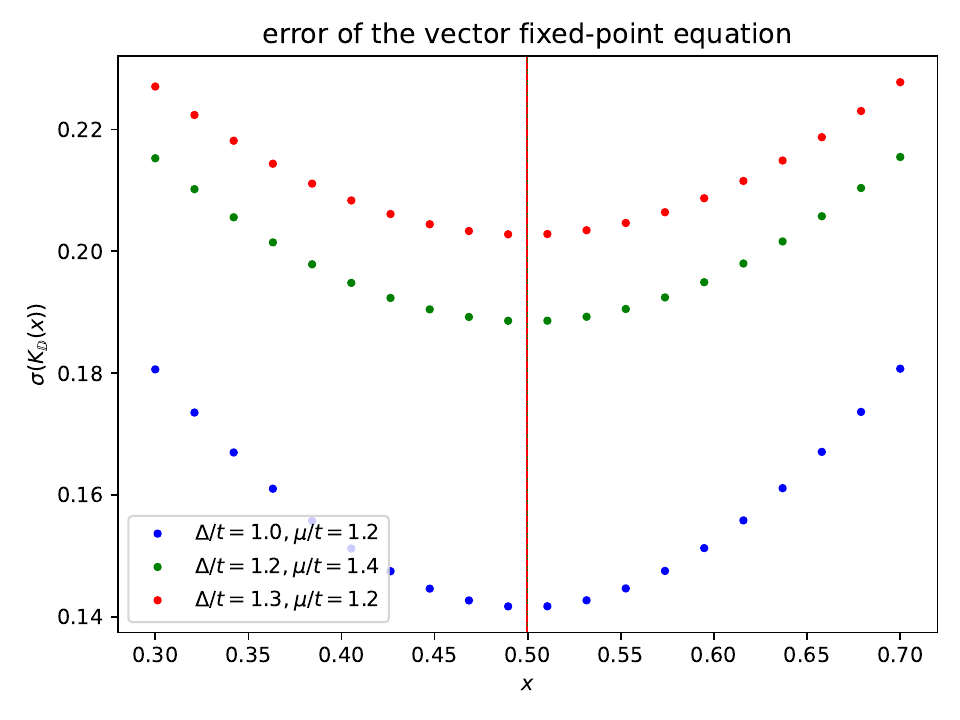}
    \includegraphics[width = 0.45\columnwidth]{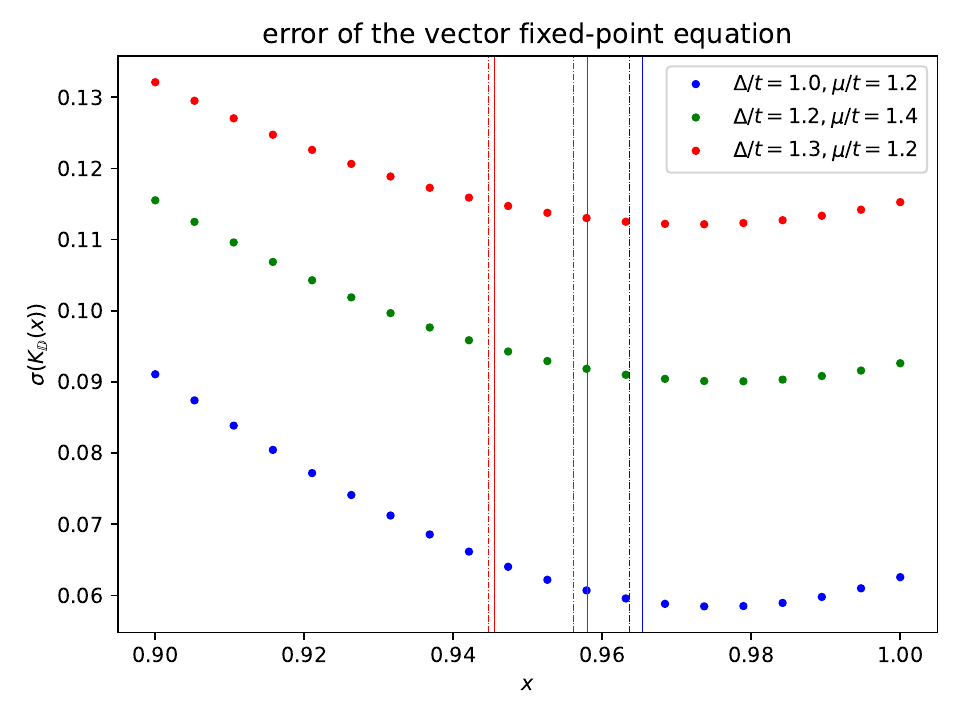}
	\caption{Error of the vector fixed point equations for test 1 (left) and test 2 (right). In the plot, the solid lines stand for the values of $\eta$ obtained from $\Delta,I$ and the dashed line stands for the values of $\eta_J$. In the left plot, since all the $\eta$ and $\eta_J$ are very close, the lines overlap.}
	\label{fig:err-VFE}
\end{figure}

\subsection{Error reduction with gradient descent using $H_{\rec}$}
\label{subsec:error-reduction}
In this subsection, we introduce a novel gradient descent scheme which can reduce the discrepancy between the numerical results and our predictions. The gradient descent can be understood as a procedure that drives the state closer to the zero correlation length IR fixed point of the phase, and therefore reduces the finite size errors.

Let us first briefly introduce the scheme. In Section~\ref{subsec:Hrec}, we introduced several reconstructed Hamiltonian such as $\Hrec^{\euler}$ and $H_{\text{rec}}^{\Delta}$ in the coarse-grained lattice.
The object function of the gradient descent is the expectation value 
$f^{\text{Euler}}(\ket{\psi}) = \langle \psi| \Hrec^{\euler} | \psi\rangle$ 
or 
$f^{\Delta}(\ket{\psi}) = 
\langle \psi| H_{\text{rec}}^{\Delta} | \psi\rangle$ 
as a function of $\ket{\psi}$, with the reconstructed Hamiltonians obtained from $\ket{\psi}$. 
By performing gradient descent, one reaches a local minimum of the object function. Let us specify the topology of the spatial manifold to be a torus, as this is the case for our numerics.
Then at the zero correlation RG fixed point $\ket{\Psi_*}$, we expect $f^{\text{Euler}}(\ket{\Psi}) = f^{\Delta}(\ket{\Psi_*}) = 0$. 
Notice both functions are lower bounded by zero, therefore $\ket{\Psi_*}$ in fact is a location of a global minimum. Therefore, if one starts with a nearby state, the gradient decent shall drive the state to the fixed point state $\ket{\Psi_*}$. 
In numerics, since we are working with a lattice with finite local dimensions, as shown in \cite{no-go}, there is no zero correlation IR fixed point of a chiral gapped phase within such a Hilbert space. Therefore, we do not expect one can obtain a minimum with the object functions being zero. Neverthless, since in $\langle H^{\Delta}_{\text{rec}}\rangle$ is a sum over all possible semi-positive \Aone entropy combinations $\Delta_i \geq 0$ in Eq.~\eqref{eq:A1}, runing gradient descent of $f^{\Delta}$ will reduce the error of \Aone. Moreover, we conjecture that runing gradient decent of $f^{\text{Euler}}$ will also achieve this goal, which we will numerically verify below. 

In practice, we are doing the ``gradient descent'' to obtain a sequence of states $\ket{\Psi_0},\ket{\Psi_1},\ket{\Psi_2},\cdots$ as follows: In the $k$-th step, we first use $\ket{\Psi_{k-1}}$ to obtain a $H_{\text{rec}}$, then we find $\ket{\Psi_{k}}$ as the groundstate of $H_{\text{rec}}$. The reason that we did not use the usual gradient descent method, namely obtaining $\ket{\Psi_{k}} = \ket{\Psi_{k-1}} - \epsilon \nabla f(\ket{\Psi_{k-1}})$ because the new state $\ket{\Psi_{k}}$ obtained this way might not be a gaussian state, which forbids us to do the computation with large system sizes. 

Strictly speaking, our method of obtaining $\ket{\Psi_{k}}$ from $\ket{\Psi_{k-1}}$ is only an approximate gradient descent, as the state is not perturbed exactly along the direction along which $f(\ket{\Psi})$ decreases the most. It is still going to decrease the value of $f(\ket{\Psi})$ along this sequence, if $\ket{\Psi}$ is already close to the local minimum. This can be understood as follows: Given a state $\ket{\Psi_k}$ and let us consider $H_{\text{rec}}^k = H_{\text{rec}}^{\Delta}$ constructed from $\ket{\Psi_k}$. Let $\ket{i}$ be the eigenstates of $H_{\text{rec}}^k$ with eigenvalues $\lambda_i$, we can write $\ket{\Psi_k} = \sum_i c_i \ket{i}$. The condition that $\ket{\Psi_k}$ is closed to the minimum implies $f(\ket{\Psi_k}) \approx 0$, which means $|\lambda_0| \approx 1$ and $\lambda_0 \approx 0$, where $\ket{0}$ is the groundstate of $H_{rec}^{k}$. That is, we can write 
\begin{align}
     \ket{0} = \frac{\ket{\Psi_k} - \sum_{i\neq 0} c_i \ket{i}}{c_0} = \frac{\ket{\Psi_k} - \epsilon \ket{\Psi'_k}}{\sqrt{1-\epsilon^2}} \equiv \ket{\Psi_{k+1}}
\end{align}
where $\epsilon = \sqrt{\sum_{i\neq 0}|c_i|^2} \approx 0$ and $\epsilon \ket{\Psi'_k} = \sum_{i\neq 0}c_i \ket{i}$. In our procedure, this groundstate $\ket{0}$ is the next state $\ket{\Psi_{k+1}}$, which can be understood by the state perturbed by $\ket{\Psi'_k}$ from $\ket{\Psi_k}$. Now we can write $H^{k+1}_{\text{rec}} = H^{k}_{\text{rec}} + \epsilon \cdot \partial H^{k}_{\text{rec}} + O(\epsilon^2)$ and $\ket{\Psi_{k+1}} = \ket{\Psi_k} + \epsilon \ket{\partial\Psi_k} + O(\epsilon^2) $, therefore we can expand $f(\ket{\Psi_{k+1}})$ as 
\begin{align}
    f(\ket{\Psi_{k+1}}) &= \langle \Psi_{k+1} | H_{\text{rec}}^{k+1} | \Psi_{k+1} \rangle \\ 
    & = \langle \Psi_{k} | H_{\text{rec}}^{k} | \Psi_{k} \rangle + \epsilon \langle \Psi_{k} | \partial H_{\text{rec}}^{k} | \Psi_{k} \rangle - \epsilon ( \langle \Psi_k| H_{\text{rec}}^k |\partial \Psi_k\rangle + c.c.) + O(\epsilon^2). 
\end{align}
We can further evaluate the $O(\epsilon)$ terms. Firstly, we have $ \langle \Psi_{k} | \partial H_{\text{rec}}^{k} | \Psi_{k} \rangle = 0$. This is because $H_{\text{rec}}^{k}$ is a linear combinations of entanglement Hamiltonians from $\ket{\Psi_k}$ and as shown in \cite{Lin:2023pvl} there is a quantum information-theoretic result that $\langle \psi | \partial K_{\bullet} | \psi\rangle $ where $\partial K_{\bullet}$ is the first order derivative of the entanglement Hamiltonian from $\ket{\Psi}$ resulting from any norm-preserving perturbation of $\ket{\psi}$. Then, we can show that 
\begin{align}
     f(\ket{\Psi_{k+1}}) - f(\ket{\Psi_{k}}) &= - \epsilon \langle \Psi_k| H_{\text{rec}}^k | \partial \Psi_k\rangle + c.c. = - \epsilon \langle \Psi_k | H_{\text{rec}}^k | \Psi'_{k}\rangle + c.c. \\ 
     & = -2 \sum_{i \neq 0} |c_i|^2 \lambda_i \leq 0. 
\end{align}
The last inequality is because the operator $H_{\text{rec}}^k$ is a positive operator and $\lambda_i \geq 0$ for all its eigenvalues. In the argument above, we use $H_{\text{rec}}^{\Delta}$ as the reconstructed Hamiltonian because it is an positive operator. In practice, one can also use $\Hrec$, because we are working with lattice with finite system sizes, this operator $\Hrec$ is bounded from below and one can repeat the argument with $\Hrec + \alpha \mathbbm{1} \geq 0$ with $\alpha$ being a sufficiently large finite number. 

Now we present the numerical results. The setup is the same as in Fig.~\ref{fig:EB-regions-tests} for \Azero and \Aone tests and Fig.~\ref{fig:corner-conf-ruler-test12} for the tests of corner entanglement, with $N_x = N_y = 36$. We made a coarse-grained lattice where each supersite is a $6\times 6$ square. To run the gradient descent, we use 
\begin{align}
    \Hrec = \sum_{f} K_f - \sum_e K_e + \sum_v K_v,
\end{align}
where $f,e,v$ runs over all the 4-supersite, 2-supersite and 1-supersite respectively, as shown in Fig.~\ref{fig:Hrec-num}. This is slightly different from the one we introduced above in Section~\ref{subsec:Hrec}, but in the same spirit. Following the derivation in App.~\ref{app:Hrec}, one can still show that, if \Aone is satisfied, it is equal to a $H_{\rec}^{\Delta}$ obtained from summing over all possible $\hat{\Delta}$ on a square (coarse-grained) lattice. Hence, we conjecture that if we start with a state with small \Aone violation, then it is possible that one can run gradient descent of $\langle H_{\rec}\rangle$ to reduce the error of \Aone. 

\begin{figure}[htb]
    \centering
    \includegraphics[width = 0.8\columnwidth]{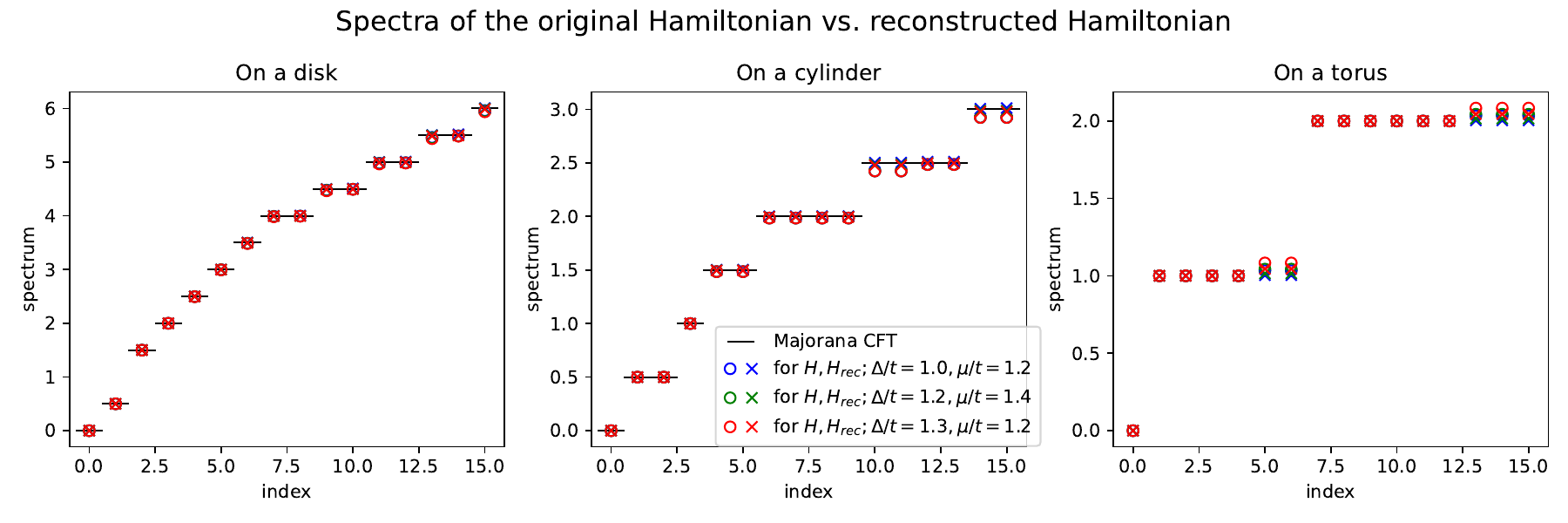}
	\caption{Low lying spectrum of the reconstructed Hamiltonians vs. the original Hamiltonians for three choices of the parameters $(\Delta/t, \mu/t)$. The systems are on a disk (left), a cylinder (middle) and a torus (right) respectively, with sizes $N_x = N_y = 36$. For the results on a disk and cylinder, we compared the spetrua with the chiral and non-chiral CFT spectra, and we rescaled the spectra so that the energy level is integer spacing in each confomral block. For the result on the torus, we rescale the spectra so that the gap is one.}
	\label{fig:Hrec-spectrum}
\end{figure}

We first checked the spectrum of the reconstructed Hamiltonians for the systems on a disk, cylinder and torus. We use the same systme sizes and coarse-grained lattice as Fig.~\ref{fig:Hrec-spectrum}. In Fig.~\ref{fig:Hrec-spectrum}, we compare the reconstructed spectra with the spectra of the the original Hamiltonians Eq.~\eqref{eq:pip-ham}. The results in Fig.~\ref{fig:Hrec-spectrum} show excellent agreements.  If the systems have bounadries, we expect the low-lying spectra will match with the edge CFT spectra. On a disk and cylinder, we expect the low-lying spectrum will be the chiral and non-chiral Majorana CFT respectively. The results in Fig.~\ref{fig:Hrec-spectrum} (left, middle) confirms this expectation. Furthermore, the reconstructed spectra show a slightly better agreement with the CFT spectra than spectra from the orignal Hamiltonians. 

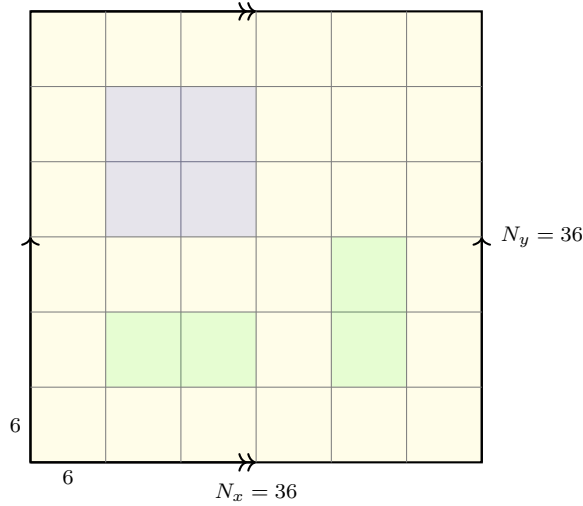
\begin{figure}[htb]
    \centering
    \begin{tikzpicture}
    \def\l{6}
    \def\d{1}
    \filldraw[color = yellow!10,draw = black,thick] (0,0) -- (\l,0) -- (\l,\l) -- (0,\l) -- cycle; 
    \draw[->>, thick] (0,0) -- ({\l/2}, 0); 
    \draw[->>, thick] (0,\l) -- ({\l/2}, \l);
    \draw[->, thick] (0,0) -- (0,{\l/2}); 
    \draw[->, thick] (\l,0) -- (\l,{\l/2});  

    \foreach \x in {1,...,5}
		  \draw[gray] (0,{\x*\l/6}) -- (\l, {\x*\l/6}); 
    \foreach \x in {1,...,5}
      \draw[gray] ({\x*\l/6},0) -- ({\x*\l/6},\l);

    \filldraw[color = blue, draw = gray, opacity = 0.1] (\d,{3*\d}) -- ({\d+2*\d},{3*\d}) -- ({\d+2*\d},{3*\d+2*\d}) -- (\d,{3*\d+2*\d}) -- cycle; 

    \filldraw[color = green, draw = gray, opacity = 0.1] (\d,{\d}) -- ({\d+2*\d},{\d}) -- ({\d+2*\d},{\d+\d}) -- (\d,{\d+\d}) -- cycle;
    
    \filldraw[color = green, draw = gray, opacity = 0.1] ({4*\d},{\d}) -- ({4*\d+\d},{\d}) -- ({4*\d+\d},{\d+2*\d}) -- ({4*\d},{\d+2*\d}) -- cycle;
    
    \node[font = \footnotesize] at ({\l/2}, -0.4) {$N_x = 36$};
    \node[font = \footnotesize] at (\l+0.8, {\l/2}) {$N_y = 36$};

    \node[font = \footnotesize] at ({\d/2}, -0.2) {$6$};
    \node[font = \footnotesize] at (-0.2, {\d/2}) {$6$};
  \end{tikzpicture}
	\caption{Coarse-grained lattice for constructing $H_{\text{rec}}$ on a torus. We can change the boundary condition and use the rule in Eq.~\eqref{eq:KD-decompose} to construct $H_{\rec}$ with boundaries. Each supersite $v$ is a block of size $6 \times 6$. The blue region and green regions are examples of a 4-supersite and two 2-supersites respectively.}
	\label{fig:Hrec-num}
\end{figure}

Let us then focus on the system on a torus. we check the error of \Azero and \Aone. The result is shown in Fig.~\ref{fig:EB-axioms-GD}. One can see that it is roughly exponentially decay as a function of the number of steps. 
\begin{figure}[htb]
    \centering
    \includegraphics[width = 0.9\columnwidth]{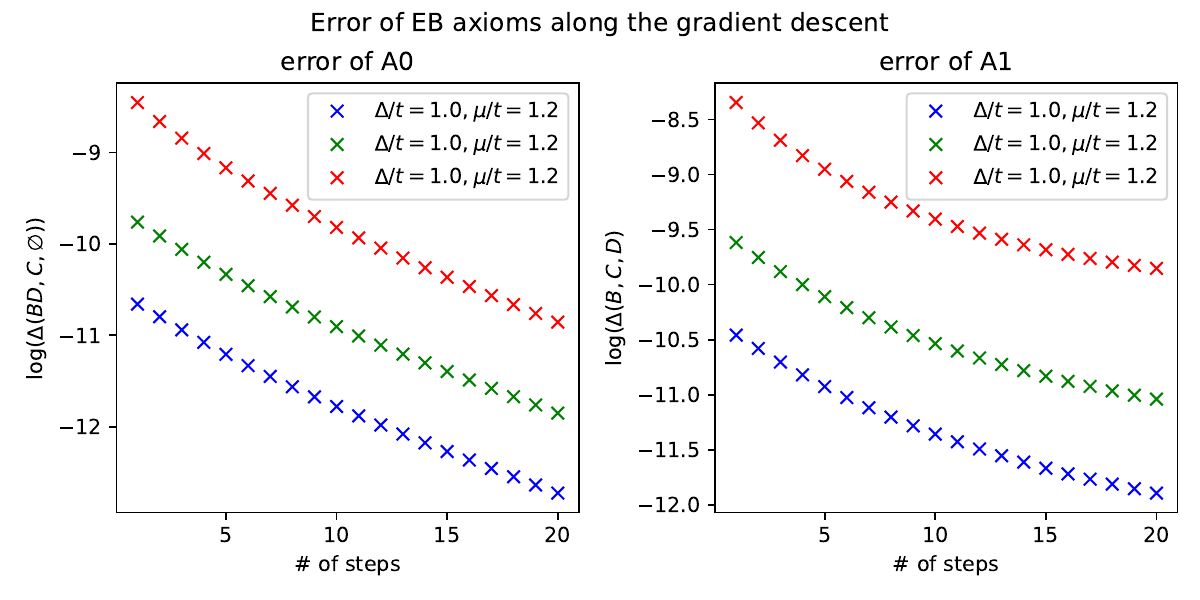}
	\caption{Error of \Azero and \Aone with gradient descent.}
	\label{fig:EB-axioms-GD}
\end{figure}

Then we use the same two setups in Fig.~\ref{fig:corner-conf-ruler-test12} to compute $\fc,\eta,\eta_J$ and error of the vector fixed point equations along the gradient descent. The result of the total central charge is shown in Fig.~\ref{fig:ctot-GD}. Its discrepancy with the expected value $\fc = 0.5$ decays roughly as power law as a function of the number of steps.  
\begin{figure}[htb]
    \centering
    \includegraphics[width = 0.9\columnwidth]{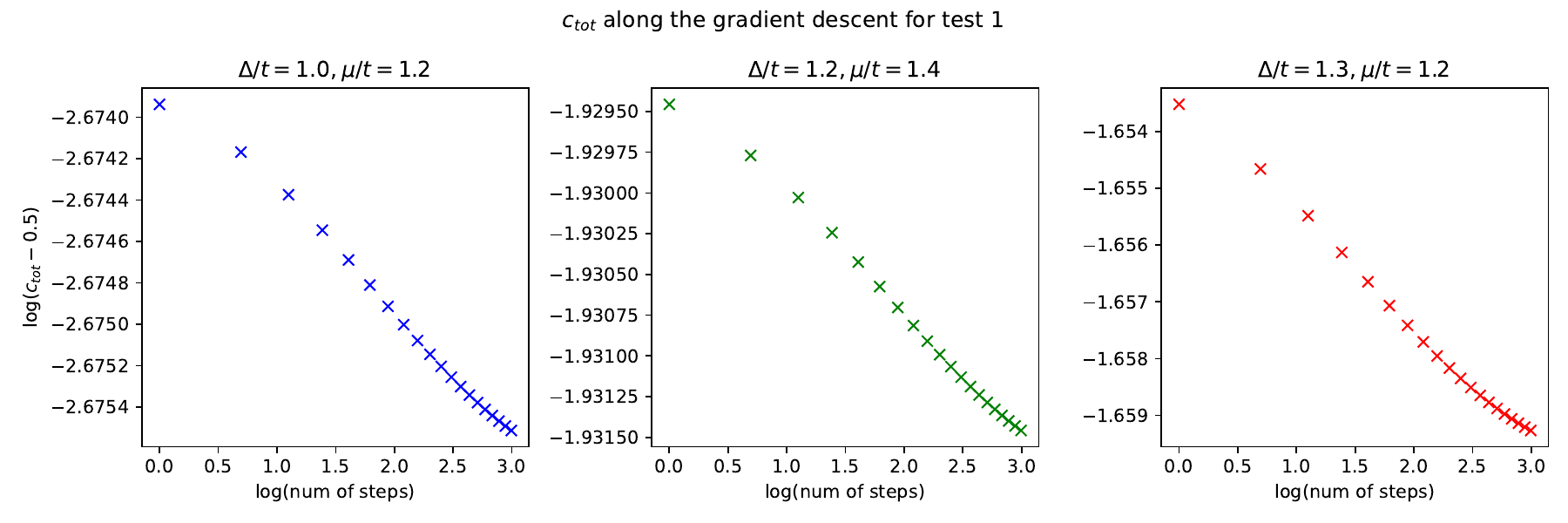}
    \includegraphics[width = 0.9\columnwidth]{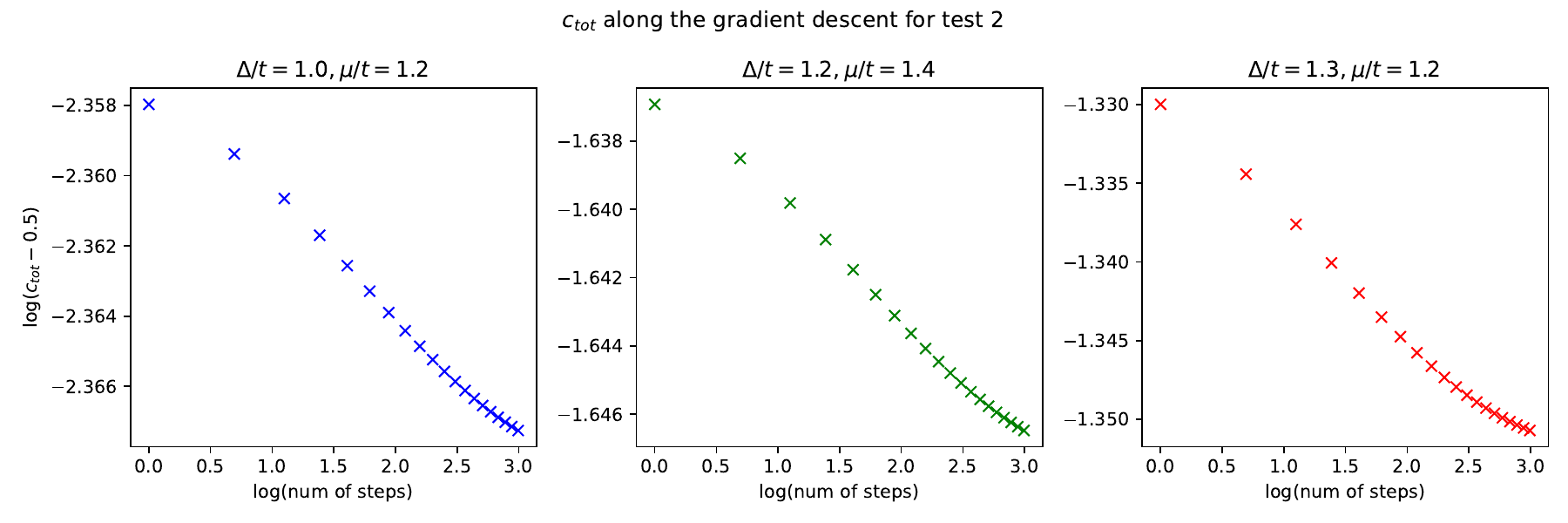}
	\caption{$\fc$ with gradient descent. The top figure is for test 1 Fig.~\ref{fig:corner-conf-ruler-test12} (left), and the bottom is for test 2 Fig.~\ref{fig:corner-conf-ruler-test12} (right).}
	\label{fig:ctot-GD}
\end{figure}
The error of the vector fixed point equations is shown in Fig.~\ref{fig:VFE-GD}, which is also roughly a power law decay.  
\begin{figure}[htb]
    \centering
    \includegraphics[width = 0.9\columnwidth]{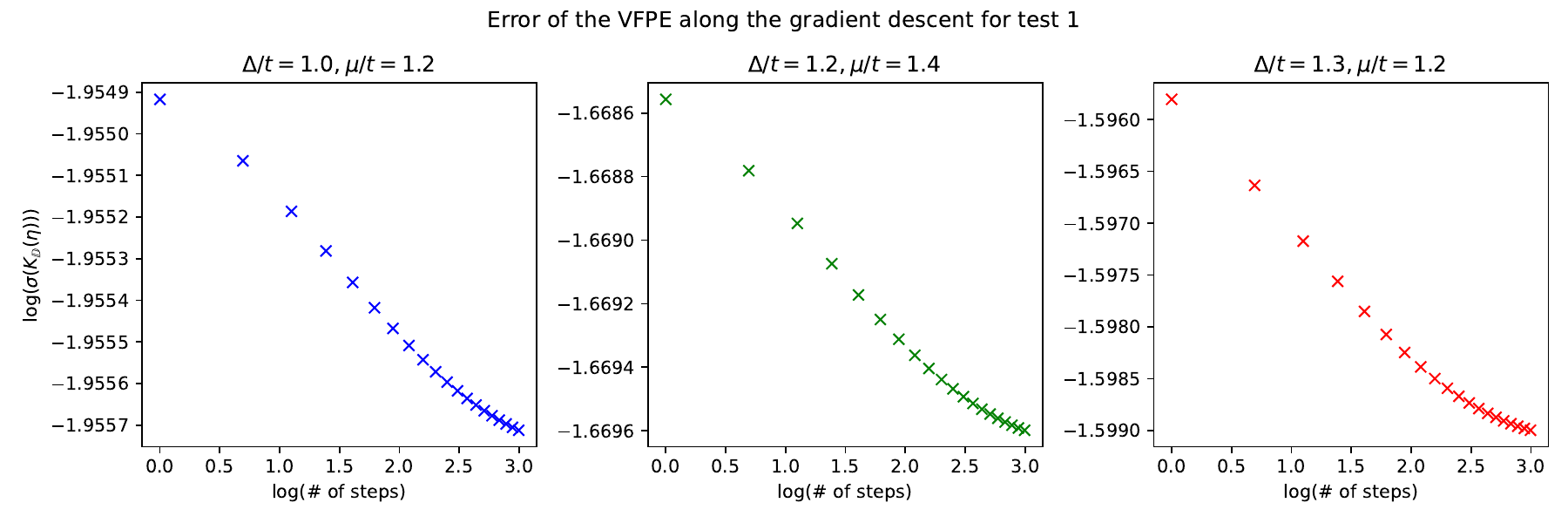}
    \includegraphics[width = 0.9\columnwidth]{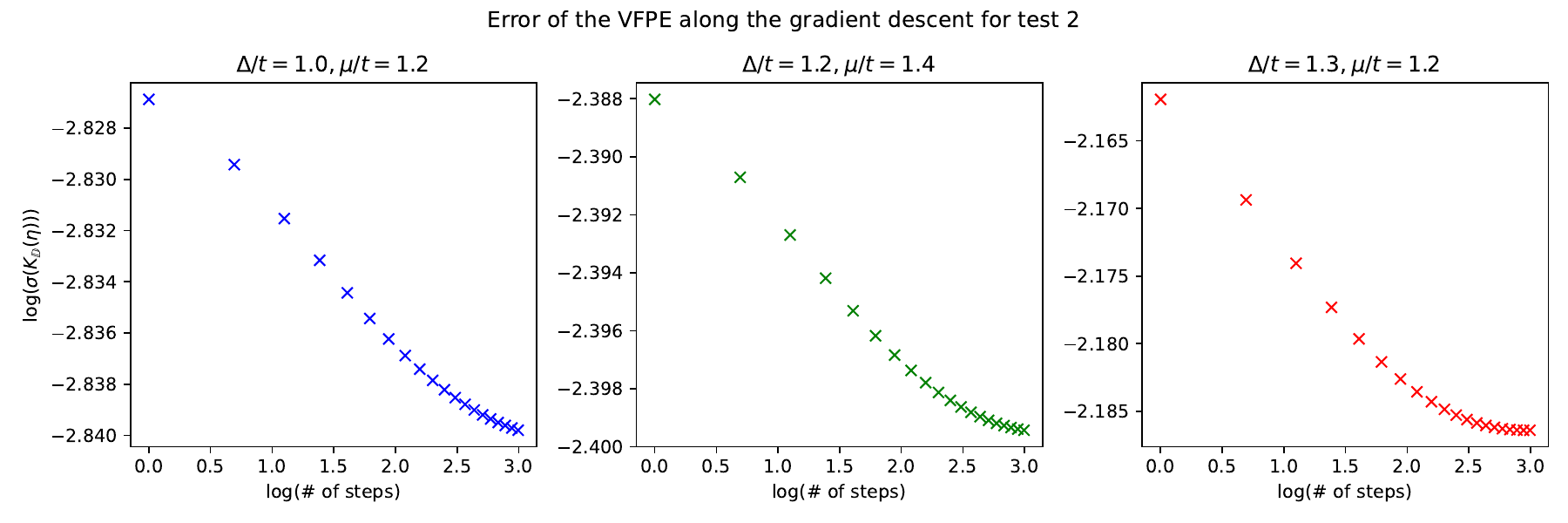}
	\caption{Error of the vector fixed point equations with gradient descent. The top figure is for test 1 Fig.~\ref{fig:corner-conf-ruler-test12} (left), and the bottom is for test 2 Fig.~\ref{fig:corner-conf-ruler-test12} (right)}
	\label{fig:VFE-GD}
\end{figure}
For the cross-ratios, the result is shown in Fig.~\ref{fig:etas-GD-test12}. 
\begin{figure}[htb]
    \centering
    \includegraphics[width = 0.8\columnwidth]{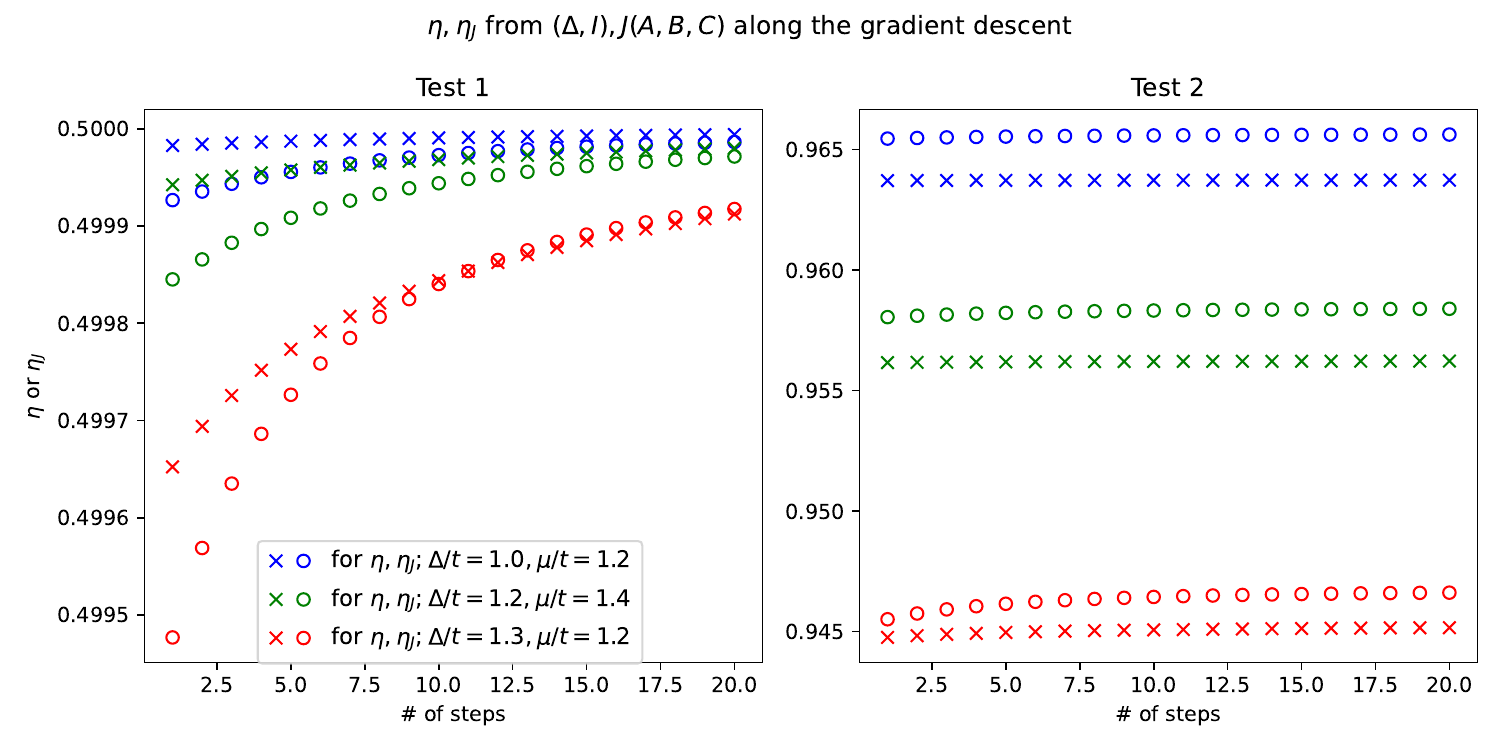}
	\caption{$\eta,\eta_J$ for test 1 (left) and test 2 (right) along the gradient descent.}
	\label{fig:etas-GD-test12}
\end{figure}
For test 1, we have a plausible expectation, namely $\eta = \eta_J = 1/2$ from the geometric cross-ratio of opening angles $\theta_A = \theta_B = \theta_C = \pi/2$. In the result, we can see a large improvement of the results along the gradient descents. We also directly plot the discrepancy from the expected value in Fig.~\ref{fig:etas-GD-test1}, which roughly shows an exponential decay as a function of the number of steps. 
\begin{figure}[htb]
    \centering
    \includegraphics[width = 0.5\columnwidth]{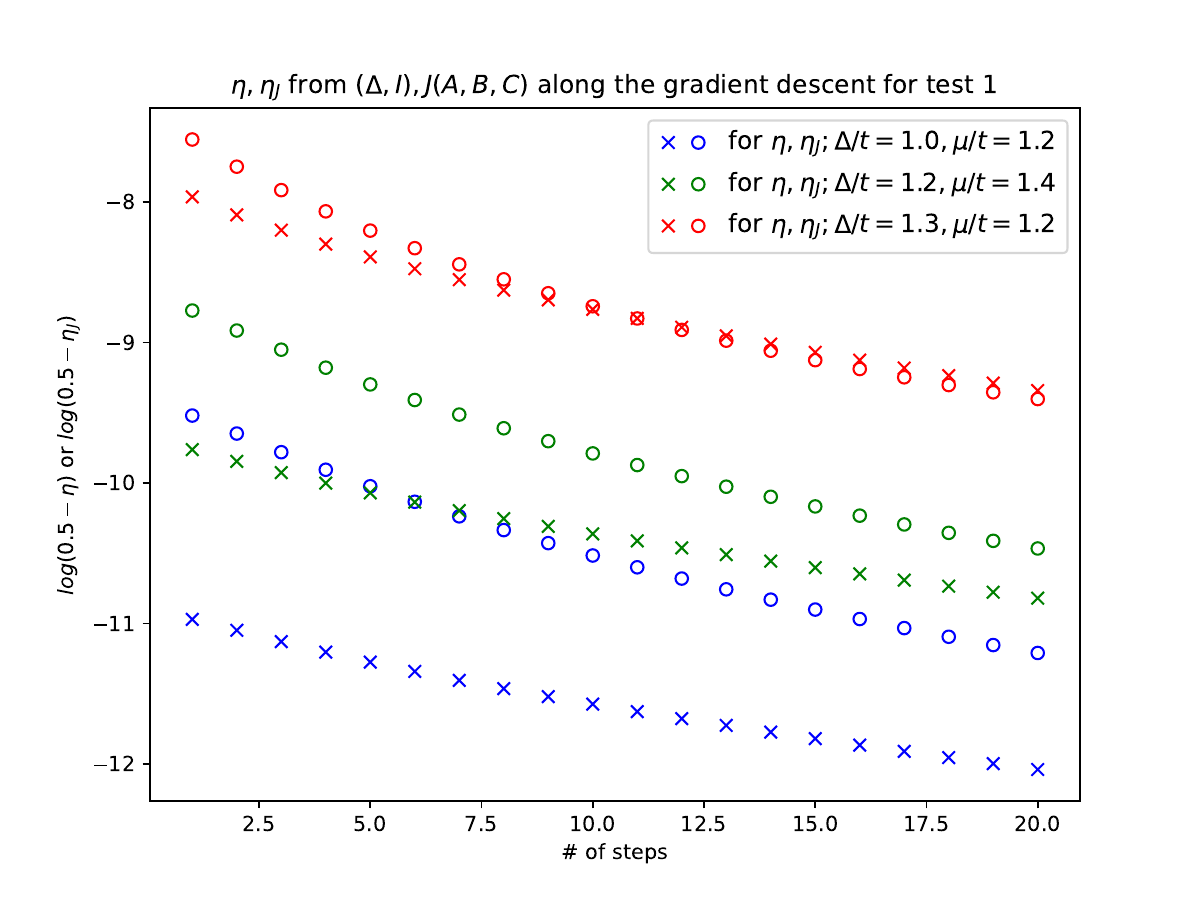}
	\caption{$\eta,\eta_J$ in test 1 with gradient descent, compared with the expected value $1/2$.}
	\label{fig:etas-GD-test1}
\end{figure}
For test 2, the improvement is relatively small compared with test 1. We suspect this might be due to the irregular shape of the corners. 

From these results, we can see that the finite size errors are reduced along the gradient descents. We do not expect that error for \Azero and \Aone can be reduced to zero and one could obtain exactly the same quantitative results as predicted from hypothesis~\ref{hypo:corner} in a Hilbert space with finite local dimensions. Nevertheless, these results suggest that the errors of \Azero and \Aone is positively correlated with the discrepancy between the corner entanglement results in finite system sizes and the expected results. This is aligned with the intuition that reducing \Azero and \Aone could drive the state closer to the IR fixed point. We shall leave the detailed correlations among these errors to future investigation. 

\section{Discussion}
In this paper, we have given a detailed analysis of the properties of the entanglement Hamiltonian of subregions in a 2+1d chiral gapped groundstate. We focus on its operator content and its role in the emergence of universal properties. We first proposed an operator bulk/edge correspondence that generalizes the Li-Haldane conjecture. Then, we've focused on the entanglement in the corner regions and proposed a physical picture [Hypothesis~\ref{hypo:corner}] as well as a logical framework [Section~\ref{sec:logical-framework}]. A concrete outcome is a solid and independent understanding of universal corner contributions to the entanglement entropy for these states \cite{Rodriguez:2010dm,Rozon:2019evk,Sirois:2020zvc,Liu:2023pdz,Liu:2024ulq} and how to extract them in a cutoff-independent way. Moreover, in the same framework, we've explained the reason why the modular commutator can compute the chiral central charge. The logical framework describes an emergence of a ``universal'' measure of a sharp corner up to a global conformal transformation. We also defined a quantity $(\fc)_{\text{min}}$ which can be used to diagnose edge ungappability. 

Perhaps most interestingly, our work implies that the bulk wavefunction knows something about geometry -- it can detect angles with the corner conformal ruler under the stationarity condition or the vector fixed point equation. This is consistent with the mathematicians' viewpoint that gapped chiral states are {\it not} described by topological field theory (e.g.~\cite{Freed:2016rqq}), but seems to go beyond the framing anomaly of \cite{Witten:1988hf} since it happens even when $c_- =0$. 
It seems that this is further evidence that a nonzero higher central charge can be regarded as a gravitational anomaly, perhaps related to modular transformations of a $T^2$ foliation of a lens space path integral \cite{Kaidi:2021gbs}.

Let us zoom out to a big picture. This work is a step toward establishing an explicit operator algebra correspondence between chiral topological order systems and 1+1D CFT. The following thought may cross the reader's mind: people already know that there is a correspondence between chiral TQFTs that describe the chiral topologically-ordered phases and rational chiral CFTs, why do we still pursue this line of study? Our goal here is to \emph{explain} why such a correspondence exists. As we explained in the introduction, it is indeed a common belief that one can use a UMTC to describe the IR fixed point of a gapped phase of matter\footnote{Here we are not considering symmetry enriched systems or fracton phases.}. However, \Azero and \Aone have the ability to \emph{explain} why this is true. Such an endeavor of analyzing how exactly the universal properties emerge can provide us with more insights and also help us examine whether there could be other cases not covered by the currently-known TQFT/CFT correspondence. For example, could there be some ungappable edges that are described by non-rational CFT or even a non-CFT? Furthermore, our method can be generalized to higher dimension. Could there be some analog of Hypothesis~\ref{hypo:corner} in higher dimension? Knowing under what conditions the bulk/edge correspondence will emerge can help us answer these questions.

Back to the concrete discussion. We envision that the operator algebra for the chiral topologically-ordered systems is obtained from local entanglement Hamiltonians. For example, one thing we can do is to extract a Virasoro algebra as in \cite{Vir}. On a corner region, by applying the Hypothesis~\ref{hypo:corner} to treat the corner as a hole with a gapless edge, one can design the same linear combinations of entanglement Hamiltonians as in \cite{Vir}. One can also apply the disentangler to actually make a small edge as discussed in Section~\ref{sec:diagnostic-for-ungappable}, then repeat the construction as in \cite{Vir}. 

As we discussed in the introduction, there are different scales $l$ on a chiral gapped wavefunction, compared with the lattice spacing $a$ and the correlation length $\xi$. In this paper, we are considering an ``extreme'' case where $a \ll l < \xi$, which is exactly the regime relevant to a corner region. This is the minimal length scale on which one can hope to identify any universal physics. The scale $l \gg \xi$, where the topological quantum field theory emerges, is the scale captured by the previous development of entanglement bootstrap by assuming no violation of \Azero and \Aone. An interesting question is how does the physics in these two scales ``interplolate'' in the region $l > \xi$. One concrete question along this direction is to find a more accurate description of the violation of \Azero and \Aone. One essence of this paper is that the Li-Haldane conjecture $\rho_D = e^{-\xi/l_D L_0}$ seems to be true across all these scales: On the scale $l$ with $a \ll l < \xi$ (i.e. the corner regions), Li-Haldane conjecture is reflected in the possibility that we can use an edge CFT groundstate to model the corner entanglement. On the scale $l \gg \xi$, one can think of $\rho_D$ as a CFT thermal state at infinite temperature and hence we have the exact area law and \Azero and \Aone with vanishing violations. On the scale $l > \xi$, one might be able to use the thermal CFT state with finite temperature to conclude some universal (as robust as the Li-Haldane conjecture) features in this regime. 

\vskip.2in
{\bf Acknowledgements.}
We are grateful to Meng Cheng, Isaac Kim, Max Metlitski, Shinsei Ryu, and Bowen Shi for helpful discussions and encouragement.
This work was supported in part by
funds provided by the U.S.~Department of Energy
(D.O.E.) under cooperative research agreement 
DE-SC0009919, and by the Simons Collaboration on Ultra-Quantum Matter, which is a grant from the Simons Foundation (652264, JM). JM received travel reimbursement from the Simons Foundation; the terms of this arrangement have been reviewed and approved by the University of California, San Diego in accordance with its conflict of interest policies.

\appendix 
\renewcommand{\theequation}
{\Alph{section}.\arabic{equation}}

\section{A sufficient condition for gapped boundary}
\label{app:siva}

In this appendix, we briefly summarize how to combine the arguments of 
\cite{Zou:2020bly} and \cite{Siva:2021cgo} 
to show that a state satisfying a strict area law admits a gapped boundary.

Suppose that for an arbitrary region $A$ with the topology of the disk, whose boundary has length $|\partial A|$, 
\be S_A = \alpha |\partial A| - \gamma_\text{topo}.\ee
This implies that $I(i-1:i+1|i) = 0$, $i=1..4$.  
Proposition 17 of \cite{Zou:2020bly} for the choice of regions in Fig.~\ref{fig:circles} 
then says that any such wavefunction can be written as a sum of polygon states 
with respect to regions 1,2,3,4.  
A sum of polygon states is of the form
\be \ket{\psi}_{ABCD}
= \sum_j \sqrt{p_j} \ket{\psi_j}_{A^j_RB_L^j}
\ket{\psi_j}_{B^j_RC_L^j}
\ket{\psi_j}_{C^j_RD_L^j}
\ket{\psi_j}_{D^j_RA_L^j}, \ee
with respect to a decomposition of each local Hilbert space
$ \CH_\alpha = \oplus_j \CH_{\alpha_L^j} \otimes \CH_{\alpha_R^j}$.  

For our purposes, the important point about the conclusion that the state is a sum of polygon states is that 
arguments of \cite{Siva:2021cgo} show that such a state would admit a gapped boundary.  The idea is that 
such a state has only bipartite and GHZ-type entanglement, and in particular 
the entanglement measure called ``Markov gap" vanishes for such states.  
The main result of \cite{Siva:2021cgo} is that the Markov gap is a bulk measure of the total central charge of the edge CFT.

\section{Reconstructed Hamiltonians}
\label{app:Hrec}
In this section, we discuss two reconstructed Hamiltonians $H_{\rec}^{\euler}$ and $H_{\rec}^{\Delta}$ constructed from entanglement Hamiltonians of a given wavefunction. In the big picture, the goal is to relate entanglement Hamiltonian to a real Hamiltonian, which is involved in the formulation of the operator bulk/edge correspondence. 

$H_{\rec}^{\euler}$ has the same form as a decomposition of an entanglement Hamiltonian of a disk. It is shown in \cite{Kim:2021tse} that for a state satisfying \Aone, one can decompose the entanglement Hamiltonian of a disk $D$ as 
\begin{equation}
\label{eq:KD-sum}
    K_D = \sum_{f \in D} K_f - \sum_{e \notin D_{\partial}} K_e + \sum_{v \notin D_{\partial}} K_v, 
\end{equation}
where $D_{\partial}$ is a single layer of lattice sites around the entanglement boundary of $D$. 

One might wonder about a similar formula for regions of other topology.  
For a generic region $D$, which may contain non-contractable loops, such an equation will not hold for the reduced density matrix from the groundstate, but for that of the maximal entropy states of the information convex set $\Sigma(D)$ defined in \cite{Shi_2020}. From the physical point of view, the reason is that the right hand side of Eq.~\eqref{eq:KD-sum} is invariant under any threading of anyons into the holes, hence the left hand side has to be the maximal entropy state which includes all the anyon sectors. Technically, the reason is that Eq.~\eqref{eq:KD-sum}, besides as a decomposition, can also be regarded as recovering the states $\rho_D$ from small pieces $\{\rho_f\}, \{\rho_e\}, \{\rho_v\}$, where the elementary step is Petz recovery map. Such a merging process will result in the maximal entropy states of $\Sigma(D)_{\rho}$ \cite{Shi_2020}.

Consider a generic region $D$ with a triangulation. $H_{\rec}^{\Delta}$ is defined as 
\begin{align}
    H_{\rec}^{\Delta} \equiv  \sum_{v \in D_{\inte}} \frac{1}{2} \overline{\hat{\Delta}_v} + \sum_{v \in D_{\partial}} \frac{1}{2}\overline{\hat{\Delta}^{\partial}_v}. 
\end{align}
where $\overline{\hat{\Delta}_v}$ is the averaging over all the possible operator bulk \Aone combination centered at $v$ and of a fixed size and $\overline{\hat{\Delta}_v^{\partial}}$ are those along the boundaries. The subscript denotes the center of the \Aone region and we consider all the possible partitions and orientations of the buffer region with one lattice spacing. This operator is manifestly positive as $\hat{\Delta} \geq 0$ due to \cite{Lin:2022jtx}. If $D$ is closed, then demanding $\ket{\psi}$ being its groundstate with zero energy will enforce the \Aone condition. If $D$ has a CFT boundary, on which there are CFT groundstate then the boundary term will give a reconstructed Hamiltonian of the edge CFT as constructed in \cite{Vir}. If it is a CFT thermal state with inverse temperature $\beta$ living on the boundary, under the limit $\beta/\ell_D$, one shall also obtain a CFT reconstructed Hamiltonian from the Hypothesis~\ref{hypo:op-bulk-edge}. This is the case where $\partial D$ is an entanglement boundary in the bulk of a chiral gapped state and one can identify $\beta/\ell_D$ as $\xi / \ell_D$.

In the result of this section, we are going to derive a relation between $H_{\rec}^{\Delta}$ and $H_{\rec}^{\euler}$. The setup is on any triangulation of a 2d manifold $\CM$. The degrees of freedom live on the vertices. Let $F, E, V$ be the set of faces (triangles), edges and vertices and we use $d_v$ to denote the degree of a vertex $v$. We shall consider $H_{\rec}^{\Delta}$ and $H_{\rec}^{\euler}$ from a state $\rho$ that satisfy \Aone in the bulk. We do not posit any specific boundary condition. In this setup, $H_{\rec}^{\euler}$ and $H_{\rec}^{\Delta}$ are related by 
\begin{align}\label{eq:Hrec-relation}
    H_{\rec}^{\Delta} - \sum_{v \in V_{\inte}} \frac{6-d_v}{6} \overline{\hat{\gamma}_v} = H_{\rec}^{\euler}, 
\end{align}
where $\overline{\hat{\gamma}_v}$ is an averaging of operator Kitaev-Preskill TEE. Explicitly, for a face $f = \langle abc \rangle$, we define 
\begin{align}
    \hat{\gamma}_f \equiv K_{ab} + K_{bc} + K_{ac} -  K_{abc} - K_a - K_b - K_c. 
\end{align}
The averaging is 
\begin{align}
    \overline{\hat{\gamma}_v} \equiv \frac{1}{d_v} \sum_{f, v\subset f} \hat{\gamma}_f. 
\end{align}

Let us first consider the bulk region. For a site $v$, we consider all the possible $\hat{\Delta}$ that are centered at $v$ and whose thickness of the buffer region is one lattice spacing. Therefore, there are $d_v = N$ number of sites in the buffer region. Let us denote them as $u_1 \cdots u_N$ in a certain order (say clockwise). We define 
\begin{align}
    \overline{\hat{\Delta}_v} \equiv \frac{1}{N_{p,o}}\sum_{p,o} \hat{\Delta}_{v,(p,o)}.
\end{align} 
Here $p$ denotes a partition of $p(N)$ into two positive integers $p(N) = m+n$ so that in $\hat{\Delta}(B,v,D)$ region $B,D$ contains $m,n$ sites, and $o$ denotes the orientations. $N_{p,o}$ denote the number of all partitions and orientations. Notice, when $N$ is even and the partition is $q(N) = N/2 + N/2$, i.e. $|B| = |D| = N/2$, $\hat{\Delta}_{v,q,o} = \hat{\Delta}_{v,q,o'}$ where $o,o'$ are related by a $\pi$ rotation. In the sum, we still count both of them. This has the advantage that for a given partition, the number of orientation is always $d_v$. Notice that we actually didn't overcount to give an extra ``weight'' for those $\hat{\Delta}$, because the factor $2$ from the double counting is canceled by regarding $N_o = d_v$. 

Now we are in a position to show Eq.~\eqref{eq:Hrec-relation}. The key in the derivation is the Markov decomposition $K_{ABC} = K_{AB}+K_{BC}-K_B$ for a state $\rho_{ABC}$ being a quantum Markov chain $I(A:C|B)_{\rho} = 0$. With this decomposition, we can decompose $\hat{\Delta}$ into a sum over $K_f, K_e,K_v$. For example, consider a $\hat{\Delta}_v$ centered at a vertex $v$ of degree $d_v$. We can decompose as follows: 
\begin{align}
    \includegraphics[width = 0.2\columnwidth, valign = c]{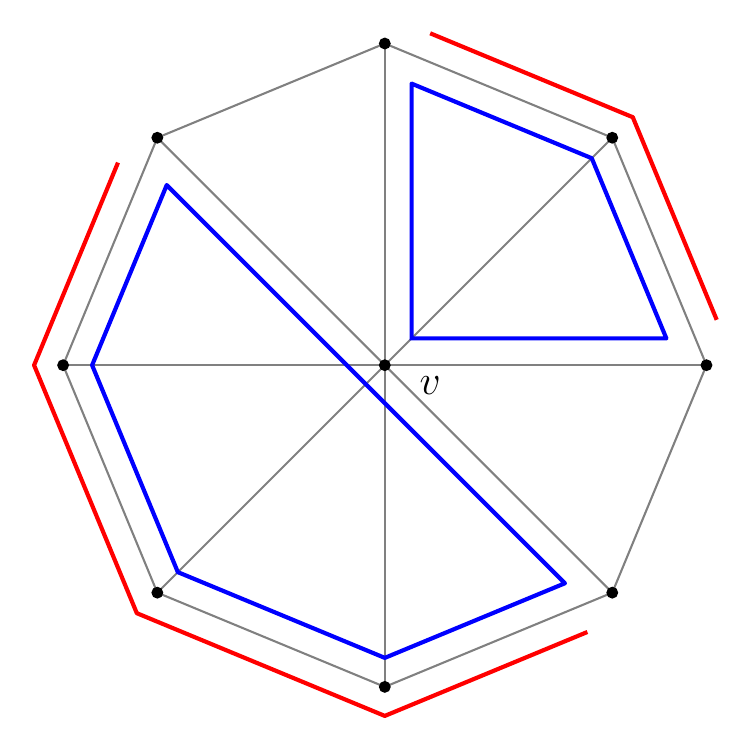} = \includegraphics[width = 0.2\columnwidth, valign = c]{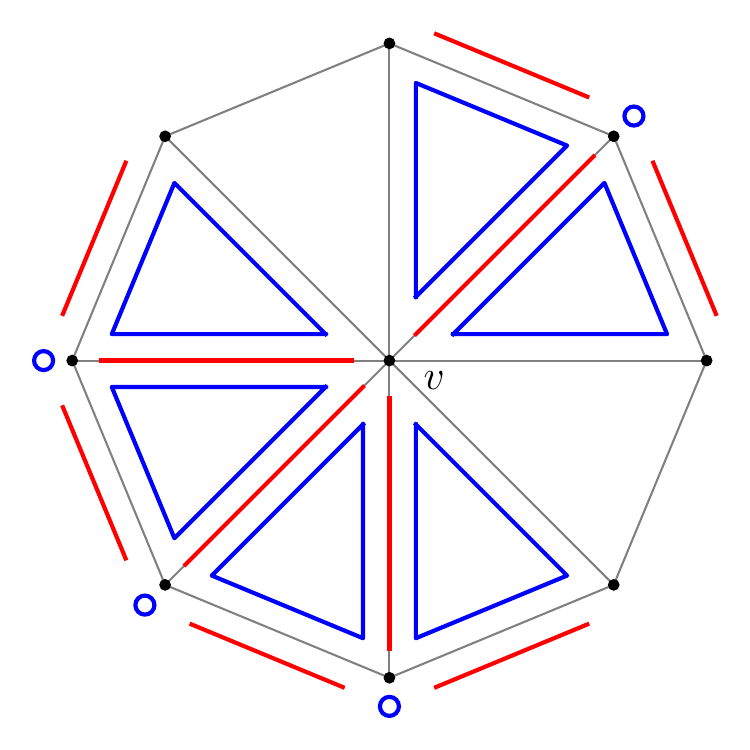}
\end{align}
Here we use a graphic notation. Each figure represents a sum over entanglement Hamiltonians. Degrees of freedom lives on the vertices represented by black dots. Each path, i.e. a collection of connected line segments, represent a $K_{X}$ in the sum with $X$ being the union of the sites along which the path goes. The blue color indicates the sign is positive and the red stands for negative. 

After the averaging over all possible partitions and orientations, we obtain 
\begin{align}
    \overline{\hat{\Delta}_v} =  \frac{d_v-2}{d_v}\includegraphics[width = 0.2\columnwidth, valign = c]{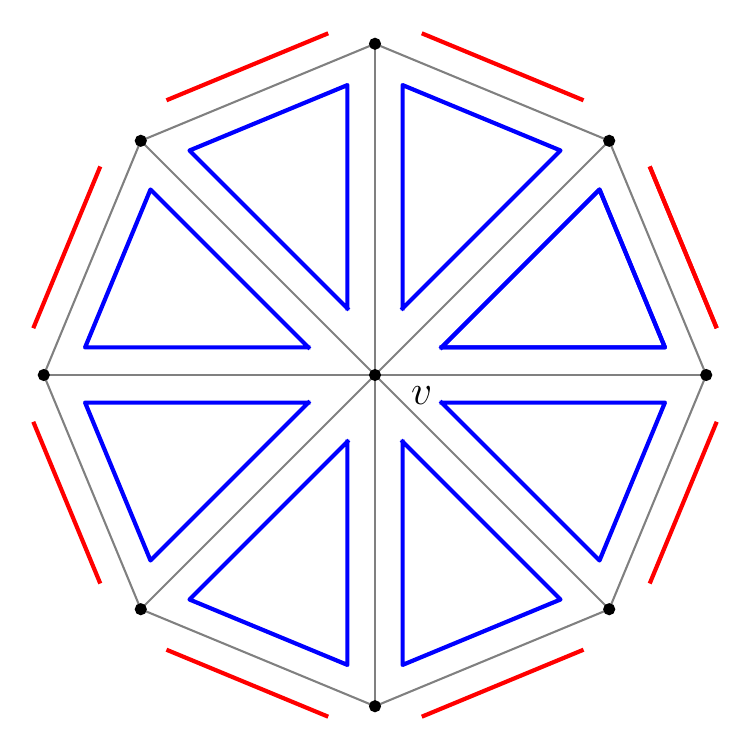} + \frac{d_v-4}{d_v} \includegraphics[width = 0.2\columnwidth, valign = c]{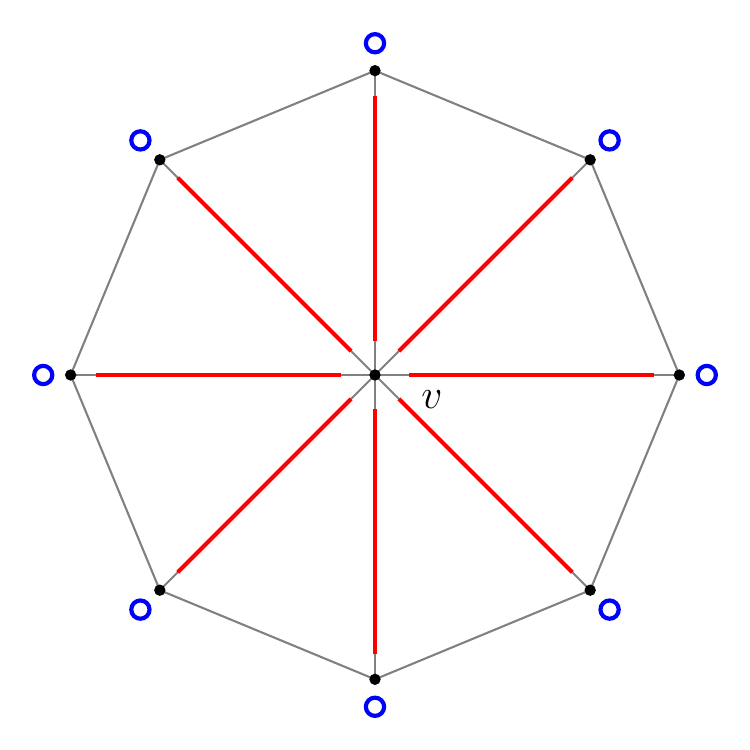}
\end{align}

For $-\overline{\hat{\gamma}_v}$, it is graphically represented as 
\begin{align}
    -\overline{\hat{\gamma}_v} =  \frac{1}{d_v}\includegraphics[width = 0.2\columnwidth, valign = c]{figs/Delta-hat-ave1.pdf} + \frac{2}{d_v} \includegraphics[width = 0.2\columnwidth, valign = c]{figs/Delta-hat-ave2.pdf} + K_v 
\end{align}

Therefore, one can obtain 
\begin{align}
    h_v&\equiv \frac{1}{2}\overline{\hat{\Delta}_v} - \frac{6-d_v}{6}\overline{\hat{\gamma}_v} \\& = \frac{1}{3} \includegraphics[width = 0.2\columnwidth, valign = c]{figs/Delta-hat-ave1.pdf} + \frac{1}{6}\includegraphics[width = 0.2\columnwidth, valign = c]{figs/Delta-hat-ave2.pdf} + \frac{6-d_v}{6} K_v
\end{align}

Now let us consider an arbitrary $f = \langle v i j \rangle, e = \langle vi\rangle, v$ shown in Fig.~\ref{fig:simplex} and examinee $K_f, K_e, K_v$ obtained from the left hand side of Eq.~\eqref{eq:Hrec-relation}. For $K_f$, there is $1/3 \cdot K_f$ from $h_v, h_i, h_j$ and hence we obtain $K_f$. For $K_{e}$, there will be $1/6 \cdot K_e$ from each of $h_{v}, h_{i}$ and $1/3 \cdot K_e$ from each of $h_j,h_l$. Therefore, there will be $K_e$ in total. For $K_v$, there will be $1/6\cdot K_v$ from each $h_{u}$ such that $u$ is connected to $v$ by a single edge and hence there are $d_v$ in total from them. There is $(6-d_v)/6 \cdot K_v$ from $h_v$. Therefore, we obtain $K_v$ in the end. Hence we finish the derivation of Eq.~\eqref{eq:Hrec-relation}. 

\begin{figure}[htb]
    \centering
    \includegraphics[width = 0.5\columnwidth]{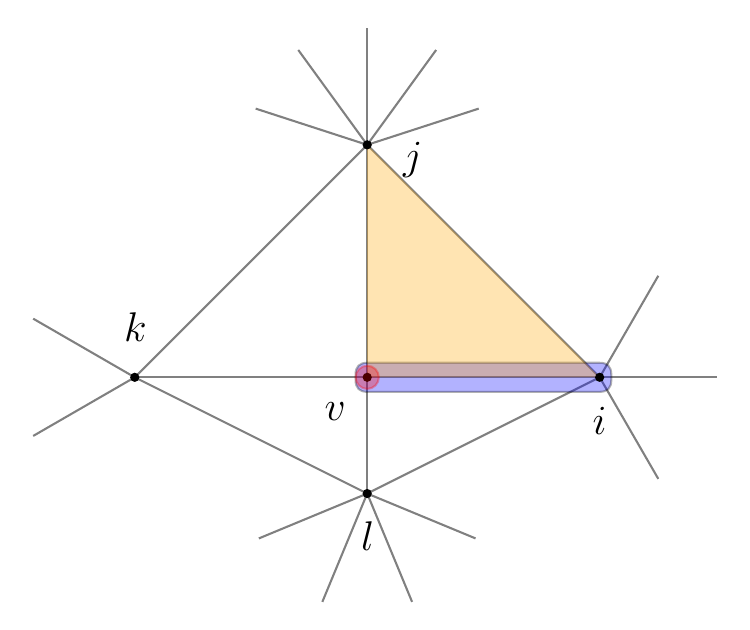}
	\caption{An arbitrary face $f = \langle vij\rangle$, edge $e = \langle vi\rangle$ and site $v$ in the triangulation.}
	\label{fig:simplex}
\end{figure}

Now we discuss the terms near the boundary. Because here we work with a coarse-grained lattice, we can choose to make triangle lattice near the boundary as shown in Fig.~\ref{fig:bdy-lattice}. 

\begin{figure}[htb]
    \centering
    \includegraphics[width = 0.5\columnwidth]{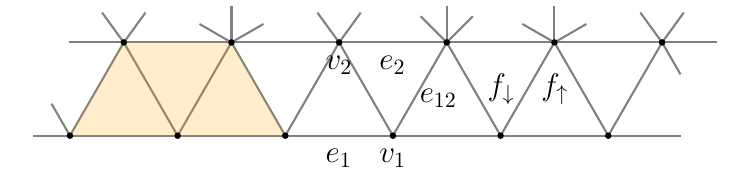}
	\caption{Triangle lattice along the boundary. $V_1=\{v_1\}$ are set of vertices that along the boundary. $V_2 = \{v_2\}$ are set of vertices that are away from the boundary vertices by one lattice spacing. $E_1 = \{e_1\},E_2 = \{e_2\}$ are set of edges that are made out of sites in $V_1, V_2$ respectively. $E_{12} = \{e_{12}\}$ are set of edges that contains both $v_1\in V_1, v_2\in V_2$. $F_{\uparrow}, F_{\downarrow}$ are faces that contains two and one boundary sites respectively.}
	\label{fig:bdy-lattice}
\end{figure}

The $\overline{\hat{\Delta}^{\partial}_v}$ is  
\begin{align}
    \overline{\hat{\Delta}^{\partial}_v} &\equiv \frac{1}{3} \(
    \includegraphics[width = 0.2\columnwidth, valign = c]{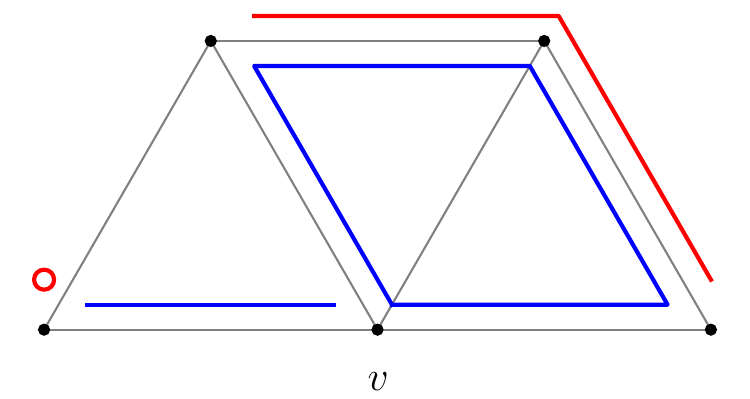} + \includegraphics[width = 0.2\columnwidth, valign = c]{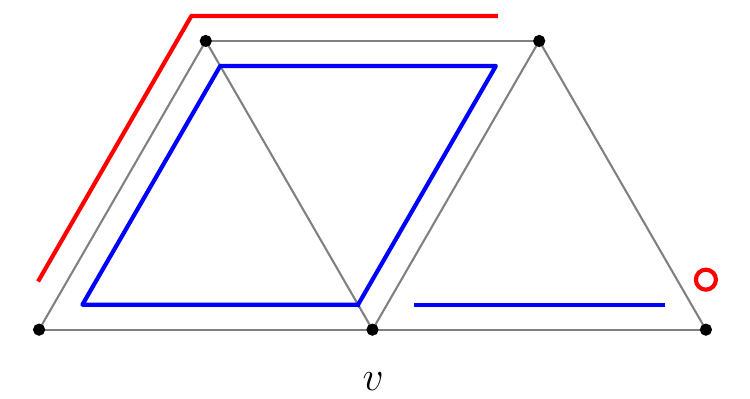} +\includegraphics[width = 0.2\columnwidth, valign = c]{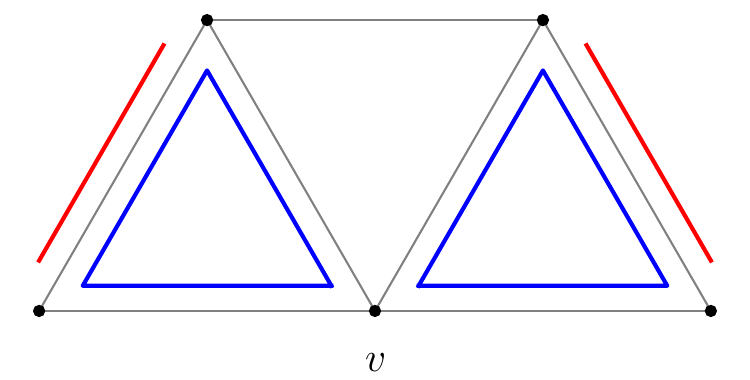}  \)  \\ 
    & = \frac{2}{3} \includegraphics[width = 0.2\columnwidth, valign = c]{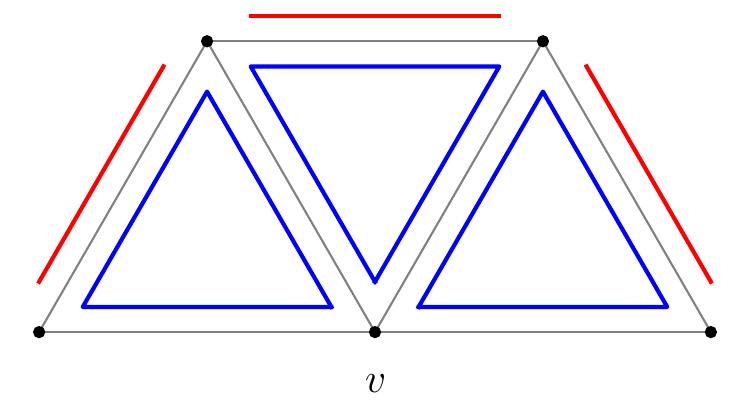} + \frac{1}{3} \includegraphics[width = 0.2\columnwidth, valign = c]{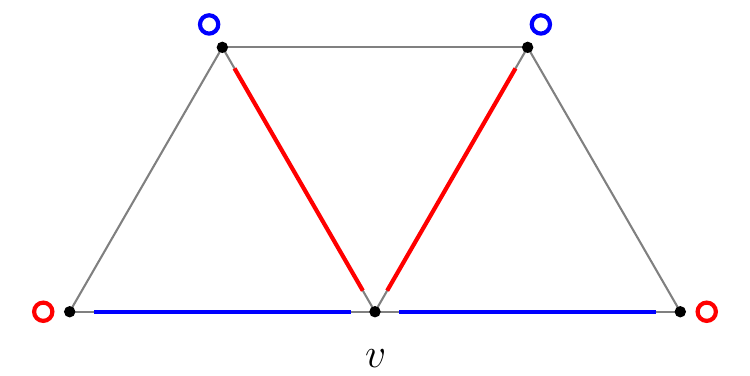},
\end{align}
where the second equality is obtained via Markov decomposition. Such decomposition only utilizes the \emph{bulk} \Aone condition and hence it is applicable regardless of the boundary condition. 

One can now do a counting similar to what we did in the bulk regions.
Notice $\overline{\hat{\Delta}_v^{\partial}}$ produce exactly the same number of $K_f$, $K_e, e \notin E_1$ and $K_{v}, v \in V_2$ as the averaged operator bulk \Aone combination, hence one shall obtain the same terms as in the bulk. For $K_{e}, K_{v}$ with $e\in E_1, v \in V_1$, the sign is flipped and one can see that they are canceled with those from $\overline{\hat{\Delta}_v}$ in the bulk. Hence we finish the derivation of Eq.~\eqref{eq:Hrec-relation}.

\section{Properties of $K_{D_\inte}$ and $K_{D_{\partial}}$}
\label{app:property}

This appendix contains the explicit derivations of various properties regarding $K_{D_\inte}$ and $K_{D_\partial}$ mentioned in the main text. 

First, we will show that $K_{D_\partial}$ fully capture the action of $K_D$ on the state. Explicitly, we will show that 
\begin{align}\label{eq:KDn-psi}
     K_D^n \ket{\Psi} = (K_{D_\partial} - \gamma \mathbbm{1})^n \ket{\Psi}. 
\end{align}
As a starter, notice that 
\begin{align}\label{eq:KDn-flip}
     K_D^n \ket{\Psi} = K_{\bar D}^n \ket{\Psi},\quad \forall n \in\mathbb{N} 
\end{align}
This can be shown by doing Schimidt decomposition of $\ket{\Psi}$ between $\CH_D$ and $\CH_{\bar D}$. One can apply this to the left-hand-side of Eq.~\eqref{eq:KDn-psi} and obtain 
\begin{align}
     K_D^n \ket{\Psi} = K_D K_{\bar D}^{n-1}\ket{\Psi} =  K_{\bar D}^{n-1} K_{D}\ket{\Psi}, 
\end{align}
where the last equality is because $[K_D, K_{\bar D}^{n-1}] = 0$ for they having different supports. Then one can apply the decomposition $K_D = K_{D_\inte} + K_{D_\partial}$ and obtain 
\begin{align} \label{eq:KDn-1}
     K_{\bar D}^{n-1} K_{D}\ket{\Psi} = K_{\bar D}^{n-1} (K_{D_\partial} - \gamma \mathbbm{1}) \ket{\Psi} = (K_{D_\partial} - \gamma \mathbbm{1}) K_{\bar D}^{n-1} \ket{\Psi} = (K_{D_\partial} - \gamma \mathbbm{1}) K_{D}^{n-1} \ket{\Psi}
\end{align}
Here, the usage of 
\begin{align}
     K_D \ket{\Psi} = (K_{D_\partial} - \gamma \mathbbm{1}) \ket{\Psi}
\end{align}
is based on the fact that $\mathbb{L} = K_D - K_{D_\partial}$ is a good modular flow generator defined in \cite{Vir}. With \Azero and \Aone, one can show that $\mathbb{L}\ket{\Psi} \propto \ket{\Psi}$ and the proportionality can be comptued by taking the overlap, which is $-\gamma$. In the last two equalities of Eq.~\eqref{eq:KDn-1}, we first use the fact that the two operator commutes to bring $K_{D_\partial} - \gamma \mathbbm{1}$ in front of $K_{\bar D}^{n-1}$, then the latter back to $K_D^{n-1}$ due to Eq.~\eqref{eq:KDn-flip}. Then for $K_D^{n-1}\ket{\Psi}$, one can repetitively apply the derivation above, until all the $K_D$ is converted to $K_{D_\partial} - \gamma \mathbbm{1}$. We thus finish the proof. 

\section{Detailed calculations for the modular commutator in chiral CFT}
\label{app:TT}
In this section, we give a detailed calculation of the modular commutator Eq.~\eqref{eq:J-chiral-CFT} for two generic intervals $[a,b]$ and $[c,d]$. The result is slightly more general that the one in \cite{Zou:2022nuj} as it also takes into account cases where the boundaries of the intervals coincide. In such cases, there will be singularities in the OPE computation. We present two ways in obtaining Eq.~\eqref{eq:J-chiral-CFT}, one is from Virasoro algebra and the other is OPE computation with a regulation of the `boundary singularities'. The Virasoro algebra gives a confirmation of the validity of the regulation. 

\subsection{Conventions}
We first introduce the notations and conventions: 
\begin{itemize}
    \item $x \in S^1$ means $x \sim x + 2\pi n$. The ``principal domain'' is $(-\pi, \pi]$. 
    \item The principal branch of $\ln(z)$ is $\im\ln(z)\in(-\pi,\pi]$. 
    \item Indicator function: 
    \begin{equation}
        \Theta_{[a,b]}(x) = \left\{ \begin{aligned}
            &1,&& \text{ if } x \in (a,b) \\ 
            &0,&& \text{ if } x \in S^1\setminus [a,b]  \\ 
        \end{aligned} \right. 
    \end{equation} 
    \item Step function 
    \begin{equation}
        \Theta(x) = \left\{ \begin{aligned}
            &1,&& \text{ if } x>0 \\ 
            &0,&& \text{ if } x<0  \\ 
        \end{aligned} \right. 
    \end{equation}  
    \item $l_{a,b} = b-a$.  
    \item $q = q_n = e^{\ii n}$.
    \item $\sum_{n} = \sum_{n\in\mathbb{Z}}$. 
\end{itemize}

\subsection{Commutators from Virasoro algebra} 

This section is to compute the modular commutators using Virasoro algebra. 
We first want to express the entanglement Hamiltonians using Virasoro generators. 
The entanglement Hamiltonian of an interval on $[a,b]$ is 
\begin{equation}
    K_{[a,b]} = \int_0^{2\pi} dx \beta_{[a,b]}(x)T(x),\quad \beta_{[a,b]}(x) = 2\Theta_{[a,b]}(x) \frac{\sin((x-a)/2)\sin((b-x)/2)}{\sin((b-a)/2)}.
\end{equation}
Using Fourier transformation 
\begin{equation}
    T(x) = \sum_{n} e^{-\ii n x}L_n,
\end{equation}
we can obtain 
\begin{equation}
    K_{[a,b]} = \sum_n \lambda^{[a,b]}_n L_n,
\end{equation}
with 
\begin{align}
   \lambda^{[a,b]}_n &= \int_0^{2\pi}dx \beta_{[a,b]}(x)e^{-\ii nx}  \\ 
    & = \frac{\ii\cot\(\frac{l_{a,b}}{2}\)\cdot (e^{-\ii nb}-e^{-\ii na}) - n (e^{-\ii nb}+e^{-\ii na})}{n^3-n} \\ 
    & = \frac{2e^{-\ii n(b+a)/2}}{n^3-n}\left[ \cot\( \frac{l_{a,b}}{2}\)\sin\(\frac{nl_{a,b}}{2}\) - n \cos\(\frac{nl_{a,b}}{2}\)\right].
\end{align}

Remarks: 
\begin{align}
    &\lambda_{n} \equiv \lambda_n^{[a,b]} \\ 
    &\lambda_{-n} = \lambda_n^* \\ 
    &\lambda_0 = 2 - (b-a) \cot \left(\frac{b-a}{2}\right) \\ 
    &\lambda_1 = \frac{e^{-i(a+b)/2}}{2} \frac{b-a+\sin(b-a)}{\sin((b-a)/2)} \\ 
    &\lambda_{-1} = \frac{e^{i(a+b)/2}}{2} \frac{b-a+\sin(b-a)}{\sin((b-a)/2)}
\end{align}

Now we compute the modular commutator: 
\begin{align}
    -\ii J_{[a,b],[c,d]}\equiv \langle [K_{[a,b]},K_{[c,d]}] \rangle &= \frac{c}{12}\sum_{n,m}\lambda_{n}^{[a,b]}\lambda_{m}^{[c,d]}(n^3-n)\delta_{n,-m} \\ 
    & = \frac{c}{12}\sum_n \lambda_n^{[a,b]}\lambda_{-n}^{[c,d]}(n^3-n) \\ 
    & = \frac{\ii c}{6}\sum_{n\geq 2}  \im(\lambda_n^{[a,b]}\lambda_{-n}^{[c,d]})(n^3-n) \\ 
    & = \frac{\ii c}{6}\sum_{n\geq 2} \Lambda_{[a,b],[c,d]}(n),
\end{align}
where 
\begin{align}
    \Lambda_{[a,b],[c,d]}(n) = \frac{4 \sin\( \frac{l_{a,b}}{2}\) \sin\(\frac{l_{c,d}}{2}\) \sin\(\frac{n(l_{a,c}+l_{b,d})}{2}\)}{n^3-n} F_{[a,b]}(n) F_{[c,d]}(n), 
\end{align}
with 
\begin{equation}
    F_{[x,y]}(n) = (n+1) \sin\( \frac{(n-1)l_{x,y}}{2} \) - (n-1) \sin\( \frac{(n+1)l_{x,y}}{2} \).
\end{equation}

The sum can be explicitly performed: 
\begin{align}
    \sum_{n\geq 2}\Lambda_{[a,b],[c,d]}(n) =(2\eta - 1) \cdot S_{[a,b],[c,d]},
\end{align}
where 
\begin{align}
    &\eta = \eta(a,c,b,d) = \frac{\sin(l_{a,c}/2)\sin(l_{b,d}/2)}{\sin(l_{a,b}/2)\sin(l_{c,d}/2)} \\ 
    &S_{[a,b],[c,d]} = g(l_{a,c}) - g(l_{b,c}) - g(l_{a,d})+g(l_{b,d}), 
\end{align}
with 
\begin{equation}
    g(x) = \frac{\ii }{2} \[ \ln(1-e^{\ii x}) - \ln(1-e^{-\ii x}) \].
\end{equation}

The result of the commutator is 
\begin{equation}\label{eq:J-vir}
    \boxed{ -\ii J_{[a,b],[c,d]} = \frac{\ii \pi c}{6} (2\eta(a,c,b,d)-1) \cdot S_{[a,b],[c,d]} }.   
\end{equation}

$g(x)$ has the following properties: 
\begin{itemize}
    \item This function is odd: 
    \begin{equation}
        g(-x) = -g(x). 
    \end{equation}
    \item This function is a periodic extension of $g_{[-\pi, \pi]}(x)$, defined as 
    \begin{equation}
        g_{[-\pi, \pi]}(x) = \sum_{n\geq 1} \frac{\sin(nx)}{n} = \left\{ \begin{aligned}
            &\frac{\pi-x}{2},&& \quad x\in(0,\pi] \\ 
            &0, && \quad x = 0 \\ 
            &-\frac{\pi-x}{2},&& \quad x\in[-\pi,0) \\ 
        \end{aligned} \right. . 
    \end{equation} 
\end{itemize}

Let's examine some cases: 
\begin{itemize}
    \item Usual case: $a<c<b<d$, then $l_{b,c}<0$ and 
    \begin{align}
        S_{[a,b],[c,d]} &= \frac{\pi - l_{a,c}}{2} + \frac{\pi - l_{c,b}}{2} - \frac{\pi- l_{a,d}}{2} + \frac{\pi-l_{b,d}}{2} \\ 
        & = \pi - \frac{1}{2}(c-a + b-c -d+a+d-b) \\ 
        & = \pi. 
    \end{align}
    Therefore we obtain 
    \begin{equation}
        -\ii J_{[a,b],[c,d]} = \frac{\ii \pi c}{6}(2\eta(a,c,b,d) - 1). 
    \end{equation}
    \item Left limit: $a<c<b<d$ with $c-a\to 0^+$. In this case, one just get the same result as above, with $\eta = 0$:  
    \begin{equation}
        J_{[a,b],[c,d]} = - \frac{\ii \pi c}{6}. 
    \end{equation}
    \item Left coincide: $a=c<b<d$. In this case 
    \begin{equation}
        g(l_{a,c}) = 0.
    \end{equation}
    Notice it's crucial that $g(0)=0$, because what actually appears is the sum. The correct order of limit is 
    \begin{equation}
        \lim_{N\to\infty}\lim_{x\to 0} \sum_{n = 1}^N \frac{\sin(nx)}{n} = 0. 
    \end{equation}
    One can explicitly check this step by step. First set $a=c$, then compute $\Lambda_{[a,b],[a,d]}(n)$. After doing the sum, one shall find there are only three $g(x)$: 
    \begin{align}
        S_{[a,b],[a,d]} =& g(l_{a,b}) - g(l_{a,d})+g(l_{b,d}) \\ 
        =& \frac{\pi}{2} +\frac{b-a-d+a+d-b}{2} \\ 
        =& \frac{\pi }{2}. 
    \end{align}
    Therefore 
    \begin{equation}\label{eq:J-ac}
        -\ii J_{[a,b],[a,d]} = -\frac{\ii \pi c}{6} \cdot {\frac{1}{2} }. 
    \end{equation}
    \item Right coincide: $a<c<b=d$. In this case 
    \begin{equation}
        g(l_{b,d}) = 0. 
    \end{equation}
    Then 
    \begin{align}
        S_{[a,b],[c,b]} &= \frac{\pi - l_{a,c}}{2} + \frac{\pi - l_{c,b}}{2} - \frac{\pi- l_{a,b}}{2} \\ 
        & = \frac{\pi}{2} - \frac{1}{2}(c-a + b-c -b+a) \\ 
        & = \frac{\pi}{2}.
    \end{align}
    Therefore 
    \begin{equation}\label{eq:J-bd}
        -\ii J_{[a,b],[c,b]} = -\frac{\ii \pi c}{6} \cdot { \frac{1}{2} }. 
    \end{equation}
\end{itemize} 

\subsection{Commutators from OPE}
This section is to compute the modular commutator from OPE of the stress-energy tensor. By doing this computation, one can see: 
\begin{itemize}
    \item the origin of $g(x)$ as a ``Fourier regulation'' of the indicator function. 
    \item the origin of $1/2$ as 
    \begin{equation}
        \Theta(0) = \frac{1}{2},
    \end{equation}
    that contributes to non-vanishing first derivative of the coolness function. 
\end{itemize}

The root of all the calculations in this section is from OPE of stress-energy tensor 
\begin{equation}
    T(x)T(y) = \frac{c/2}{(x-y)^4} + \frac{2 T(y)}{(x-y)^2} + \frac{\partial T(y)}{x-y} + \cdots.  
\end{equation}
Here the coordinate of space-time is 
\begin{equation}
    z = x+ \ii \tau. 
\end{equation}
If we require the OPE to be time-ordered, then 
\begin{equation}\begin{split} 
    [T(x),T(y)] &= \lim_{\epsilon \to 0^{+}} \Big[ T(x+\ii \epsilon)T(y-\ii\epsilon) - T(x - \ii\epsilon)T(y+\ii\epsilon) \Big] \\ 
     &=  \frac{\ii \pi c}{6}\partial^3_x \delta(x-y) + 4\pi \ii \partial_x \delta(x-y)T(y) - 2\pi \ii \partial T(y)\delta(x-y),
\end{split}
\end{equation}
where 
\begin{equation}
     \delta(x-y) = \frac{1}{\pi }\lim_{\epsilon\to 0^+} \frac{\epsilon}{(x-y)^2 + \epsilon^4}
\end{equation}
is used. 

Let's now compute $[K_{[x_1,x_2]},K_{[y_1,y_2]}]$. 
Let $\beta_{[a,b]}$ be the coolness function.  Using integration by parts many times,
\begin{align}
    &[K_{[x_1,x_2]},K_{[y_1,y_2]}] \\ 
    = &\int dx dy \beta_{[x_1,x_2]}(x) \beta_{[y_1,y_2]}(y) [T(x),T(y)] \\ 
    = &\int dx dy \beta_{[x_1,x_2]}(x) \beta_{[y_1,y_2]}(y) \Big[ \frac{\ii \pi c}{6}\partial^3_x \delta(x-y) + 4\pi \ii \partial_x \delta(x-y)T(y) - 2\pi \ii \partial T(y)\delta(x-y)\Big] \\  
    =& \int dx \Big[\(-\frac{\ii \pi c}{6}\) \beta'''_{[x_1,x_2]}(x)\beta_{[y_1,y_2]}(x) - 4\pi \ii \beta'_{[x_1,x_2]}(x)\beta_{[y_1,y_2]}(x) - 2\pi \ii \beta_{[x_1,x_2]}(x)\beta_{[y_1,y_2]}(x) \partial T(x)\Big] \\ 
    = & -\frac{\ii \pi c}{6}\int dx  \beta_{[x_1,x_2]}'''(x) \beta_{[y_1,y_2]}(x) - 2\pi \ii \int dx  \Big[ \beta_{[x_1,x_2]}'(x) \beta_{[y_1,y_2]}(x) - \beta_{[x_1,x_2]}(x) \beta_{[y_1,y_2]}'(x)\Big]T(x) ~.
\end{align}

Therefore 
\begin{align}
   -\ii J &\equiv \langle [K_{[x_1,x_2]},K_{[y_1,y_2]}] \rangle \\ 
    & = \underbrace{-\frac{\ii \pi c}{6}\int dx  \beta_{[x_1,x_2]}'''(x) \beta_{[y_1,y_2]}(x)}_{J_1}  \underbrace{- 2\pi \ii \int dx  \Big[ \beta_{[x_1,x_2]}'(x) \beta_{[y_1,y_2]}(x) - \beta_{[x_1,x_2]}(x) \beta_{[y_1,y_2]}'(x)\Big]\langle T(x) \rangle }_{J_2} \\ 
    & = J_1 + J_2.   
\end{align}

First evaluate $J_1$.  
\begin{align}
    \beta_{[x_1,x_2]} &= 2\Theta(x-x_1)\Theta(x_2-x)\frac{\sin((x-x_1)/2)\sin((x_2-x)/2)}{\sin((x_2-x_1)/2)}  \\ 
    & = \Theta(x-x_1)\Theta(x_2-x) \frac{\cos\((x_2+x_1-2x)/2\) - \cos\((b-a)/2\)}{\sin((x_2-x_1)/2)} \\
    \beta_{[x_1,x_2]}' &= \Theta(x-x_1)\Theta(x_2-x)\frac{\sin((x_2+x_1 - 2x)/2)}{\sin((x_2-x_1)/2)} \\ 
    \beta_{[x_1,x_2]}'' &= \delta(x-x_1) + \delta(x-x_2) -\Theta(x-x_1)\Theta(x_2-x)\frac{\cos((x_2+x_1 - 2x)/2)}{\sin((x_2-x_1)/2)} \\ 
    \beta_{[x_1,x_2]}''' &= \delta'(x-x_1) + \delta'(x-x_2) -\cot((x_2-x_1)/2)(\delta(x-x_1)-\delta(x-x_2)) + \beta'_{[x_1,x_2]}(x) 
\end{align}
Then 
\begin{equation}\begin{split} 
    J_1 = &-\frac{\ii \pi c}{6}\int dx  \beta_{[x_1,x_2]}'''(x) \beta_{[y_1,y_2]}(x)  \\ 
    = & \frac{\ii \pi c}{6} \Big[ \beta'_{[y_1,y_2]}(x_1) + \beta'_{[y_1,y_2]}(x_2) - \frac{\beta_{[y_1,y_2]}(x_2) - \beta_{[y_1,y_2]}(x_1)}{\tan((x_2-x_1)/2)} - \int dx \beta'_{[x_1,x_2]}(x)\beta_{[y_1,y_2]}(x)\Big]
\end{split}
\end{equation}

For the $J_2$ term, notice it doesn't vanish since the Casimir energy on the circle isn't zero! 
\begin{equation}
    \langle T(x)\rangle = - \(\frac{2\pi}{L}\)^2\frac{c}{24}
\end{equation}
Here we are working with a circle with radius 1, so the circumference is $L = 2\pi$. 
Therefore, 
\begin{align}
    J_2 &= 2 \cdot (-2\pi \ii) \cdot \frac{-c}{24} \int dx \beta'_{[x_1,x_2]}(x)\beta_{[y_1,y_2]}(x) \\ 
    & = \frac{\ii \pi c}{6} \int dx \beta'_{[x_1,x_2]}(x)\beta_{[y_1,y_2]}(x),
\end{align}
where integration by parts is used. The boundary term is $\beta_{[x_1,x_2]}(x)\beta_{[y_1,y_2]}(x)|_a^b$, which vanishes as $a,b$ take values from $\{x_1,x_2,y_1,y_2\}$. 
Notice this $J_2$ cancels the third term in $J_1$. 

Finally 
\begin{equation}
   -\ii J = \frac{\ii \pi c}{6} \Big[ \beta'_{[y_1,y_2]}(x_1) + \beta'_{[y_1,y_2]}(x_2) - \frac{\beta_{[y_1,y_2]}(x_2) - \beta_{[y_1,y_2]}(x_1)}{\tan((x_2-x_1)/2)} \Big].
\end{equation}
Using the same notation as Eq.~\eqref{eq:J-vir}, 
\begin{equation}
    \boxed{-\ii J_{[a,b],[c,d]} = \frac{\ii \pi c}{6} \Big[ \beta'_{[c,d]}(a) + \beta'_{[c,d]}(b) - \frac{\beta_{[c,d]}(b) - \beta_{[c,d]}(a)}{\tan((b-a)/2)} \Big] 
    }. 
\end{equation}
Notice 
\begin{align}
    \beta'_{[c,d]}(b) - \frac{\beta_{[c,d](b)}}{\tan(b-a)/2} = \[2\eta(a,c,b,d)-1\]\cdot \Theta(b-c)\Theta(d-b) \\ 
    \beta'_{[c,d]}(a) - \frac{\beta_{[c,d](a)}}{\tan(a-b)/2} = \[2\eta(b,c,a,d)-1\]\cdot \Theta(a-c)\Theta(d-a),
\end{align}
where 
\begin{equation}
    \eta(a,c,b,d) = \frac{\sin((c-a)/2)\sin((d-b)/2)}{\sin((b-a)/2)\sin((d-c)/2)} = 1-\eta(b,c,a,d). 
\end{equation}

Therefore the result is 
\begin{equation}
    \boxed{ 
        -\ii J_{[a,b],[c,d]} = \frac{\ii \pi c}{6}(2\eta(a,c,b,d)-1) \cdot \Big[\Theta_{[c,d]}(b)-\Theta_{[c,d]}(a)\Big]
    },
\end{equation}
where $\Theta_{[a,b]}$ is the $2\pi$ periodic extension of 
\begin{equation}
    \Theta_{[a,b]}(x) = \Theta(x-a)\Theta(b-x). 
\end{equation}

Now let us use it to compute several cases: 
\begin{itemize}
    \item Usual case, $a<c<b<d$: 
    \begin{align}
        -\ii J = \frac{\ii \pi c}{6} &= \frac{\ii \pi c}{6}\[2\eta(a,c,b,d)-1\] \cdot (1-0). \\ 
        &= \frac{\ii \pi c}{6}\[2\eta(a,c,b,d)-1\] . 
    \end{align}
    \item Left coincide, $a=c<b<d$: Using 
$        \Theta(0) = \frac{1}{2} $, 
    \begin{align}
        -\ii J_{[a,b],[a,d]} &= \frac{\ii \pi c}{6}\[2\eta(a,a,b,d)-1\] \cdot (1-1/2) \\ 
        &= -\frac{\ii \pi c}{6} \cdot {\frac{1}{2} },
    \end{align}
    which agrees with Eq.~\eqref{eq:J-ac} from using Virasoro algebra. 
    \item Right coincide, $a<c<b=d$: 
    \begin{align}
        -\ii J_{[a,b],[c,b]} &= \frac{\ii \pi c}{6}\[2\eta(a,c,b,b)-1\] \cdot (1/2 - 0) \\ 
        &= -\frac{\ii \pi c}{6} \cdot {\frac{1}{2} },
    \end{align}
    which agrees with Eq.~\eqref{eq:J-bd}. 
\end{itemize}

\subsection{Fourier regulation of the step function}

In the Virasoro computation, the $1/2$ factor comes from the $g(x)$ function, while in the OPE computation, it comes from $\Theta(0) = 1/2$. In fact, they are related by Fourier regulation of the step function, as we now show.

If $f(x)$ is piecewise continuous, then the Fourier series (or integral) converges to 
\begin{equation}
    \frac{1}{2}(f(x^{+}) + f(x^{-}))
\end{equation}
with $x^\pm \equiv x \pm \eps$ with $\eps \to 0$.

Consider the step function on a circle: 
\begin{equation}
    \Theta_{[0,\phi]}(x) = \left\{ \begin{aligned}
        &1, && \text{ if } x \in (0,\phi) \\ 
        &0, && \text{ if } x \in (\phi,2\pi) \\ 
    \end{aligned} \right.
\end{equation}
Then 
\begin{equation}
    \Theta_n = \int_0^{2\pi} \Theta(x) e^{-\ii n x} = \frac{1-e^{-\ii n\phi}}{\ii n}. 
\end{equation}
When $x=0$, we obtain 
\begin{equation}
    \Theta(0) = \frac{1}{2\pi}\sum_{n\in\mathbb{Z}} \Theta_n = \sum_{n\geq 1} \frac{\sin(n\phi)}{\pi n} + \frac{\phi}{2\pi},
\end{equation}
where $\phi/(2\pi)$ comes from $\Theta_0/(2\pi)$. 

The sum of the sinc function is exactly the $g(\phi)$ function defined above
\begin{align}
   \sum_{n\geq 1} \frac{\sin(n\phi)}{\pi n} =& \frac{1}{2 \ii}\Big[ - \ln\left( 1-e^{\ii\phi} \right) + \ln\left( 1 - e^{-\ii\phi}\right) \Big] \\ 
   =& g(\phi) 
\end{align}
where 
\begin{equation}
    \ln(1-z) = - \sum_{n\geq 1} \frac{z^n}{n}
\end{equation}
was used to perform the sum. 

Using this result, and recalling that  $g(\phi)$ is the $2\pi$ periodic extension of $g_{[-\pi,\pi]}(\phi)$: 
\begin{equation}
    g_{[-\pi, \pi]}(\phi) = \sum_{n\geq 1} \frac{\sin(nx)}{n} = \left\{ \begin{aligned}
        &\frac{\pi-\phi}{2},&& \quad \phi\in(0,\pi] \\ 
        &0, && \quad \phi = 0 \\ 
        &-\frac{\pi-\phi}{2},&& \quad \phi\in[-\pi,0) \\ 
    \end{aligned} \right. , 
\end{equation} 
we can obtain 
\begin{equation}
    \Theta(0) = \frac{1}{2}. 
\end{equation}

\section{Disentangler on a corner region}
\label{app:disentangler}
In this section, we will explicitly construct the disentangler.  

We start with a state $\ket{\Psi}$ that satisfies \Azero and \Aone. Let us consider the region in Fig.~\ref{fig:disentangler}. Because of the \Azero condition on region $XY$, one can conclude that 
\begin{align}
   \rho_{X \overline{XY}} = \rho_X \otimes \rho_{\overline{XY}},  
\end{align}
where $\rho_{\bullet}$ is the reduced density matrix from $\ket{\Psi}$. Then by Uhlmann theorem \cite{Uhlmann1976}, one can conclude that there exists a unitary $V_Y$ on $Y$ 
\begin{align}
     \ket{\Psi} = V_{Y} \ket{\Psi'}_{\overline{XY} Y_1} \otimes \ket{\Psi'}_{XY_2},
\end{align} 
where $Y = Y_1Y_2$ and $\ket{\Psi'}_{\overline{XY} Y_1} \otimes \ket{\Psi'}_{XY_2}$ is a purification of $\rho_X \otimes \rho_{\overline{XY}}$. Note that the decomposition $Y = Y_1Y_2$ is just an abstract decomposition of the Hilbert space $\CH_Y = \CH_{Y_1} \otimes \CH_{Y_2}$ and $Y_1,Y_2$ does not necessarily represent two regions (i.e. a collection of underlying lattice sites). For the state $\ket{\Psi'}_{XY_2}$, one can further apply some unitary $V_{XY_2}$ to trivialize the degrees of freedom in region $X$. That is, 
\begin{align}
     V_{XY_2}\ket{\Psi'}_{XY_2} = \ket{0}_X \otimes \ket{\Psi''}_{Y_2},
\end{align}
where $\ket{0}_X$ is a product state over the lattice Hilbert space $\CH_X = \otimes_{v \in X}\CH_v$. 

Therefore, we obtain the disentangler $U_{XY} = V^{\dagger}_{XY_2}V_Y$ such that we actually create a hole on $X$ and an actual boundary in $\partial X$. Explicitly, we obtain a unitary $U_{XY}$ such that 
\begin{align}
    U_{XY}\ket{\Psi} = \ket{\Psi'}_{\overline{X}} \otimes \ket{0}_{X},    
\end{align}  
where $\ket{\Psi'}_{\overline{X}}$ and $\ket{\Psi}$ have the same reduced density matrices on $\overline{XY}$. 

\section{Relation between $\fc$ and correlation length}
\label{app:ctot-xi}

In this section, we give an explicit derivation of Eq.~\eqref{eq:ctot-xi}. 

The goal is to consider $\fc$ in the limit $R \to \infty$, with the assumption that the leading order contribution to $\Delta(\mathbb{D}), I(\mathbb{D})$ is in Eq.~\eqref{eq:scaling-Delta-I} that is satisfied for a generic representative wavefunction of a gapped phase. 

Plug Eq.~\eqref{eq:scaling-Delta-I} into the definition of $\fc$ in Eq.~\eqref{eq:def-c-eta}, and with the expansion 
$e^{-6\Delta/\fc} \sim 1 - 6\Delta / \fc$ as $\Delta \sim e^{-R/\xi_{\Aone}}$ is exponentially small, we obtain 
\begin{align}
     1 = e^{-6\Delta/\fc} + e^{-6 I /\fc} = 1 - \frac{6\Delta}{\fc} + e^{-6 I / \fc}. 
\end{align}
Then $\fc =6 x_*$ where $x_*$ is the solution to the following equation
\begin{align}\label{eq:equ-x}
    x \ln x = R \( 2 \alpha - \frac{x}{\xi} \).
\end{align}
One can see from the figure below that, as $R \to \infty$, the solution is approaching to $2 \alpha \xi$ from below. 

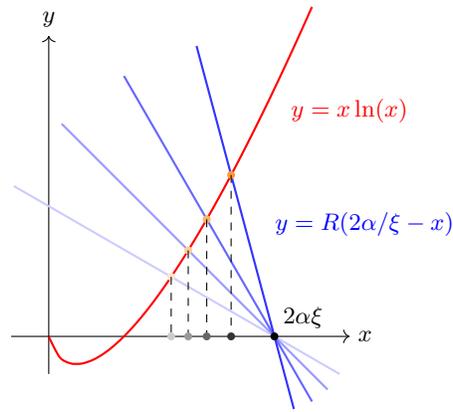
\begin{figure}[htb]
  \begin{tikzpicture}
  \usetikzlibrary{calc,intersections}

  \draw[->] (-0.5,0) -- (4,0) node[right] {$x$};
  \draw[->] (0,-0.5) -- (0,4) node[above] {$y$};

  \draw[domain=0.001:3.5, smooth, variable=\x, red, thick, name path=curve]
    plot ({\x}, {\x*ln(\x)});
  
  \foreach \i/\angle/\opacity in {1/30/20, 2/45/40, 3/60/60, 4/75/80} {
    \draw[blue!\opacity, thick, name path=line\i]
      ($(3,0) + (180-\angle:4)$) -- ($(3,0) + (-\angle:1)$);
    
    \fill[orange!\opacity, name intersections={of=curve and line\i}]
      (intersection-1) circle (1.5pt);
    \draw[dashed] let \p1=(intersection-1) in
      (intersection-1) -- (\x1, 0);
    \fill[black!\opacity] let \p1=(intersection-1) in
      (\x1, 0) circle (1.5pt);
  }

  \fill (3,0) circle (1.5pt) node[above right] {$2\alpha \xi$};
  \node[font = \small] at (4,3) {\color{red} $y = x \ln(x)$}; 
  \node[font = \small] at (4.2,1.5) {\color{blue} $y = R(2\alpha/\xi - x)$}; 
\end{tikzpicture}
\caption{Solution of Eq.~\eqref{eq:equ-x}. The blue and red lines are the right and left hand side of Eq.~\eqref{eq:equ-x}. The increasing of opacity shows the increasing of $R$ and the solutions are shown by the gray dots in the $x$ axis with increasing opacity. When $R\to \infty$, the solution approaches to $2 \alpha \xi$. }
\end{figure}

\bibliographystyle{ucsd}
\bibliography{ref} 
\end{document}